\documentclass{article}

\usepackage{amsmath}
\usepackage{amsfonts}
\usepackage{amssymb}
\usepackage{amsthm}
\usepackage{graphicx}
\usepackage{bm}
\usepackage{algorithm}
\usepackage{algpseudocode}
\usepackage{setspace}
\usepackage{blkarray}
\usepackage{subcaption}
\usepackage{booktabs}
\usepackage{multirow}
\usepackage{array}
\usepackage{url}
\usepackage{hyperref}
\usepackage{cleveref}
\usepackage{placeins}
\usepackage{printlen}

\usepackage[left = 2cm, right = 2cm, top = 2cm, bottom = 2cm]{geometry}

\numberwithin{figure}{section}
\numberwithin{table}{section}
\numberwithin{equation}{section}

\title{Overcoming Intensity Limits for Long-Distance Quantum Key Distribution}
\author{Ibrahim Almosallam\\
Ministry of Communication and Information Technology, Riyadh 12382, Saudi Arabia}
\date{}

\newtheorem{theorem}{Theorem}[section]
\newtheorem{lemma}[theorem]{Lemma}
\newtheorem{corollary}[theorem]{Corollary}

\begin{document}

\maketitle

\begin{abstract}
Quantum Key Distribution (QKD) enables the sharing of cryptographic keys secured by quantum mechanics. The BB84 protocol~\cite{bennett1984quantum, shor2000simple} assumed single-photon sources, but practical systems rely on weak coherent pulses vulnerable to photon-number-splitting (PNS) attacks~\cite{huttner1995quantum}. The Gottesman-Lo-Lütkenhaus-Preskill (GLLP) framework addressed these imperfections, deriving secure key rate bounds under limited PNS scenarios. The Decoy-state protocol~\cite{lo2005decoy} further improved performance by refining single-photon yield estimates, but still considered multi-photon states as insecure, thereby limiting intensities and constraining key rate and distance. More recently, finite-key security bounds for decoy-state QKD have been extended to address general attacks~\cite{lim2013concise}, ensuring security against adversaries capable of exploiting arbitrary strategies. In this work, we focus on a specific class of attacks, the generalized PNS attack, and demonstrate that higher pulse intensities can be securely used by employing Bayesian inference to estimate key parameters directly from observed data. By raising the pulse intensity to 10 photons, we achieve a 50-fold increase in key rate and a 62.2\% increase in operational range (about 200 km) compared to the decoy-state protocol. Furthermore, we accurately model after-pulsing using a Hidden Markov Model and reveal inaccuracies in decoy-state calculations that may produce erroneous key-rate estimates. While this methodology does not address all possible attacks, it provides a new approach to security proofs in QKD by shifting from worst-case assumption analysis to observation-dependent inference, advancing the reach and efficiency of discrete-variable QKD protocols.
\end{abstract}

\section{Introduction}
\label{sec:introduction}
QKD allows two parties, Alice and Bob, to share a secret key secured by quantum mechanics. The BB84 protocol ~\cite{bennett1984quantum} introduced QKD using single-photon polarization states and its unconditional security was proved in ~\cite{shor2000simple} via quantum error correction and entanglement purification, assuming ideal single-photon sources. True single-photon sources are hard to implement. Instead, most QKD systems use weak coherent pulses from attenuated lasers, which produce multi-photon states. These enable a photon-number-splitting (PNS) attack~\cite{huttner1995quantum}, where an eavesdropper, commonly referred to as Eve, intercepts some photons and forwards the rest, gaining key information undetected and compromising security.

The GLLP framework~\cite{gottesman2004security} addressed this by extending security proofs to imperfect devices and multi-photon pulses. It ensures secure QKD by treating multi-photon pulses as “tagged” and deriving a secure lower bound on the secret key rate that can be extracted from the sifted key:
\begin{align}
	K \geq \max\left((1 - \Delta) \left[1 - H^{}_{2}\left(\delta/\left(1 - \Delta\right)\right)\right] - H^{}_{2}(\delta),0\right),	\label{eq:GLLP}
\end{align}
where $\delta$ is the overall error rate, $\Delta$ the fraction of tagged photons, and $H^{}_{2}$ the binary Shannon entropy function.

Building on GLLP’s framework, the decoy-state protocol~\cite{hwang2003quantum,lo2004quantum,lo2005decoy,ma2005practical,wang2005beating} improved key rates and operational ranges by enhancing single-photon yield estimation through random pulse-intensity variations, and was first successfully demonstrated in~\cite{zhao2006experimental}. However, while it improves performance, the decoy-state approach relies on conservative assumptions for multi-photon pulses—assuming they are fully compromised—which forces conservative intensity choices, limits maximum transmission distances, and imposes high detector efficiency requirements.

While decoy-state protocols extend BB84, other approaches like SARG04 and Coherent One-Way (COW) also address PNS attacks. SARG04~\cite{scarani2004sarg04} uses the original four BB84 states but alters the sifting procedure, improving resilience to multi-photon attacks at the expense of reduced key rates. The COW protocol~\cite{barrett2005coherent} transmits weak coherent pulses in time-bin sequences and checks their coherence to detect eavesdropping. Although these methods enhance certain security aspects, they generally yield lower key rates and remain more susceptible to side-channel attacks compared to decoy-state protocols.

Despite its improved accuracy in estimating single-photon contributions, the decoy-state protocol still treats multi-photon pulses as insecure. By relying on additional intensity settings, it refines the statistical bounds on single-photon yields, but the need to suppress multi-photon probabilities enforces conservative intensity choices. This limits the maximum transmission distance and demands high detector efficiencies, posing significant experimental challenges for near-term implementations.

Advancements such as those presented in \cite{lim2013concise} have extended the security analysis by considering more general attack models, thereby further advancing the decoy-state framework. In summary, \cite{lim2013concise} broadens the scope of adversarial strategies beyond the conventional decoy approach.

Although previous works have accounted for various device imperfections, the correlations induced by after-pulsing are often overlooked. After-pulsing, where a detector emits spurious clicks following a real detection~\cite{restelli2013afterpulse}, induces correlations that violate the i.i.d. assumption. A common mitigation strategy is to increase the detector’s dead time, reducing after-pulse probabilities to negligible levels at the expense of lowering the achievable key rate. Even then, neglecting residual after-pulsing can still distort error estimates and weaken QKD security~\cite{yuan2007high}.

Twin-Field QKD (TF-QKD)~\cite{lucamarini2018overcoming} exploits quantum interference at a central node to surpass conventional rate-loss limits. By sending weak coherent pulses from both Alice and Bob to a measurement station, TF-QKD reduces reliance on trusted relays and enables secure key distribution over longer distances with minimal loss.

This work reframes QKD security from worst-case assumptions to a probabilistic inference model. Treating eavesdropper detection as an inference task, we use Bayesian methods to estimate key parameters like $\Delta$ from observed data. This tighter parameter estimation refines the GLLP formula, enabling higher pulse intensities and improving both key rate and operational range.

Instead of limiting source intensity to reduce Eve’s potential interference, we recast her presence as an inference problem. By modeling the QKD protocol as a probability distribution over system parameters and detection outcomes and applying Bayesian inference~\cite{gelman2013bayesian}, we directly estimate Eve’s interception rate $\Delta$ from observed data. Notably, our analysis focuses on the generalized PNS attack model—which captures a broader spectrum of adversarial strategies than the conventional PNS attack—rather than the more comprehensive attack model discussed in \cite{lim2013concise}. This approach eliminates overly conservative assumptions, extends PNS-attack analysis to arbitrary interception scenarios and channel conditions, and enables secure operation at higher pulse intensities when $\Delta$ is negligible.

To account for temporal correlations introduced by detector effects like after-pulsing and device variabilities, we incorporate a Hidden Markov Model (HMM)~\cite{rabiner1989tutorial}. This refinement yields more accurate error estimates and facilitates reliable key extraction.

Building on the GLLP framework without imposing intensity constraints, our approach supports greater operational distances, reduces reliance on high-efficiency detectors, and enhances key rates.

The remainder of this paper is organized as follows: \Cref{sec:probabilistic-modeling} presents the probabilistic model; \Cref{sec:bayesian-inference} outlines the Bayesian inference and HMM formulations; \Cref{sec:experimental-results} provides simulation results and comparisons to the decoy-state protocol; \Cref{sec:time_complexity} addresses time complexity; and \Cref{sec:conclusion} concludes with final remarks and potential future directions.

\section{Probabilistic Modeling}
\label{sec:probabilistic-modeling}

To better analyze and secure Quantum Key Distribution (QKD) systems, a deeper understanding of the underlying processes is essential. Realistic devices are subject to various sources of noise and inherent randomness, making it necessary to adopt a probabilistic modeling approach. In this section, we construct a comprehensive probabilistic framework that captures these complexities, accounting for factors such as detector efficiency, dark counts, beam-splitter misalignment, laser intensity, and fiber attenuation. By modeling these random processes, we derive the probabilities of different click events, laying the foundation for a more rigorous security analysis of QKD implementations.

We begin by modeling each component of the QKD setup separately to understand how the distribution of photons evolves as they travel from Alice to Bob. This modular approach enables us to systematically build the complete probability distribution, ensuring that we incorporate all relevant factors that influence the detection process.

Next, we extend our analysis to incorporate eavesdropping scenarios, specifically generalizing the photon-number-splitting (PNS) attack. By introducing Eve's parameters into our model, we calculate how the probabilities of click events change under various strategies she might employ. This generalization enables us to quantify the security of the QKD system under more realistic and flexible attack models.

Following this, we address the temporal dependencies introduced by after-pulsing. We develop a Hidden Markov Model (HMM) to capture these dependencies and refine our probability estimates accordingly. The HMM framework allows us to model correlations between detection events, which are crucial for accurately assessing the impact of after-pulsing on error rates and key generation.

Finally, we derive the error and gain probabilities from both the independent and identically distributed (i.i.d.) model and the HMM-based model. These probabilities are essential for key rate calculations and form the basis for evaluating the security and performance of the QKD protocol under realistic conditions. By the end of this section, we will have constructed a detailed and adaptable probabilistic model that serves as the foundation for our Bayesian inference and subsequent analysis.

\subsection{QKD Components}
Accurately modeling the QKD system requires a detailed understanding of the components involved, including detectors, beam splitters, and the interaction with quantum channels. In this section, we construct a probabilistic framework to model the behavior of a QKD system, incorporating various noise sources and potential attack vectors.

\subsubsection{Detectors}
The detection process in QKD systems can be characterized by three key probabilities:

\begin{table}[h!]
\centering
\begin{tabular}{|l|c|l|}
\hline
\textbf{Characteristic}         & \textbf{Notation} & \textbf{Description} \\ \hline
\textbf{After-pulse}       & $p^{}_{a}$        & Probability of a click given a click in the previous detection window. \\ 
\textbf{Efficiency}            & $p^{}_{c}$        & Probability of a click given a single photon. \\ 
\textbf{Dark Count}        & $p^{}_{d}$        & Probability of a click in the absence of a photon. \\ 
\textbf{Misalignment}      & $p^{}_{e}$        & Probability of detection error due to misalignment. \\ \hline
\end{tabular}
\caption{Key detector characteristics}
\label{tab:characteristics}
\end{table}

Throughout the paper, we will also use the notation $q=1-p$ for several parameters, for example $q^{}_{c} = 1-p^{}_{c}$, $q^{}_{d} = 1-p^{}_{d}$, and so on.

For detectors $D^{}_{0}$ and $D^{}_{1}$, we denote their respective click probabilities as $p^{}_{c^{}_{0}}$ and $p^{}_{c^{}_{1}}$, and similarly for dark count and after-pulse probabilities. When $n$ photons reach a detector and each has an independent detection probability $p^{}_{c}$, the detected signal photon count, $s$, is binomially distributed: $s \sim \text{Binomial}(n, p^{}_{c})$. Similarly, the occurrence of a dark count follows a Bernoulli distribution with probability $p^{}_{d}$: $d\sim \text{Bernoulli}(p^{}_{d})$. 

These distributions assume independence and fixed probabilities. To overcome this, we will use a Bayesian inference framework, modeling parameters as random variables with priors and updating posteriors based on observed data to incorporate uncertainties. Additionally, a hidden Markov model (HMM) will address dependencies between successive detection events, relaxing the independence assumption.

The probability of a click event at the detector, $D$, can be considered as the combined probability of detecting at least one photon or registering a dark count, see \Cref{fig:detector}. Given these considerations, we present the following lemma, which formalizes the combined click probability at a detector:
\begin{figure}[h]
    \centering
    \includegraphics[scale=0.25]{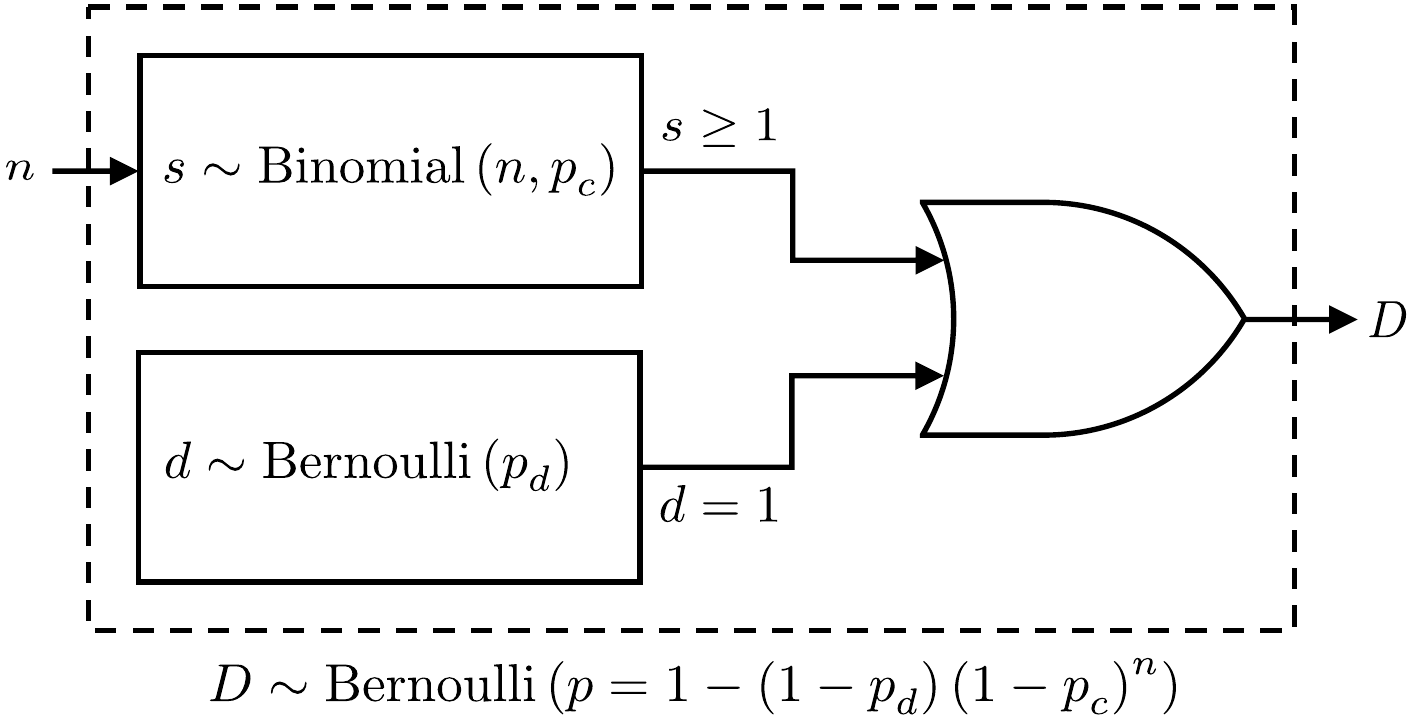}
    \caption{Schematic representation of the detection process in a single-photon detector considering $n$ incoming photons with a detector efficiency of $p^{}_{c}$ and dark count probability $p^{}_{d}$. The detection event $D$ will follow a Bernoulli distribution as per \Cref{lem:detector}}
    \label{fig:detector}
\end{figure}
\begin{lemma}
\label{lem:detector}
Given $n$ photons incident on a detector, each detected independently with probability $p^{}_{c}$, and a dark count probability $p^{}_{d}$, the detection event $D$, detecting at least one photon or registering a dark count, follows a Bernoulli distribution:

\begin{align}
	P(D\mid n, p^{}_{c}, p^{}_{d}) = \text{Bernoulli}(D\mid p = 1 - q^{}_{d} q^{n}_{c}),
\end{align}
where $q^{}_{c}=1-p^{}_{c}$ and $q^{}_{d}=1-p^{}_{d}$
\end{lemma}

\begin{proof}
The detection event $D$ occurs if at least one signal photon is detected or a dark count occurs. Since these are independent events, the probability that neither occurs is:

\[
P(D = 0 \mid n, p^{}_{c}, p^{}_{d}) = P(d = 0) P(s = 0) = q^{}_{d} q^{n}_{c}.
\]

Therefore, the probability of a detection event is:

\[
P(D = 1 \mid n, p^{}_{c}, p^{}_{d}) = 1 - P(D = 0\mid n, p^{}_{c}, p^{}_{d}) = 1 - q^{}_{d} q^{n}_{c}.
\]

\end{proof}

\subsubsection{Fiber-Detector}

Optical fibers affect the probability of photons arriving at the detector by attenuating their intensity. The attenuation rate $\alpha$ (dB/km) determines the probability $p^{}_{f}$ of a photon successfully traversing a distance $d$, given by:

\begin{align}
	p^{}_{f} = 10^{-\alpha \frac{d}{10}}.
\end{align}

Each photon independently emerges with probability $p^{}_{f}$, making the number of photons $k$ exiting the fiber, given $n$ entered, follow a binomial distribution: $k \sim \text{Binomial}(n, p^{}_{f})$.

Given this model, the probability of a detection event at a detector with efficiency $p^{}_{c}$ and dark count probability $p^{}_{d}$, after photons pass through a fiber with loss $p^{}_{f}$, is derived using the binomial distribution of photons emerging from the fiber and their detection probabilities. This is formally presented in \Cref{lem:fiber_detector} (see also \Cref{fig:fiber_detector}).

\begin{figure}[h]
    \centering
    \includegraphics[scale=0.25]{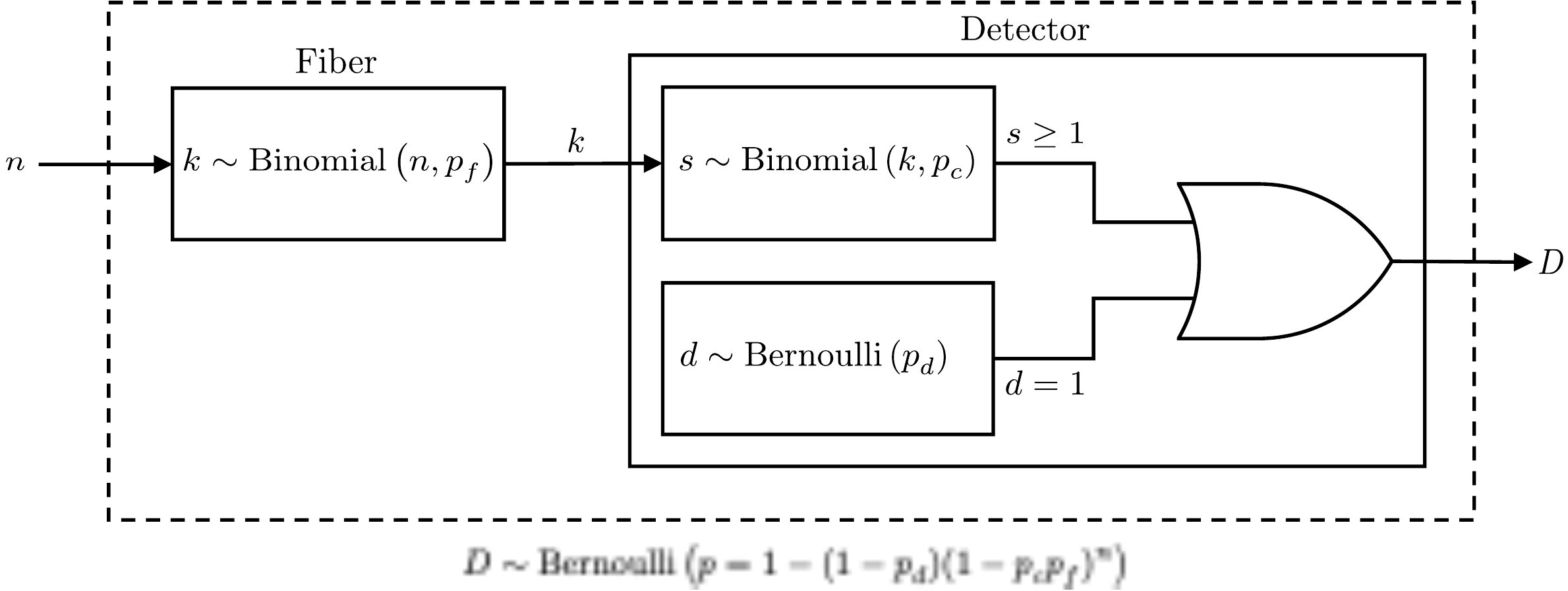}
    \caption{Illustration of the detection process considering both photon transmission through a fiber and the detector's response. The detection event $D$ will follow a Bernoulli distribution as per \Cref{lem:fiber_detector}}
    \label{fig:fiber_detector}
\end{figure}

\begin{lemma}
\label{lem:fiber_detector}
When $n$ photons pass through a fiber before reaching a detector, the probability of a detection event at the detector, given the fiber transmission probability $p^{}_{f}$, detector efficiency $p^{}_{c}$, and dark count probability $p^{}_{d}$, is equivalent to the probability of a detection event for $n$ photons directly incident on a detector with its efficiency scaled by the fiber transmission probability. Specifically, this means:

\begin{align}
	P(D\mid n, p^{}_{c}, p^{}_{d}, p^{}_{f}) = P(D \mid n, p^{}_{c}\cdot p^{}_{f}, p^{}_{d}),
\end{align}
where the right-hand side is the detection probability derived in \Cref{lem:detector}.
\end{lemma}

\begin{proof}
Since each photon has a probability $p^{}_{f}$ of passing through the fiber and reaching the detector, the number of photons $k$ that reach the detector is a binomial random variable:

\[
k \sim \text{Binomial}(n, p^{}_{f}).
\]

The probability of a click event at the detector is computed by marginalizing over all possible values of $k$, summing the probabilities of a click given $k$, weighted by the probability that exactly $k$ photons reach the detector:

\[
P(D = 1 \mid n, p^{}_{c}, p^{}_{d}, p^{}_{f}) = \sum^{n}_{k=0} P(D = 1 \mid k, p^{}_{c}, p^{}_{d})   P(k \mid n, p^{}_{f}),
\]
where:

\begin{itemize}
    \item $P(D = 1 \mid k, p^{}_{c}, p^{}_{d})$ is the probability of a click given $k$ photons reach the detector, as previously derived in \Cref{lem:detector}.
    \item $P(k \mid n, p^{}_{f})$ is the probability that exactly $k$ photons pass through the fiber, modeled as $\text{Binomial}(k \mid n, p^{}_{f})$.
\end{itemize}

Substituting these into the sum:

\[
P(D = 1 \mid n, p^{}_{c}, p^{}_{d}, p^{}_{f}) = \sum^{n}_{k=0} \left(1 - q^{}_{d} q^{k}_{c}\right) \binom{n}{k} p^{k}_{f} (1 - p^{}_{f})^{n-k}.
\]

Expanding and simplifying:

\[
P(D = 1 \mid n, p^{}_{c}, p^{}_{d}, p^{}_{f}) = 1 - q^{}_{d} \sum^{n}_{k=0} \binom{n}{k} \left(q^{}_{c} p^{}_{f}\right)^k \left(1 - p^{}_{f}\right)^{n-k}.
\]

Applying the binomial theorem~\cite{bishop2006pattern}:

\[
\sum^{n}_{k=0} \binom{n}{k} x^k y^{n-k} = (x + y)^n,
\]
where $x = q^{}_{c} p^{}_{f}$ and $y = 1 - p^{}_{f}$, we get:

\begin{align}
P(D = 1 \mid n, p^{}_{c}, p^{}_{d}, p^{}_{f}) &= 1 - q^{}_{d} \left(q^{}_{c} p^{}_{f} + 1 - p^{}_{f}\right)^n, \nonumber \\
 &= 1 - q^{}_{d} \left(1 - p^{}_{c} p^{}_{f}\right)^n.\nonumber
\end{align}
This matches the form in \Cref{lem:detector}, where the efficiency is effectively $p^{}_{c} p^{}_{f}$.
\end{proof}

\subsubsection{Fiber-Beam Splitter}

If, after passing through the fiber, the photon stream encounters a beam splitter, each photon is independently directed to one of two paths, labeled $0$ and $1$, with probabilities $p^{}_{0}$ and $p^{}_{1} = 1 - p^{}_{0}$, respectively. Then, the number of photons taking the $0$ path follows $m^{}_{0} \sim \text{Binomial}(m, p^{}_{0})$, while those taking the $1$ path follow $m^{}_{1} = m - m^{}_{0} \sim \text{Binomial}(m, p^{}_{1})$.

The objective is to derive the distribution of photons $m^{}_{i}$ directed to detector $D^{}_{i}$ ($i \in \{0, 1\}$) given $n$ photons initially entering the fiber, combining the effects of fiber attenuation and the beam splitter. This is formally presented in \Cref{lem:fiber_beamsplitter} (see also \Cref{fig:fiber_beamsplitter}).

\begin{figure}[h]
    \centering
    \includegraphics[scale=0.25]{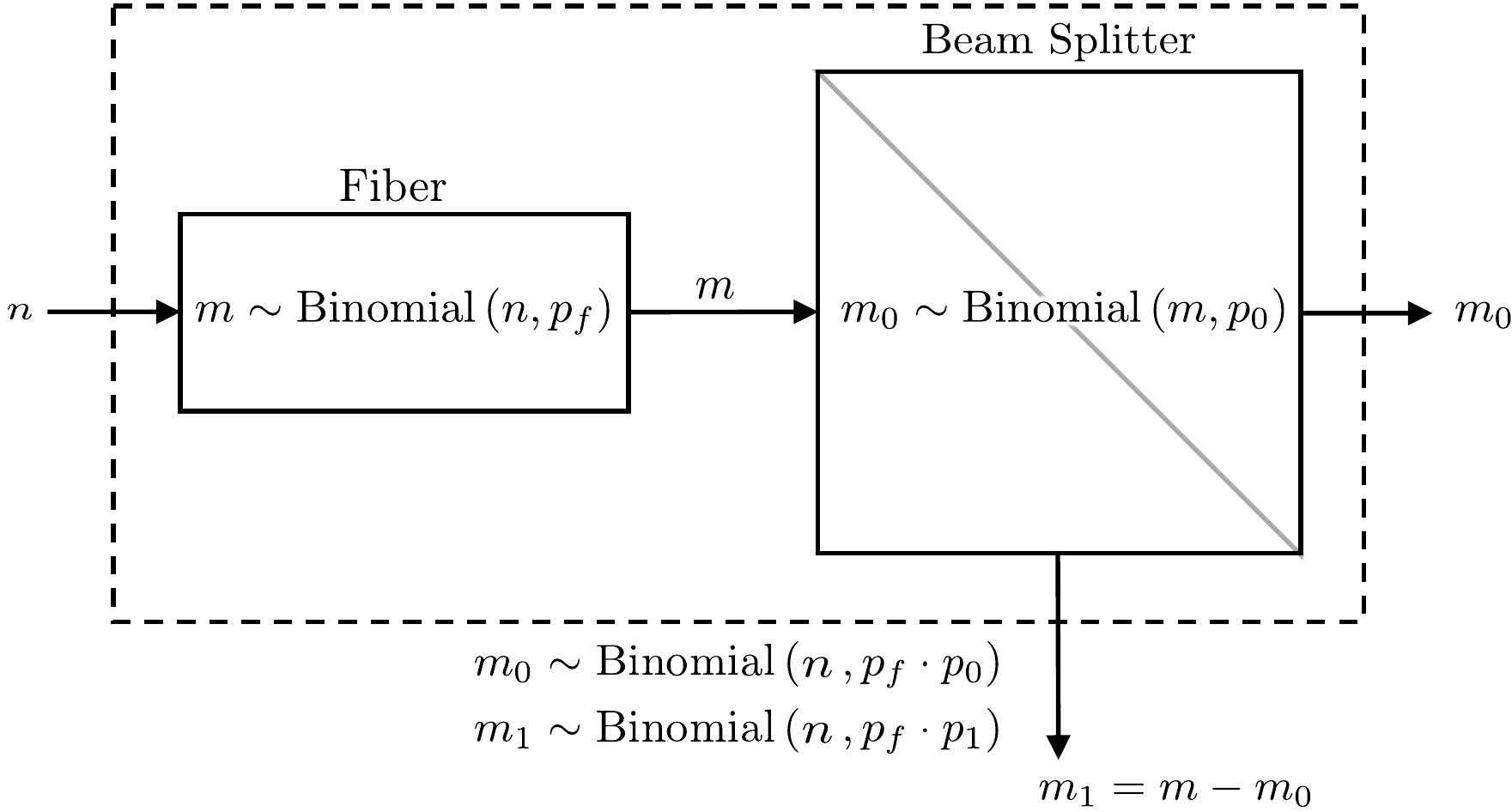}
    \caption{A diagram representing the probabilistic model for a fiber followed by a beam splitter. The number of photons $m^{}_{i}$ directed towards detector $D^{}_{i}$ will follow a binomial distribution as per \Cref{lem:fiber_beamsplitter}.}
    \label{fig:fiber_beamsplitter}
\end{figure}

\begin{lemma}
\label{lem:fiber_beamsplitter}	
If $n$ photons go through a fiber with a transmission probability $p^{}_{f}$ and then encounter a beam splitter that directs each photon to path $i$ with probability $p^{}_{i}$, the number of photons $m^{}_{i}$ directed towards detector $i$ follows a binomial distribution with a combined success probability of $p^{}_{f}   p^{}_{i}$. Formally, this can be expressed as:

\begin{align}
	m^{}_{i} \sim \text{Binomial}(n, p^{}_{f} p^{}_{i}).
\end{align}

\end{lemma}

\begin{proof}
$m^{}_{i}$ results from two successive binomial processes: first, the probability of $m$ successes in $n$ trials with probability $p^{}_{f}$, and then the probability of $m^{}_{i}$ successes in $m$ trials with probability $p^{}_{i}$. To find the distribution of $m^{}_{i}$, we marginalize over the latent variable $m$:

\[
P(m^{}_{i} \mid n, p^{}_{f}, p^{}_{i}) = \sum^{n}_{m=m^{}_{i}} P(m^{}_{i} \mid m, p^{}_{i})   P(m \mid n, p^{}_{f}).
\]

The index starts at $m^{}_{i}$, as $m^{}_{i} > m$ is not possible. Since $P(m^{}_{i} \mid m, p^{}_{i})$ and $P(m \mid n, p^{}_{f})$ are both binomial distributions, we substitute them as follows:

\begin{align}
P(m^{}_{i} \mid n, p^{}_{f}, p^{}_{i}) &= \sum^{n}_{m=m^{}_{i}} \text{Binomial}(m^{}_{i} \mid m, p^{}_{i})   \text{Binomial}(m \mid n, p^{}_{f}) \nonumber\\
&= \sum^{n}_{m=m^{}_{i}} \binom{m}{m^{}_{i}} p^{m^{}_{i}}_{i} (1-p^{}_{i})^{m-m^{}_{i}} \binom{n}{m} p^{m}_{f} (1-p^{}_{f})^{n-m} \nonumber\\
&= \sum^{n}_{m=m^{}_{i}} \binom{m}{m^{}_{i}} p^{m^{}_{i}}_{i} q^{m-m^{}_{i}}_{i} \binom{n}{m} p^{m}_{f} q^{n-m}_{f} \nonumber\\
&= \sum^{n}_{m=m^{}_{i}} \frac{m!}{m^{}_{i}!(m-m^{}_{i})!} p^{m^{}_{i}}_{i} q^{m-m^{}_{i}}_{i} \frac{n!}{m!(n-m)!} p^{m}_{f} q^{n-m}_{f} \nonumber\\
&= \frac{n!}{m^{}_{i}!} p^{m^{}_{i}}_{i} \sum^{n}_{m=m^{}_{i}} \frac{q^{m-m^{}_{i}}_{i}p^{m}_{f} q^{n-m}_{f} }{(m-m^{}_{i})!(n-m)!}. \nonumber
\end{align}

Letting $c = m - m^{}_{i}$ and $C = n - m^{}_{i}$, and re-indexing:

\begin{align}
P(m^{}_{i} \mid n, p^{}_{f}, p^{}_{i}) &= \frac{n!}{m^{}_{i}!} p^{m^{}_{i}}_{i} \sum^{n}_{c=0} \frac{q^{c}_{i}p^{c+m^{}_{i}}_{f} q^{C-c}_{f} }{c!(C-c)!}, \nonumber\\
 &= \frac{n!}{m^{}_{i}!} p^{m^{}_{i}}_{i}p^{m^{}_{i}}_{f} \sum^{C}_{c=0} \frac{q^{c}_{i}p^{c}_{f} q^{C-c}_{f}}{c!(C-c)!}, \nonumber\\
 &= \frac{n!}{m^{}_{i}!C!} (p^{}_{f} p^{}_{i})^{m^{}_{i}} \sum^{C}_{c=0} \frac{C!}{c!(C-c)!}(q^{}_{i}p^{}_{f})^c q^{C-c}_{f}, \nonumber\\
&= \frac{n!}{m^{}_{i}!C!} (p^{}_{f} p^{}_{i})^{m^{}_{i}} \sum^{C}_{c=0} \binom{C}{c} (q^{}_{i}p^{}_{f})^c q^{C-c}_{f}. \nonumber
\end{align}

Recognizing the binomial theorem~\cite{bishop2006pattern}:

\[
P(m^{}_{i} \mid n, p^{}_{f}, p^{}_{i}) = \frac{n!}{m^{}_{i}!C!} (p^{}_{f} p^{}_{i})^{m^{}_{i}} (q^{}_{i}p^{}_{f} + q^{}_{f})^{C}.
\]

Simplifying and substituting back:

\[
P(m^{}_{i} \mid n, p^{}_{f}, p^{}_{i}) = \binom{n}{m^{}_{i}} (p^{}_{f} p^{}_{i})^{m^{}_{i}} (1 - p^{}_{f} p^{}_{i})^{n-m^{}_{i}}.
\]

This has the form of the binomial distribution with number of trials $m^{}_{i}$ and success probability $p^{}_{f} p^{}_{i}$.
\end{proof}

\subsubsection{Pair of Detectors}
\label{sec:pair_of_detectors}

In practical BB84 implementations, two detectors are used to distinguish between the detection of a true 0 bit and the absence of photon arrival. Therefore, to fully characterize the detection system, we compute the four joint probabilities of detector outcomes:

\begin{itemize}
    \item $P^{00}$: The probability of neither detectors clicking.
    \item $P^{01}$: The probability of only $D^{}_{0}$ clicking.
    \item $P^{10}$: The probability of only $D^{}_{1}$ clicking.
    \item $P^{11}$: The probability of both detectors clicking.
\end{itemize}

If we have access to the marginal probabilities for each detector:

\begin{align}
	P(D^{}_{0} = 1) &= P^{*1} = P^{01} + P^{11},\\
	P(D^{}_{1} = 1) &= P^{1*} = P^{10} + P^{11},
\end{align}
as well as the union probability: 
\begin{align}
		P(D^{}_{0} \vee D^{}_{1} = 1) = P^{\vee} = P^{01} + P^{10} + P^{11},
\end{align}
we can use these to infer the joint probabilities as follows:

\begin{align}
P^{00} &= 1 - P^{\vee}, \\
P^{01} &= P^{\vee} - P^{1*}, \\
P^{10} &= P^{\vee} - P^{*1}, \\
P^{11} &= P^{*1} + P^{1*} - P^{\vee}.
\end{align}

The marginals are straight forward results of previously stated lemmas. From \Cref{lem:fiber_beamsplitter} we showed that the effect of a fiber with loss $p^{}_{f}$ followed by a beam splitter path probability $p^{}_{i}$, is essentially equivalent to a single fiber with a loss of $p^{}_{f}p^{}_{i}$. Additionally, in \Cref{lem:fiber_detector}, we showed that the effect of a fiber with loss $p^{}_{f}$ connected to a detector with efficiency $p^{}_{c}$, is equivalent to a single detector with an efficiency of $p^{}_{f}p^{}_{c}$. Therefore, we can combine these results and conclude that a fiber connected to a beam splitter that is connected to a detector, is equivalent to a single detector with an efficiency of $p^{}_{f}p^{}_{i}p^{}_{c}$, or more formally:
\begin{align}
	P(D^{}_{i}=1\mid n,p^{}_{c},p^{}_{d},p^{}_{f},p^{}_{i}) = P(D=1\mid n,p^{}_{f}  p^{}_{i}  p^{}_{c},p^{}_{d})
\end{align}
In the following Lemma, we derive the union probability in order to derive the joint probabilities. A schematic representation of the two-detector system is presented in \Cref{fig:pair_detectors}.
\begin{figure}[h]
    \centering
    \includegraphics[scale=0.25]{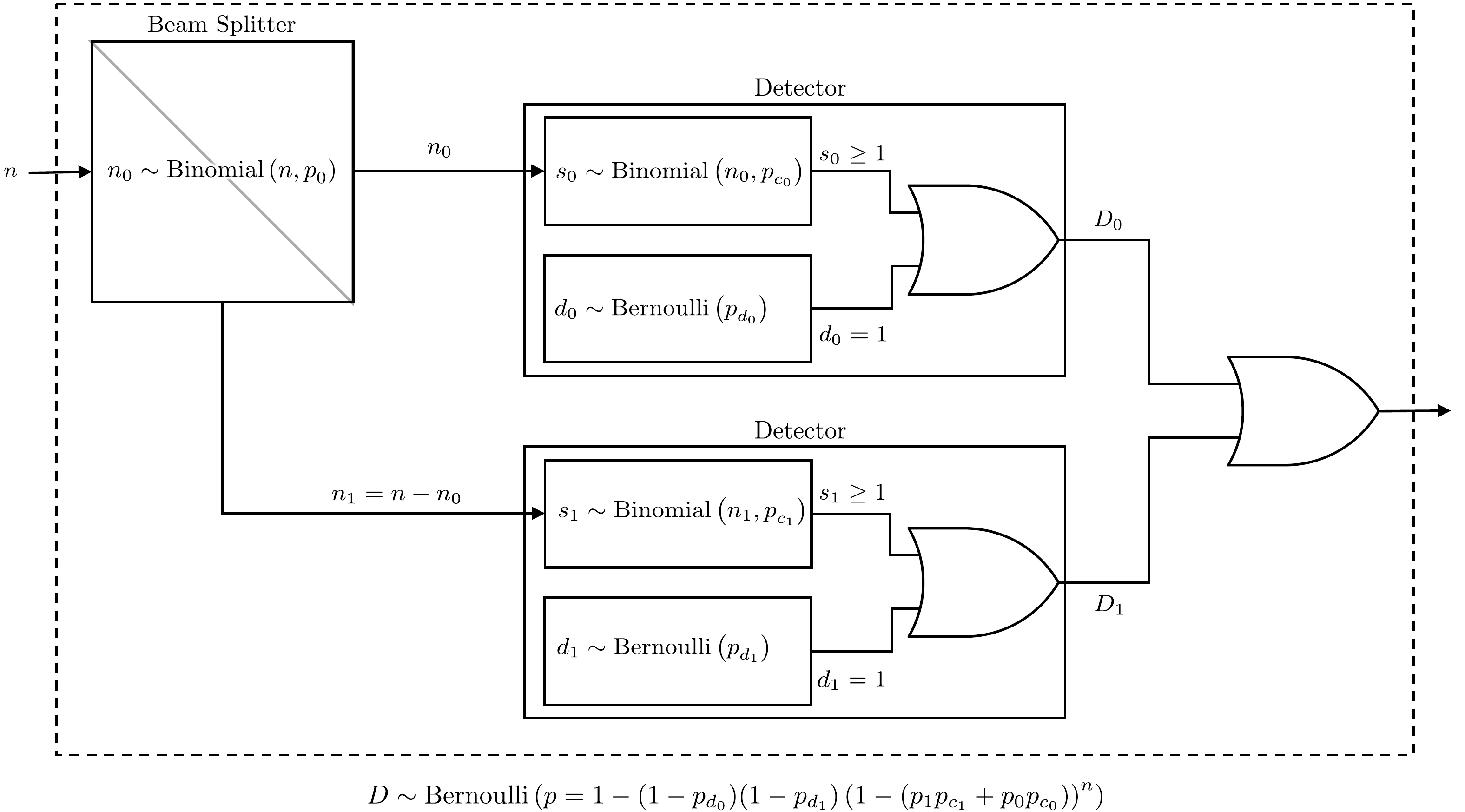}
    \caption{A diagram illustrating the detection setup for two detectors after a beam-splitter. The detection event follows a Bernoulli distribution as per \Cref{lem:pair_detectos}.}
    \label{fig:pair_detectors}
\end{figure}
\begin{lemma}
\label{lem:pair_detectos}
If $n$ photons encounter a beam splitter that independently directs each photon to one of two paths, $i \in \{0, 1\}$, with probabilities $p^{}_{0}$ and $p^{}_{1} = 1 - p^{}_{0}$, and these paths lead to two detectors $D^{}_{0}$ or $D^{}_{1}$ with click probabilities $p^{}_{c^{}_{0}}$ and $p^{}_{c^{}_{1}}$, and dark count probabilities $p^{}_{d^{}_{0}}$ and $p^{}_{d^{}_{1}}$, respectively, then the probability of at least one of the detectors clicking is equivalent to the click probability of a single pseudo detector with $n$ photons directly incident on it, where the click probability is $p^{\vee}_{c}$ and the dark count probability is $p^{\vee}_{d}$. Formally, this probability is given by:

\begin{align}
	P(D^{}_{0} \vee D^{}_{1} = 1 \mid n, p^{}_{c}, p^{}_{d}, p^{}_{0}) = P(D = 1 \mid n, p^{\vee}_{c}, p^{\vee}_{d}),
\end{align}
where

\begin{align}
	p^{\vee}_{c} = p^{}_{1}p^{}_{c^{}_{1}} + p^{}_{0}p^{}_{c^{}_{0}}, \quad p^{\vee}_{d} = 1 - q^{}_{d^{}_{0}} q^{}_{d^{}_{1}},
\end{align}
and $P(D = 1 \mid n, p^{}_{c}, p^{}_{d})$ is as defined in \Cref{lem:detector}.
\end{lemma}

\begin{proof}
To compute the probability that at least one detector, either $D^{}_{0}$ or $D^{}_{1}$, registers a click, we sum over all possible numbers of photons $k$ that could be passed to $D^{}_{0}$ and marginalize over this variable:
\[
P(D^{}_{0} \vee D^{}_{1} = 1 \mid n, p^{}_{c}, p^{}_{d}, p^{}_{0}) = \sum^{n}_{k=0} P(D^{}_{0} \vee D^{}_{1} = 1 \mid k, p^{}_{c}, p^{}_{d})   \text{Binomial}(k \mid n, p^{}_{0}).
\]

Since after the beam-split, the click events at $D^{}_{0}$ and $D^{}_{1}$ become independent, the probability that neither detector clicks is the product of their independent non-click probabilities. Therefore, the probability that at least one detector clicks (the union probability) is:

\[
P(D^{}_{0} \vee D^{}_{1} = 1 \mid k, p^{}_{c}, p^{}_{d}) = 1 - P(D^{}_{0} = 0 \mid k, p^{}_{c^{}_{0}}, p^{}_{d^{}_{0}})   P(D^{}_{1} = 0 \mid n-k, p^{}_{c^{}_{1}}, p^{}_{d^{}_{1}}).
\]

We substitute the non-click probabilities:

\[
P(D^{}_{0} = 0 \mid k, p^{}_{c^{}_{0}}, p^{}_{d^{}_{0}}) = q^{}_{d^{}_{0}} q^{k}_{c^{}_{0}}, \quad P(D^{}_{1} = 0 \mid n-k, p^{}_{c^{}_{1}}, p^{}_{d^{}_{1}}) = q^{}_{d^{}_{1}} q^{n-k}_{c^{}_{1}}.
\]

Thus, the probability of a click in either detector is:

\[
P(D^{}_{0} \vee D^{}_{1} = 1 \mid n, p^{}_{c}, p^{}_{d}, p^{}_{0}) = \sum^{n}_{k=0} \left(1 - q^{}_{d^{}_{0}} q^{}_{d^{}_{1}} q^{k}_{c^{}_{0}} q^{n-k}_{c^{}_{1}}\right) \text{Binomial}(k \mid n, p^{}_{0}).
\]

This can be rewritten by expanding and simplifying as:

\[
P(D^{}_{0} \vee D^{}_{1} = 1 \mid n, p^{}_{c}, p^{}_{d}, p^{}_{0}) = 1 - q^{}_{d^{}_{0}}   q^{}_{d^{}_{1}}   \sum^{n}_{k=0} \binom{n}{k} (p^{}_{0}   q^{k}_{c^{}_{0}}) (q^{}_{0}   q^{}_{c^{}_{1}})^{n-k}.
\]

Recognizing the binomial theorem~\cite{bishop2006pattern}:

\[
P(D^{}_{0} \vee D^{}_{1} = 1 \mid n, p^{}_{c}, p^{}_{d}, p^{}_{0}) = 1 - q^{}_{d^{}_{0}}   q^{}_{d^{}_{1}}   (p^{}_{0}   q^{}_{c^{}_{0}} + q^{}_{0}   q^{}_{c^{}_{1}})^n.
\]

Further simplifying:

\[
P(D^{}_{0} \vee D^{}_{1} = 1 \mid n, p^{}_{c}, p^{}_{d}, p^{}_{0}) = 1 - \underset{1-p^{\vee}_{d}}{\underbrace{q^{}_{d^{}_{0}}   q^{}_{d^{}_{1}}}}   (1 - (\underset{p^{\vee}_{c}}{\underbrace{p^{}_{1}p^{}_{c^{}_{1}}+p^{}_{0}p^{}_{c^{}_{0}}}})))^n.
\]

This has the same form as the probability in \Cref{lem:detector}, but with the dark count probability defined as $1 - q^{}_{d^{}_{0}} q^{}_{d^{}_{1}}$, and the click probability defined as $p^{}_{1}p^{}_{c^{}_{1}} + p^{}_{0}p^{}_{c^{}_{0}}$.
\end{proof}

\subsubsection{Laser-Fiber}

In quantum key distribution (QKD) systems, when using a weak coherent laser source with randomized phases, the number of photons emitted by the source follows a Poisson distribution~\cite{ma2005practical}:
\[
n \sim \text{Poisson}(\lambda),
\]
where $\lambda$ is the average photon number per pulse, determined by the intensity of the laser source. Since each of the $n$ photons emitted from the laser has an independent probability $p^{}_{f}$ of passing through the optical fiber, the number of photons $m$ that successfully exit the fiber follows a binomial distribution with parameters $n$ and $p^{}_{f}$.

\begin{figure}[h]
    \centering
    \includegraphics[scale=0.25]{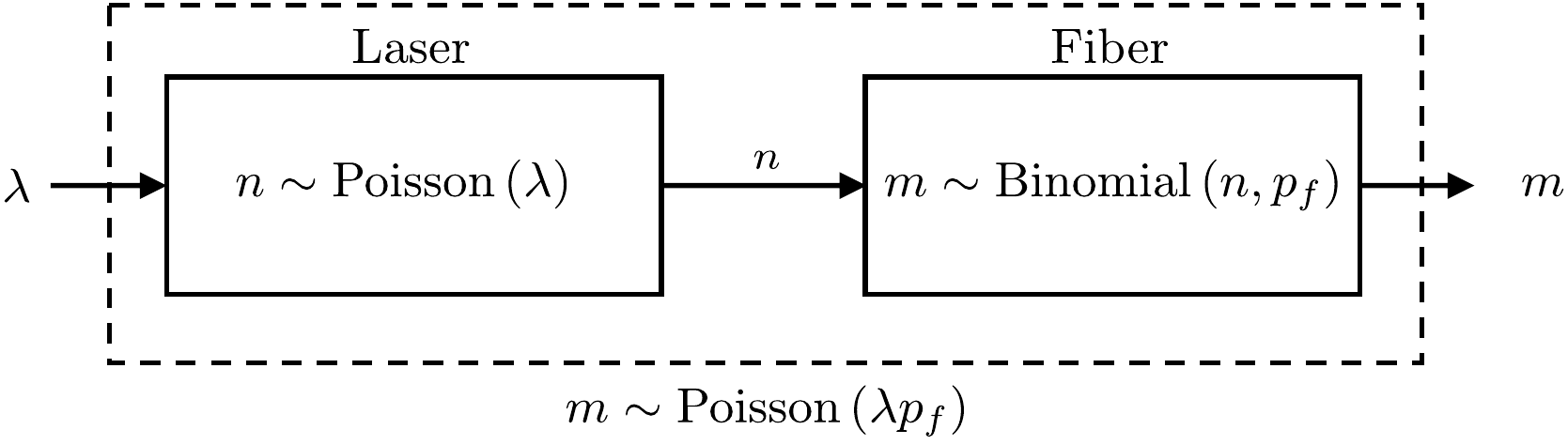}
    \caption{Schematic representation of a laser source connected to a fiber. The number of photons $m$ that passes the fiber follows a Poisson distribution as per \Cref{lem:laser_fiber}.}
    \label{fig:laser_fiber}
\end{figure}

To determine the distribution of $m$ given $\lambda$ and $p^{}_{f}$, we need to account for the combined effect of the Poisson-distributed photon source and the binomial transmission process through the fiber. This relationship is formally described in \Cref{lem:laser_fiber} and visually illustrated in \Cref{fig:laser_fiber}, which shows that the effect of the fiber is to reduce the average number of photons by a factor of $p^{}_{f}$.

\begin{lemma}
\label{lem:laser_fiber}
If the number of photons $n$ entering a fiber is Poisson-distributed with parameter $\lambda$, and each photon has an independent probability $p^{}_{f}$ of being transmitted through the fiber, then the number of photons $m$ exiting the fiber follows a Poisson distribution with parameter $\lambda p^{}_{f}$:
\begin{align}
	P(m \mid \lambda, p^{}_{f}) = \text{Poisson}(m \mid \lambda p^{}_{f}).	
\end{align}
\end{lemma}

\begin{proof}
To derive the distribution of $m$, we calculate the total probability of $m$ photons exiting by summing over the probabilities of all possible photon counts $n$ that could enter the fiber:

\[
\begin{aligned}
P(m \mid \lambda, p^{}_{f}) &= \sum^{\infty}_{n=m} P(m \mid n, p^{}_{f})   P(n \mid \lambda) \\
&= \sum^{\infty}_{n=m} \text{Binomial}(m \mid n, p^{}_{f})   \text{Poisson}(n \mid \lambda),
\end{aligned}
\]
where the indexing is set to start at $m$, since $\text{Binomial}(m \mid n, p^{}_{f})=0\quad\forall\, m>n$. Expanding the expressions for the binomial and Poisson probabilities:

\[
\begin{aligned}
P(m \mid \lambda, p^{}_{f}) &= \sum^{\infty}_{n=m} \frac{n!}{m!(n-m)!} p^{m}_{f} q^{n-m}_{f} \frac{\lambda^n}{n!}e^{-\lambda}.
\end{aligned}
\]

Simplifying this expression:

\[
\begin{aligned}
P(m \mid \lambda, p^{}_{f}) &= \frac{p^{m}_{f}}{m!} e^{-\lambda} \sum^{\infty}_{n=m} \frac{\lambda^n q^{n-m}_{f}}{(n-m)!}.
\end{aligned}
\]

We define $k = n - m$ and re-index the sum:

\[
\begin{aligned}
P(m \mid \lambda, p^{}_{f}) &= \frac{p^{m}_{f}}{m!} e^{-\lambda} \sum^{\infty}_{k=0} \frac{\lambda^{k+m} q^{k}_{f}}{k!},\\
 &= \frac{p^{m}_{f}}{m!} e^{-\lambda} \lambda^{m}\sum^{\infty}_{k=0} \frac{(\lambda q^{}_{f})^{k}}{k!},\\
\end{aligned}
\]

Recognizing the power series expansion for the exponential function~\cite{NIST:DLMF}:

\[
\sum^{\infty}_{k=0} \frac{(\lambda q^{}_{f})^{k}}{k!} = e^{\lambda q^{}_{f}},
\]

the expression simplifies to:

\[
P(m \mid \lambda, p^{}_{f}) = \frac{(\lambda p^{}_{f})^m}{m!} e^{-\lambda p^{}_{f}}.
\]

This final form is the Poisson distribution with parameter $\lambda p^{}_{f}$.
\end{proof}

\subsubsection{Laser-Detector}

Consider a scenario where a Poisson-distributed photon stream is directly connected to a detector with efficiency $p^{}_{c}$ and a dark count probability $p^{}_{d}$. The objective is to derive the probability of a click event at the detector, considering the stochastic nature of both the photon emission process and the detection. This involves marginalizing over all possible numbers of photons $n$ emitted by the source, weighting each by the probability that $n$ photons cause a detection event. This process is formally presented in \Cref{lem:laser_detector} and illustrated in \Cref{fig:laser_detector}.

\begin{figure}[h]
    \centering
    \includegraphics[scale=0.25]{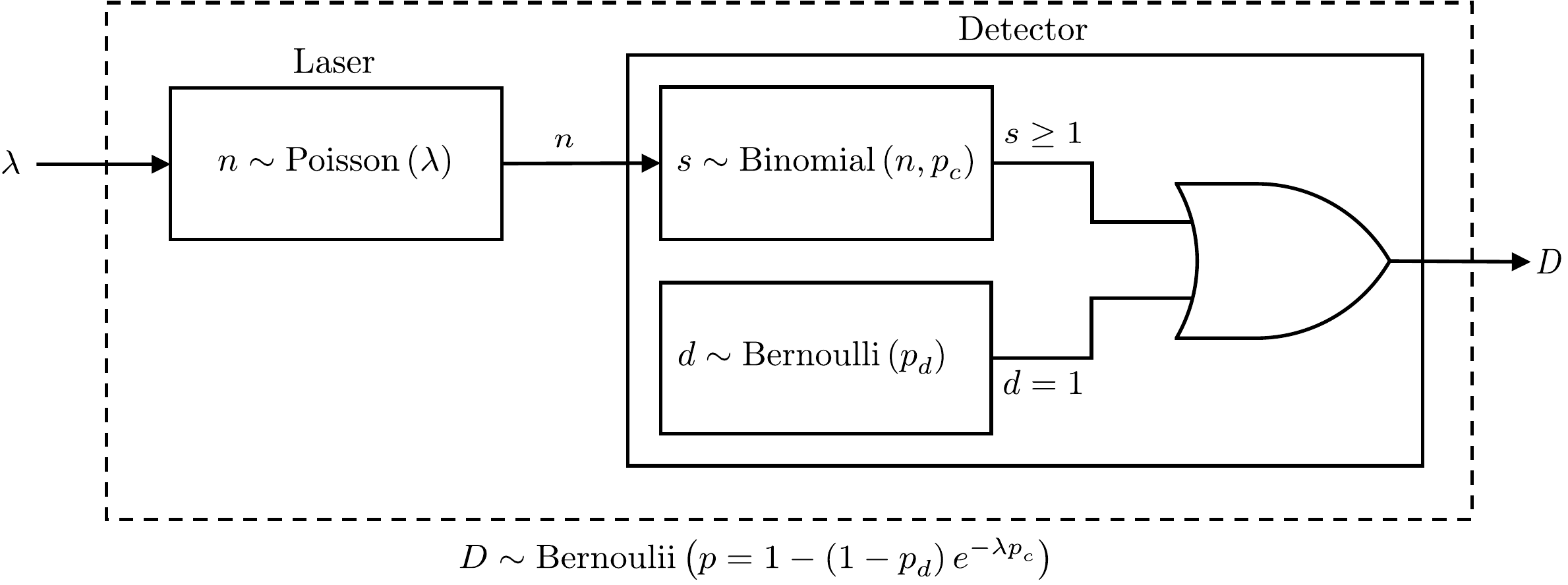} 
    \caption{Illustration of the detection process for photons emitted directly from a laser source. The detection event follows a Bernoulli distribution as per \Cref{lem:laser_detector}}
\label{fig:laser_detector}
\end{figure}

\begin{lemma}
\label{lem:laser_detector}
If a Poisson-distributed number of photons with mean $\lambda$ enters a detector with efficiency $p^{}_{c}$ and dark count probability $p^{}_{d}$, then the click event $D$ is Bernoulli distributed as follows:
\begin{align}
	D \sim \text{Bernoulli}\left(p = 1 - (1-p^{}_{d}) e^{-p^{}_{c} \lambda}\right).
\end{align}
\end{lemma}

\begin{proof}
The probability that the detector registers at least one click, since the processes of incoming photons and dark counts are independent, can be expressed as:
\[
P(D=1 \mid \lambda, p^{}_{c}, p^{}_{d}) = 1 - P(d=0 \mid p^{}_{d}) P(s=0 \mid \lambda, p^{}_{c}).
\]

From \Cref{lem:laser_fiber}, we know that if the laser source emits photons following a Poisson distribution with mean $\lambda$, and the signal detection process has efficiency $p^{}_{c}$, then the number of detected signal photons $s$ is Poisson-distributed with mean $\lambda p^{}_{c}$. Thus, the probability that no photon is detected:
\[
P(s=0 \mid \lambda, p^{}_{c}) = \text{Poisson}(s=0\mid \lambda p^{}_{c})= e^{-p^{}_{c} \lambda}.
\]

Combining these results, we have:
\[
\begin{aligned}
P(D=1 \mid \lambda, p^{}_{c}, p^{}_{d}) &= 1 - P(d=0 \mid p^{}_{d}) P(s=0 \mid \lambda, p^{}_{c}) \\
&= 1 - (1-p^{}_{d}) e^{-p^{}_{c} \lambda},
\end{aligned}
\]
\end{proof}

\subsection{PNS Attack}

We now model Eve's interception of the key, considering specific choices of the bit Alice sends ($x$) and the bases chosen by Alice ($a$) and Bob ($b$). The model integrates the effects of the laser source, fiber losses, Eve's attack vector, misalignment in the optical components, and detector characteristics.

\subsubsection{Generalized PNS Attack}
\label{subsec:generalized_pns}

In our analysis, we consider a generalized Photon Number Splitting (PNS) attack, where Eve has a range of strategic choices that extend beyond the standard assumptions typically made in QKD security models. This model allows Eve to optimize her attack strategy according to a broader set of parameters, providing a more flexible and realistic scenario. Specifically, Eve can choose:

\begin{enumerate}
    \item \textbf{Distance from Alice, $d^{}_{AE}$}: Eve's interception point is between Alice and Bob, with $0 \leq d^{}_{AE} \leq d^{}_{AB}$, where $d^{}_{AB}$ is their total distance.
    \item \textbf{Proportion of intercepted pulses, $\Delta$}: Eve intercepts a fraction $\Delta$ of pulses, $0 \leq \Delta \leq 1$, balancing information gain and detection risk.
    \item \textbf{Number of intercepted photons, $k$}: Eve intercepts $k \geq 1$ photons per intercepted pulse.
    \item \textbf{Channel efficiency, $p^{}_{EB}$}: Eve selects a channel efficiency $0 \leq p^{}_{EB} \leq 1$ to transmit intercepted pulses, simulating legitimate loss to avoid detection.
\end{enumerate}

While many traditional PNS attack models assume that Eve optimizes the channel efficiency $p^{}_{EB}$ to maintain the appearance of normal channel loss, they often implicitly assume some or all of the following:

\begin{itemize}
    \item Eve intercepts every pulse ($\Delta = 1$).
    \item Eve intercepts immediately after Alice ($d^{}_{AE} = 0$).
    \item Eve intercepts exactly one photon per pulse ($k = 1$).
\end{itemize}

By allowing Eve to vary these parameters beyond just optimizing $p^{}_{EB}$, the generalized PNS attack model captures a broader range of potential eavesdropping strategies, which is crucial for accurately inferring $\Delta$, the proportion of intercepted pulses. Assuming a fixed scenario for $d^{}_{AE}$, $k$, or $\Delta$ could lead to incorrect inferences about the extent of Eve's presence. Furthermore, this generalized model enhances the inference on $\Delta$; for instance, if Eve intercepts $k > 1$ photons per pulse while the model assumes $k = 1$, the model might compensate by predicting a larger $\Delta$ than actually exists. By considering a wider variety of possible attack scenarios, the generalized model provides a more comprehensive understanding of the robustness of QKD protocols and improves the accuracy of security assessments.

\subsubsection{PNS Model}

To simplify our derivation, we start with a basic scenario where Alice, Eve, and Bob are directly connected, with Bob using a single detector. This simplified model provides a foundational basis that can be extended to more complex and realistic setups using previously derived lemmas. More specifically, a detection setup that includes both a fiber and a detector can be modeled as a single detector with an effective efficiency that accounts for the fiber loss factor (\Cref{lem:fiber_detector}). For setups involving multiple detectors, the beam-splitter probabilities can be integrated into the detection efficiencies (\Cref{lem:fiber_beamsplitter} and \Cref{lem:pair_detectos}). Additionally, a laser source followed by a fiber can be modeled as a source with reduced intensity, scaled by the fiber's transmission probability (\Cref{lem:laser_fiber}). This modular approach allows us to seamlessly extend the simplified model to more sophisticated real-world configurations without sacrificing analytical rigor.

In this scenario, Alice sends $n$ photons, sampled from a Poisson distribution with parameter $\lambda$, directly to Bob, while Eve intercepts the pulse by capturing up to $k$ photons. If $n < k$, Eve captures all $n$ photons, and Bob receives the remaining photons (see \Cref{fig:PNS}). We present a closed-form expression for the probability of Bob's detector clicking under these conditions in \Cref{theorem:PNS}.
\begin{figure}[h]
    \centering
    \includegraphics[scale=0.25]{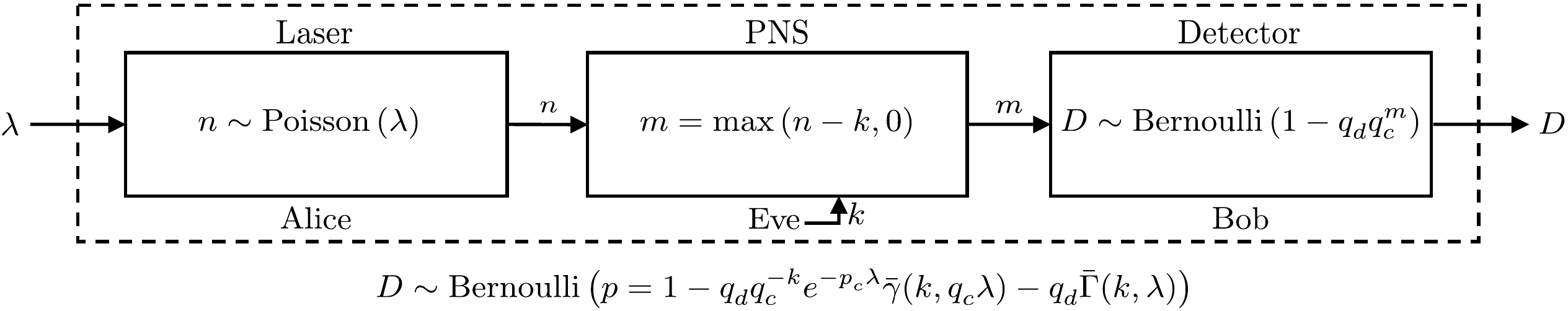}
    \caption{Model representation of the Photon Number Splitting (PNS) attack scenario. Here, Alice uses a laser source emitting photons following a Poisson distribution with mean $\lambda$. Eve intercepts the communication, capturing $k$ photons using a PNS attack strategy. The remaining photons, $m = \max(n - k, 0)$, reach Bob. The probability of a detection event at Bob's detector follows a Bernoulli distribution as per \Cref{theorem:PNS}.}
    \label{fig:PNS}
\end{figure}
\begin{theorem}
\label{theorem:PNS}
Consider a QKD setup where Alice's photon source emits pulses with the number of photons following a Poisson distribution with mean $\lambda$. Eve intercepts up to $k$ photons from each pulse. The remaining photons are sent to Bob, who uses a detector characterized by an efficiency $p^{}_{c}$ and a dark count probability $p^{}_{d}$. Under these conditions, the detection event $D$ at Bob's detector is distributed as follows:

\begin{align}
	P(D \mid \lambda, p^{}_{c}, p^{}_{d}, k)  = \text{Bernoulli}\left(D\mid p=1 - q^{}_{d} q^{-k}_{c} e^{-p^{}_{c} \lambda} \bar{\gamma}(k,q^{}_{c} \lambda) - q^{}_{d} \bar{\Gamma}(k,\lambda)\right),
\end{align}
where $q^{}_{c}=1-p^{}_{c}$ and $q^{}_{d}=1-p^{}_{d}$. $\bar{\Gamma}(s,x)$ and $\bar{\gamma}(s,x) = 1 - \bar{\Gamma}(s,x)$ are the upper and lower regularized incomplete gamma functions, respectively~\cite{NIST:DLMF}. The regularized upper incomplete gamma function is defined as~\cite{NIST:DLMF}:

\[
\bar{\Gamma}(s,x) = \frac{\Gamma(s, x)}{\Gamma(s, 0)}, \quad \Gamma(s, x) = \int^{\infty}_{x} t^{s-1} e^{-t} \, dt.
\]
\end{theorem}

\begin{proof}
To derive the click probability $P(D=1 \mid \lambda, p^{}_{c}, p^{}_{d}, k)$, we consider all possible numbers of photons $n$ that Alice could send, and we marginalize over these photon counts to account for the different scenarios of interception and detection. The total probability is computed as:

\[
\begin{aligned}
P(D=1 \mid \lambda, p^{}_{c}, p^{}_{d}, k) &= \sum^{\infty}_{n=0} P(D=1 \mid p^{}_{c}, p^{}_{d}, n)   P(n \mid \lambda, k).
\end{aligned}
\]

In this expression, the first term represents the click probability given $n$ photons and the detector's characteristics. For $n < k$, Eve intercepts all photons, leaving none for Bob, so the click probability reduces to the dark count probability $p^{}_{d}$. For $n \geq k$, Eve intercepts exactly $k$ photons, and the click probability is determined by the remaining $n-k$ photons. Therefore, we can split the summation into two parts based on the value of $n$:

\[
\begin{aligned}
P(D=1 \mid \lambda, p^{}_{c}, p^{}_{d}, k) &= \sum^{\infty}_{n=k} P(D=1 \mid p^{}_{c}, p^{}_{d}, n-k)   P(n \mid \lambda) + p^{}_{d} \sum^{k-1}_{n=0} P(n \mid \lambda).
\end{aligned}
\]

Expanding this expression, we use the Bernoulli distribution for the click probability, from \Cref{lem:detector}, and the Poisson distribution for the photon count:

\[
\begin{aligned}
P(D=1 \mid \lambda, p^{}_{c}, p^{}_{d}, k) &= \sum^{\infty}_{n=k} \left(1 - q^{}_{d} q^{n-k}_{c}\right) \frac{\lambda^n e^{-\lambda}}{n!} + p^{}_{d} \sum^{k-1}_{n=0} \frac{\lambda^n e^{-\lambda}}{n!}.
\end{aligned}
\]

Recognizing that the second summation represents the cumulative distribution function (CDF) of the Poisson distribution, and utilizing the following relationship~\cite{NIST:DLMF}:

\[
\sum^{k}_{n=0} \text{Poisson}(n \mid \lambda) = \text{PoissonCDF}(k,\lambda) = \bar{\Gamma}(k+1,\lambda),
\]
we rewrite the expression as:

\[
P(D=1 \mid \lambda, p^{}_{c}, p^{}_{d}, k) = \sum^{\infty}_{n=k} \left(1 - q^{}_{d} q^{n-k}_{c}\right) \frac{\lambda^n e^{-\lambda}}{n!} + p^{}_{d} \bar{\Gamma}(k,\lambda).
\]

To further simplify, we split the exponential series:

\begin{align}
P(D=1 \mid \lambda, p^{}_{c}, p^{}_{d}, k) &= \sum^{\infty}_{n=k} \frac{\lambda^n e^{-\lambda}}{n!} - q^{}_{d} \sum^{\infty}_{n=k} q^{n-k}_{c} \frac{\lambda^{n} e^{-\lambda}}{n!} + p^{}_{d} \bar{\Gamma}(k,\lambda),\nonumber\\
=&~1-\bar\Gamma\left(k,\lambda\right)-q^{}_{d}q^{-k}_{c}e^{-\lambda}\sum^{\infty}_{n=k}\frac{\left(q^{}_{c}\lambda\right)^{n}}{n!}+p^{}_{d}\bar\Gamma\left(k,\lambda\right)\nonumber\\
=&~1-q^{}_{d}q^{-k}_{c}e^{-\lambda}e^{q^{}_{c}\lambda}\sum^{\infty}_{n=k}\frac{\left(q^{}_{c}\lambda\right)^{n}}{n!}e^{-q^{}_{c}\lambda}-\bar\Gamma\left(k,\lambda\right)+p^{}_{d}\bar\Gamma\left(k,\lambda\right)\nonumber\\
=&~1-q^{}_{d}q^{-k}_{c}e^{-\lambda+q^{}_{c}\lambda}\bar\gamma\left(k,q^{}_{c}\lambda\right)-\left(1-p^{}_{d}\right)\bar\Gamma\left(k,\lambda\right)\nonumber\\
=&~1-q^{}_{d}q^{-k}_{c}e^{-p^{}_{c}\lambda}\bar\gamma\left(k,q^{}_{c}\lambda\right)-q^{}_{d}\bar\Gamma\left(k,\lambda\right).\nonumber
\end{align}
\end{proof}

Notably, by utilizing the relationship between the Poisson CDF and the incomplete gamma function, we extend the domain of $k$ to the non-negative real numbers. This extension is particularly useful when computing gradients later, as $k$ can now be treated as a continuous parameter rather than a discrete summation index and the derivatives of the incomplete gamma functions with respect to it can be derived.

\subsubsection*{Alice to Eve}

According to \Cref{lem:laser_fiber}, the combination of the laser source and the fiber from Alice to Eve can be modeled as a Poisson distribution with a reduced intensity parameter $\lambda p^{}_{AE}$, where $p^{}_{AE} = 10^{-\alpha d^{}_{AE}/10}$, $d^{}_{AE}$ being the distance between Alice and Eve, $\lambda$ representing the average photon number emitted by Alice, and $\alpha$ as the attenuation coefficient in dB/km.

\subsubsection*{Eve to Bob}

If the measured error probability due to misalignment in the beam splitter is $p^{}_{e}$, the corresponding angular misalignment can be computed as:

\begin{align}
	\hat{p}_{e} = \arcsin\left(\sqrt{p^{}_{e}}\right).
\end{align}

The probability that a photon reaches detector $i$ for specific choices of bit $x$, and bases $a$ and $b$, is given by:

\begin{align}
p(i \mid x, a, b, p^{}_{e}) = \cos\left( \frac{\pi}{2}(i + x) - \frac{\pi}{4}(a - b) + \hat{p}_{e} \right)^2.\label{eq:p_bs}
\end{align}

We can continue modeling the photon transmission through the system by applying \Cref{lem:fiber_beamsplitter} and \Cref{lem:fiber_detector} to combine the effects of the fiber loss from Eve to Bob, the beam splitter, and the detector into a single detector model with an effective efficiency scaled by both the fiber loss and the beam splitter path probability. We then state the following corollary derived from \Cref{theorem:PNS}:

\begin{corollary}
Let $\theta$ represent the set of system parameters:

\begin{align}
	\theta = \{\underset{\theta^{}_{A}}{\underbrace{\lambda, \alpha, d^{}_{AB}}}, \underset{\theta^{}_{B}}{\underbrace{p^{}_{c}, p^{}_{d}, p^{}_{e}}}, \underset{\theta^{}_{E}}{\underbrace{d^{}_{AE}, p^{}_{EB}, k, \Delta}}\}.
\end{align}

The probability of detection at detector $D^{}_{i}$ for specific choices of bases ($a$ and $b$), a bit choice ($x$), intensity ($\lambda$), and whether Eve intercepted or not ($e$), is given by:

\begin{align}
P(D^{}_{i}\mid \theta, a, b, x, e,\lambda) = P(D \mid \hat{\lambda}, p^{}_{c^{}_{i}}, p^{}_{d^{}_{i}}, \hat{k}),
\end{align}
where $P(D \mid \hat{\lambda}, \hat{p}_{c}, p^{}_{d^{}_{i}}, \hat{k})$ is the probability of detection from \Cref{theorem:PNS}, and the adjusted parameters $\hat{\lambda}$, $\hat{p}_{c}$, and $\hat{k}$ are defined as:

\begin{align}
\hat{\lambda} &= \lambda p^{(1-e)}_{AB} p^{e}_{AE},\\
\hat{p}_{c} &= p(i \mid x, a, b, p^{}_{e}) p^{}_{c^{}_{i}} p^{e}_{EB},\\
\hat{k} &= e k,
\end{align}
where $p^{}_{AE} = 10^{-\alpha d^{}_{AE}/10}$ and $p^{}_{AB} = 10^{-\alpha d^{}_{AB}/10}$ are the channel efficiencies from Alice to Bob and Alice to Eve respectively.
\end{corollary}

\subsubsection{All Possible Detection Events}

We now have all the necessary components to derive the probabilities of the four joint detection events that Bob can observe, given the set of session parameters, $\theta$, and pulse parameters, namely $a$, $b$, $x$ and $e$. By deriving the marginal and union probabilities, as discussed in \Cref{sec:pair_of_detectors} , we have:

\begin{align}
    \mathbf{P}_{abxe\lambda}^{\vee}(\theta) &= P(D^{}_{0} \vee D^{}_{1} = 1 \mid \theta,a,b,x,e,\lambda),\\
    \mathbf{P}_{abxe\lambda}^{*1}(\theta) &= P(D^{}_{0} = 1 \mid \theta,a,b,x,e,\lambda),\\
    \mathbf{P}_{abxe\lambda}^{1*}(\theta) &= P(D^{}_{1} = 1 \mid \theta,a,b,x,e,\lambda).
\end{align}

Using these, the joint probabilities can be reconstructed as:

\begin{align}
    \mathbf{P}_{abxe\lambda}^{00}(\theta) &= 1 - \mathbf{P}_{abxe\lambda}^{\vee}(\theta),\\
    \mathbf{P}_{abxe\lambda}^{01}(\theta) &= \mathbf{P}_{abxe\lambda}^{\vee}(\theta) - \mathbf{P}_{abxe\lambda}^{1*}(\theta),\\
    \mathbf{P}_{abxe\lambda}^{10}(\theta) &= \mathbf{P}_{abxe\lambda}^{\vee}(\theta) - \mathbf{P}_{abxe\lambda}^{*1}(\theta),\\
    \mathbf{P}_{abxe\lambda}^{11}(\theta) &= \mathbf{P}_{abxe\lambda}^{*1}(\theta) + \mathbf{P}_{abxe\lambda}^{1*}(\theta) - \mathbf{P}_{abxe\lambda}^{\vee}(\theta).
\end{align}

From \Cref{lem:pair_detectos}, we know that the union probability for a two-detector system can be modeled as a single detector with adjusted dark count and click probabilities:

\begin{align}
    p^{\vee}_{d} &= 1 - q^{}_{d^{}_{0}}q^{}_{d^{}_{1}},\\
    p^{\vee}_{c} &= p(1, x, a, b, p^{}_{e})p^{}_{c^{}_{1}} + p(0, x, a, b, p^{}_{e})p^{}_{c^{}_{0}}.
\end{align}

We denote the set of four possible outcomes given the session and pulse parameters as:
\begin{align}
	\mathbf{P}_{abxe\lambda}(\theta) = \left[\mathbf{P}_{abxe\lambda}^{00}(\theta)\quad \mathbf{P}_{abxe\lambda}^{01}(\theta)\quad \mathbf{P}_{abxe\lambda}^{10}(\theta)\quad \mathbf{P}_{abxe\lambda}^{11}(\theta)\right].\label{eq:P_abxel}
\end{align}
\subsubsection{Marginalization of Unknowns in the Sifting Phase}

For the purpose of inferring the percentage of tagged photons $\Delta$, we rely only on the information available during the sifting phase, specifically, the basis choices $a$ and $b$ and marginalize over the variables $x$ and $e$:

\begin{align}
	\mathbf{P}_{ab\lambda}(\theta) = \sum^{}_{x\in\{0,1\}}\sum^{}_{e\in\{0,1\}}\mathbf{P}_{abxe\lambda}(\theta)P(e)P(x)\label{eq:Pabl}.
\end{align}

Assuming that Alice sends the bits with equal probabilities using a quantum random number generator, then $P(x=0) = P(x=1) = \frac{1}{2}$. If Eve tags a pulse with probability $\Delta$, then $P(e=1) = \Delta$ and $P(e=0) = 1 - \Delta$. Note from \Cref{eq:p_bs} that:

\begin{align}
p(i\mid x, b, a, p^{}_{e}) = 
\begin{cases} 
p(i\mid 1-x, a, b, p^{}_{e}) & \text{if } a \neq b, \\
p(i\mid x, a, b, p^{}_{e}) & \text{if } a = b.
\end{cases}
\end{align}

Thus, swapping the basis will either swap the terms of the sum in \Cref{eq:Pabl} when $a \neq b$, or it will have no effect when $a = b$. Consequently, the probability only depends on whether the bases matched or not, or more formally:

\begin{align}
	\mathbf{P}_{ab\lambda}(\theta) = \mathbf{P}_{ba\lambda}(\theta).\label{eq-symmetric-bases}
\end{align}

Therefore, for analytical purposes, we can assume Bob always measures in a fixed basis (e.g., $b = 1$), while Alice switches her basis with 50\% probability. The observed probabilities for $a = b$ and $a \neq b$ will be identical.

In conclusion, for a specific intensity $\lambda$, Bob can observe two distinct probabilities for a given detector $i$ based on whether the bases matched or not:

\begin{align}
	\mathbf{P}_{m\lambda}(\theta) &= 	\mathbf{P}_{(a=m)(b=1)\lambda}(\theta)\label{eq:Pi_m},
\end{align}
where $m$ is a flag to encode whether Alice and Bob have matching basis. Finally, the set of probabilities for all possible outcomes is:

\begin{align}
    \mathbf{P}_{\lambda}(\theta) = \left[\mathbf{P}_{(m=0)\lambda}(\theta)   P(m=0)\quad \mathbf{P}_{(m=1)\lambda}(\theta)   P(m=1)\right],\label{eq:P_iid}
\end{align}

If we assume Alice and Bob switch bases with equal probability, $P(m=0) = P(m=1) = \frac{1}{2}$, this simplifies to:

\begin{align}
    \mathbf{P}_{\lambda}(\theta) = \left[\mathbf{P}_{(m=0)\lambda}(\theta)\quad \mathbf{P}_{(m=1)\lambda}(\theta) \right]/2.
\end{align}

\subsubsection{Multiple Intensities}
\label{sec-multiple-intensities}
Thus far, our probabilistic model has assumed a single intensity $\lambda$ for the session. However, as discussed in \Cref{subsec:generalized_pns} , Eve can manipulate the channel by choosing a different attenuation rate, allowing her to disguise her presence as normal channel loss. With a single intensity, Eve can select a channel efficiency that minimizes the statistical differences between the click distributions with and without her interference, potentially reducing the distinguishability between $P(\theta|e=0)$ and $P(\theta|e=1)$ (e.g., by minimizing the KL divergence between them). 

To mitigate this, we must use multiple intensities, ensuring that Eve cannot conceal herself by tuning a single channel efficiency value. By varying the intensities, it becomes more challenging for Eve to simultaneously minimize the statistical differences across all intensity settings, thereby making her presence more detectable.

We adopt a straightforward strategy of selecting equally spaced intensities, using at least four to account for the four unknown parameters in Eve's potential attack (i.e. $d^{}_{AE}$, $p^{}_{EB}$, $k$ and $\Delta$). The range for the minimum and maximum intensity values is determined using the following heuristics:

\begin{itemize}
    \item $\lambda^{}_{\min}$: The lowest intensity where, if Eve intercepts a single photon immediately after Alice, she must use a channel with efficiency $p^{}_{EB} = 1$ to avoid detection.
    
    \item $\lambda^{}_{\max}$: The highest intensity that maximizes the proportion of events where $D^{}_{0} \oplus D^{}_{1}$ occurs, as cases where both detectors either click or do not click provide no useful information and appear uniformly random from Bob's perspective.
\end{itemize}

To reflect the use of multiple intensities, we update our notation as follows:

\begin{align}
\mathbf{P}(\theta) = \left[\mathbf{P}_{\lambda^{}_{1}}(\theta)\quad \mathbf{P}_{\lambda^{}_{2}}(\theta)\quad\hdots\quad \mathbf{P}_{\lambda^{}_{N^{}_{\lambda}}}(\theta)\right]/N^{}_{\lambda},\label{eq:P_iid}	
\end{align}
where $N^{}_{\lambda}$ is the number of intensities. In this study, we assume that each intensity is equally likely to be selected. Note that while other protocols discard non-matching bases, we still probabilistically modeled them to utilize their statistics when infering Eve’s presence (though not for key generation), thus increasing the statistical power to detect eavesdropping.

\subsection{After-Pulsing}

In real-world QKD systems, photon detectors are affected by after-pulsing, a phenomenon where a click from a prior detection event increases the chance of another click in the next pulse, regardless of whether a new photon is actually detected~\cite{restelli2013afterpulse}. This effect can introduce correlations between detection events, breaking the independence assumption that would otherwise imply each event is isolated from others. After-pulsing can be reduced by increasing the detector's dead time, which lowers the pulse frequency and consequently reduces the key rate. On the other hand, simply neglecting the impact of after-pulsing results in incorrect error rate assessments and inaccurate key rate calculations. Both approaches fail to fully address the complexities introduced by after-pulsing, potentially compromising the security analysis of a QKD system or the key generation rate.

To address this challenge, we employ a Hidden Markov Model (HMM) to rigorously capture the behavior of detectors affected by after-pulsing. This approach improves the accuracy of security analysis and key rate calculations. The details of this HMM-based modeling are presented in the following subsections.

\subsubsection{Single Detector, Single Probability}
Consider a scenario where Alice repeatedly sends identical pulses to Bob, consistently choosing the same bit value $x$, basis $a$, and intensity $\lambda$. Bob, on the other hand, measures these pulses in a fixed basis $b$, with no interference from Eve. The signal will cause the detector to click with probability $p^{}_{1}$, and, consequently, does not click with probability $p^{}_{0}=1-p^{}_{1}$. The detector then can be in one of three internal states, depending on its recent detection history, see \Cref{fig:hmm_single}:

\begin{itemize}
    \item \textbf{State $S^{}_{0}$}: No click occurs, at which it remains in this state with probability $p^{}_{0}$, or transition to the $S^{}_{1}$ state with probability $p^{}_{1}$.
    \item \textbf{State $S^{}_{1}$}: A click due to either a signal or a dark count. It will then transition to the after-pulse state $S^{}_{\hat{1}}$ with probability $q$, or with probability $p^{}_{0}(1 - q)$ transition to state $S^{}_{0}$ or remain in state $S^{}_{1}$ with probability $p^{}_{1}(1 - q)$.
    \item \textbf{State $S^{}_{\hat{1}}$}: A click due to an after-pulse effect. This state occurs following a click in the previous detection window (state $S^{}_{1}$). Then it transitions to the $S^{}_{0}$ state with probability $p^{}_{0}$ or state $S^{}_{1}$ with probability $p^{}_{1}$, but it cannot remain in this state.
\end{itemize}
\begin{figure}[h]
    \centering
    \includegraphics[scale=0.25]{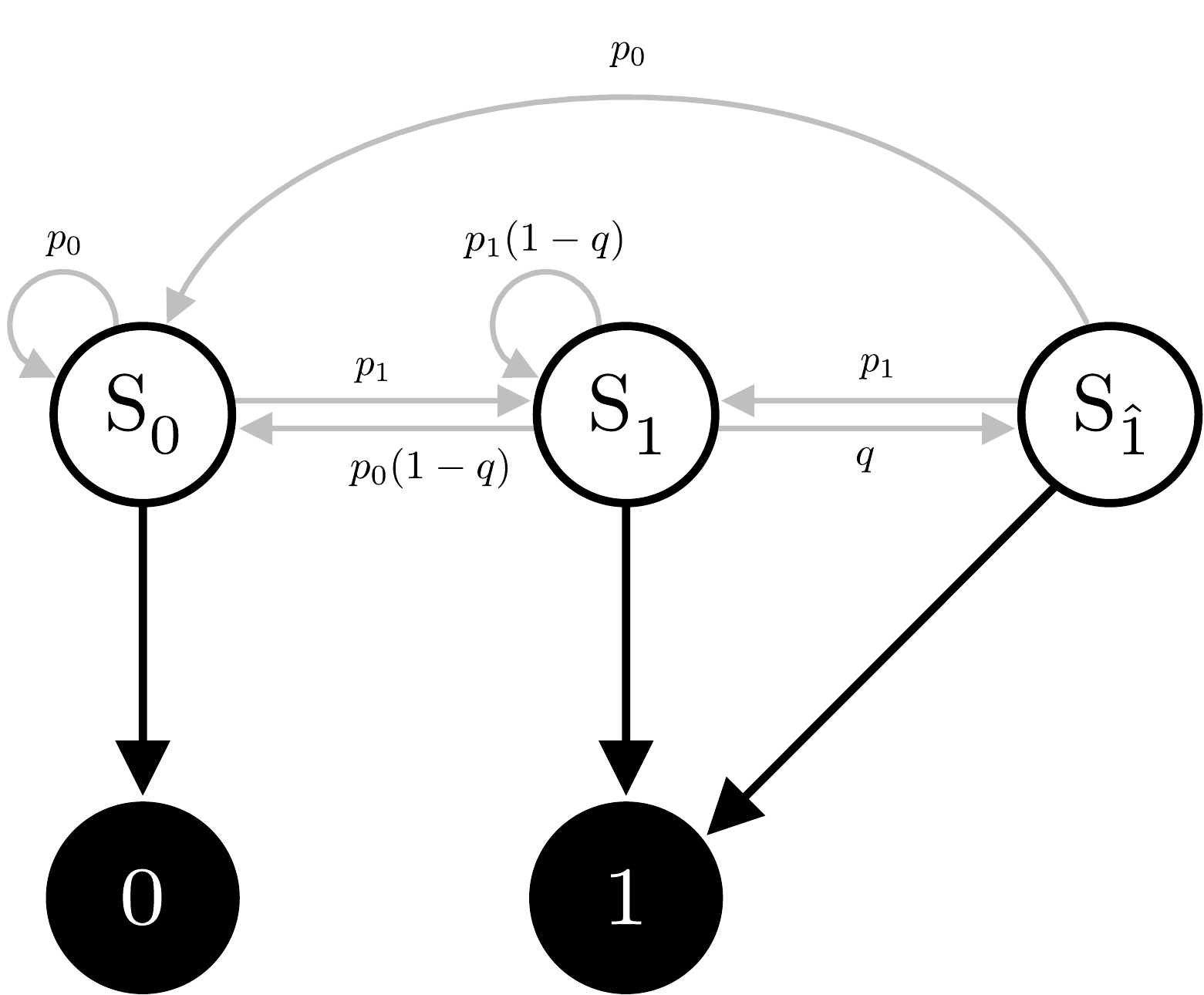}
    \caption{Hidden Markov Model (HMM) representation of the single detector scenario. The states $S^{}_{0}$, $S^{}_{1}$, and $S^{}_{\hat{1}}$ represent no click, a click due to signal or dark count, and a click due to an after-pulse, respectively. Arrows indicate possible state transitions with associated probabilities.}
    \label{fig:hmm_single}
\end{figure}

Hidden Markov Models (HMMs) are characterized by three sets of probabilities:

\begin{itemize}
    \item \textbf{Transition matrix ($\mathbf{T}$)}: Probabilities of transitioning from one hidden state to another (grey arrows in \Cref{fig:hmm_single}).

    \item \textbf{Emission matrix ($\mathbf{E}$)}: Probabilities of observing specific outputs given the current hidden state (black arrows in \Cref{fig:hmm_single}).

    \item \textbf{Initial state probabilities ($\boldsymbol{\pi}$)}: Represents the probabilities of starting in each hidden state at the beginning of the process (usually assumed uniform or alwayas starts at a specific state).
\end{itemize}

For the HMM presented in \Cref{fig:hmm_single}, the transition and emission matrix are defined as follows:

\begin{align}
\mathbf{T}(p,q) = 
\begin{blockarray}{llll}
  & S^{}_{0} & S^{}_{1} & S^{}_{\hat{1}} \\
\begin{block}{l[lll]}
  S^{}_{0} & p^{}_{0} & p^{}_{1} & 0 \\
  S^{}_{1} & p^{}_{0}(1-q) & p^{}_{1}(1-q) & q \\
  S^{}_{\hat{1}} & p^{}_{0} & p^{}_{1} & 0 \\
\end{block}
\end{blockarray}, \quad \mathbf{E} = 
	\begin{blockarray}{lcc}
  		& 0 & 1  \\
		\begin{block}{l[cc]}
			S^{}_{0} & 1 & 0 \\
			S^{}_{1} & 0 & 1 \\
			S^{}_{\hat{1}} & 0 & 1 \\
		\end{block}
	\end{blockarray}, \label{eq:emission}
\end{align}
where $p=\left[p^{}_{0}~p^{}_{1}\right]$. The emission matrix $\mathbf{E}$, which represents the probabilities of observing each possible outcome (0 or 1) given the detector's state, the black arrows in \Cref{fig:hmm_single}, is structured as follows:

We assume a first-order Markov chain (i.e., the next state depends only on the current state) and disallow consecutive after-pulse states. This reflects typical detector behavior; however, the model is flexible: transition probabilities can be adjusted for different after-pulsing dynamics, and higher-order Markov chains can be implemented by expanding the state space~\cite{murphy2002dynamic}.

\subsubsection{Two Detectors, Single Probability}

For a two-detector setup, the click probabilities for each detector combination—namely (00), (01), (10), and (11)—have already been derived in \Cref{eq:P_abxel} for specific pulse parameters $a$, $b$, $x$, $e$, and $\lambda$. We denote these probabilities as $p^{}_{00}$, $p^{}_{01}$, $p^{}_{10}$, and $p^{}_{11}$, respectively. Let $q^{}_{0}$ and $q^{}_{1}$ represent the after-pulse probabilities for detectors $D^{}_{0}$ and $D^{}_{1}$, respectively. We also define the following marginal probabilities:

\begin{align}
    p^{}_{*0} &= p^{}_{00} + p^{}_{10}, \\
    p^{}_{*1} &= p^{}_{01} + p^{}_{11}, \\
    p^{}_{0*} &= p^{}_{00} + p^{}_{01}, \\
    p^{}_{1*} &= p^{}_{10} + p^{}_{11}.
\end{align}

For the after-pulse probabilities we would need to derive the opposite, we already have access to the marginals, $q^{}_{0}$ and $q^{}_{1}$, and we would like to derive the joint distributions. Since the after-pulse triggers are considered independent:

\begin{align}
    q^{}_{00} &= (1 - q^{}_{0})(1 - q^{}_{1}), \\
    q^{}_{01} &= 1 - q^{}_{00} - q^{}_{1}, \\
    q^{}_{10} &= 1 - q^{}_{00} - q^{}_{0}, \\
    q^{}_{11} &= q^{}_{0}q^{}_{1}.
\end{align}

In this two-detector scenario, the state and observation spaces expand accordingly. The set of possible states $S$ and observations $O$ are given by:

\begin{align}
    S &= \left\{S^{}_{00}, S^{}_{01}, S^{}_{10}, S^{}_{11}, S^{}_{0\hat{1}}, S^{}_{\hat{1}0}, S^{}_{1\hat{1}}, S^{}_{\hat{1}1}, S^{}_{\hat{1}\hat{1}}\right\}, \\
    O &= \left\{00, 01, 10, 11\right\},
\end{align}
where, for example, $S^{}_{0\hat{1}}$ denotes the state where detector $D^{}_{1}$ did not click and detector $D^{}_{0}$ clicked due to an after-pulse.

The emission matrix $\mathbf{E}$ is constructed to reflect the probabilities of each possible observation given the detector states:

\begin{align}
	\mathbf{E} = 
	\begin{blockarray}{rcccc}
		\begin{block}{rcccc}
		 	& 00 & 01 & 10 & 11 \\
		\end{block}
		\begin{block}{r[cccc]}
			S^{}_{00} & 1 & 0 & 0 & 0 \\
			S^{}_{01} & 0 & 1 & 0 & 0 \\
			S^{}_{10} & 0 & 0 & 1 & 0 \\
			S^{}_{11} & 0 & 0 & 0 & 1 \\
			S^{}_{0\hat{1}} & 0 & 1 & 0 & 0 \\
			S^{}_{\hat{1}0} & 0 & 0 & 1 & 0 \\
			S^{}_{1\hat{1}} & 0 & 0 & 0 & 1 \\
			S^{}_{\hat{1}1} & 0 & 0 & 0 & 1 \\
			S^{}_{\hat{1}\hat{1}} & 0 & 0 & 0 & 1 \\
		\end{block}
	\end{blockarray}.
\end{align}

Constructing the transition matrix $\mathbf{T}$ is more complex due to the extended state space. The matrix $\mathbf{T}$ captures the transition probabilities between each state, accounting for both the regular transitions and those influenced by after-pulsing effects:

\begin{align}
\mathbf{T}(p,q) = \left[
\begin{blockarray}{cccclllll}
\begin{block}{cccc][lllll}
 \phantom{(1-q^{}_{0})\cdot}~[p^{}_{00} & p^{}_{01} & p^{}_{10} & p^{}_{11} & 0 & 0 & 0 & 0 & 0^{}_{\phantom{11}}]\\
\end{block}
\begin{block}{cccc][lllll}
 (1-q^{}_{0})\cdot [p^{}_{00} & p^{}_{01} & p^{}_{10} & p^{}_{11} & p^{}_{0*} & 0 & p^{}_{1*} & 0 & 0^{}_{\phantom{11}}]\cdot q^{}_{0}\\
\end{block}
\begin{block}{cccc][lllll}
 (1-q^{}_{1})\cdot[p^{}_{00} & p^{}_{01} & p^{}_{10} & p^{}_{11} & 0 & p^{}_{*0} & 0 & p^{}_{*1} & 0^{}_{\phantom{11}}]\cdot q^{}_{1}\\
\end{block}
\begin{block}{cccc][lllll}
 \phantom{(1-)}q^{}_{00}\cdot[p^{}_{00} & p^{}_{01} & p^{}_{10} & p^{}_{11} & p^{}_{0*}q^{}_{01} & p^{}_{*0}q^{}_{10} & p^{}_{1*}q^{}_{01} & p^{}_{*1}q^{}_{10} & q^{}_{11}]\\
\end{block}
\begin{block}{cccc][lllll}
 \phantom{(1-q^{}_{0})\cdot}~[p^{}_{00} & p^{}_{01} & p^{}_{10} & p^{}_{11} & 0 & 0 & 0 & 0 & 0^{}_{\phantom{11}}] \\
\end{block}
\begin{block}{cccc][lllll}
 \phantom{(1-q^{}_{0})\cdot}~[p^{}_{00} & p^{}_{01} & p^{}_{10} & p^{}_{11} & 0 & 0 & 0 & 0 & 0^{}_{\phantom{11}}] \\
\end{block}
\begin{block}{cccc][lllll}
 (1-q^{}_{1})\cdot[p^{}_{00} & p^{}_{01} & p^{}_{10} & p^{}_{11} & 0 & p^{}_{*0} & 0 & p^{}_{*1} & 0^{}_{\phantom{11}}]\cdot q^{}_{1}\\
\end{block}
\begin{block}{cccc][lllll}
 (1-q^{}_{0})\cdot [p^{}_{00} & p^{}_{01} & p^{}_{10} & p^{}_{11} & p^{}_{0*} & 0 & p^{}_{1*} & 0 & 0^{}_{\phantom{11}}]\cdot q^{}_{0}\\
\end{block}
\begin{block}{cccc][lllll}
 \phantom{(1-q^{}_{0})\cdot}~[p^{}_{00} & p^{}_{01} & p^{}_{10} & p^{}_{11} & 0 & 0 & 0 & 0 & 0^{}_{\phantom{11}}]\\
\end{block}
\end{blockarray}\right],\label{eq:T_p_q}
\end{align}
where $p=\left[p^{}_{00}~p^{}_{01}~p^{}_{01}~p^{}_{01}\right]$ and $q=\left[q^{}_{0}~q^{}_{1}\right]$. The order of states in both the rows and columns of the transition matrix $\mathbf{T}$ is consistent with the order of the states in the rows of the emission matrix $\mathbf{E}$. 

The matrix is organized into sub-blocks to highlight different transitions. The left sub-block represents the transitions without after-pulsing, while the right sub-blocks account for transitions influenced by after-pulsing. For instance, when the system is in state $S^{}_{01}$, second row, there is a $q^{}_{0}$ probability of transitioning to an after-pulse state, either $S^{}_{0\hat{1}}$ with probability $p^{}_{0*}$ or $S^{}_{1\hat{1}}$ with probability $p^{}_{1*}$. If no after-pulse occurs, the system transitions among states $S^{}_{00}$, $S^{}_{01}$, $S^{}_{10}$, and $S^{}_{11}$ with their respective probabilities $p^{}_{00}$, $p^{}_{01}$, $p^{}_{10}$, and $p^{}_{11}$ but scaled by ($1 - q^{}_{0}$). An illustration for specific sets of $p$ and $q$ is provided in \Cref{fig:hmm_pair}.

\begin{figure}[h]
    \centering
    \includegraphics[scale=0.25]{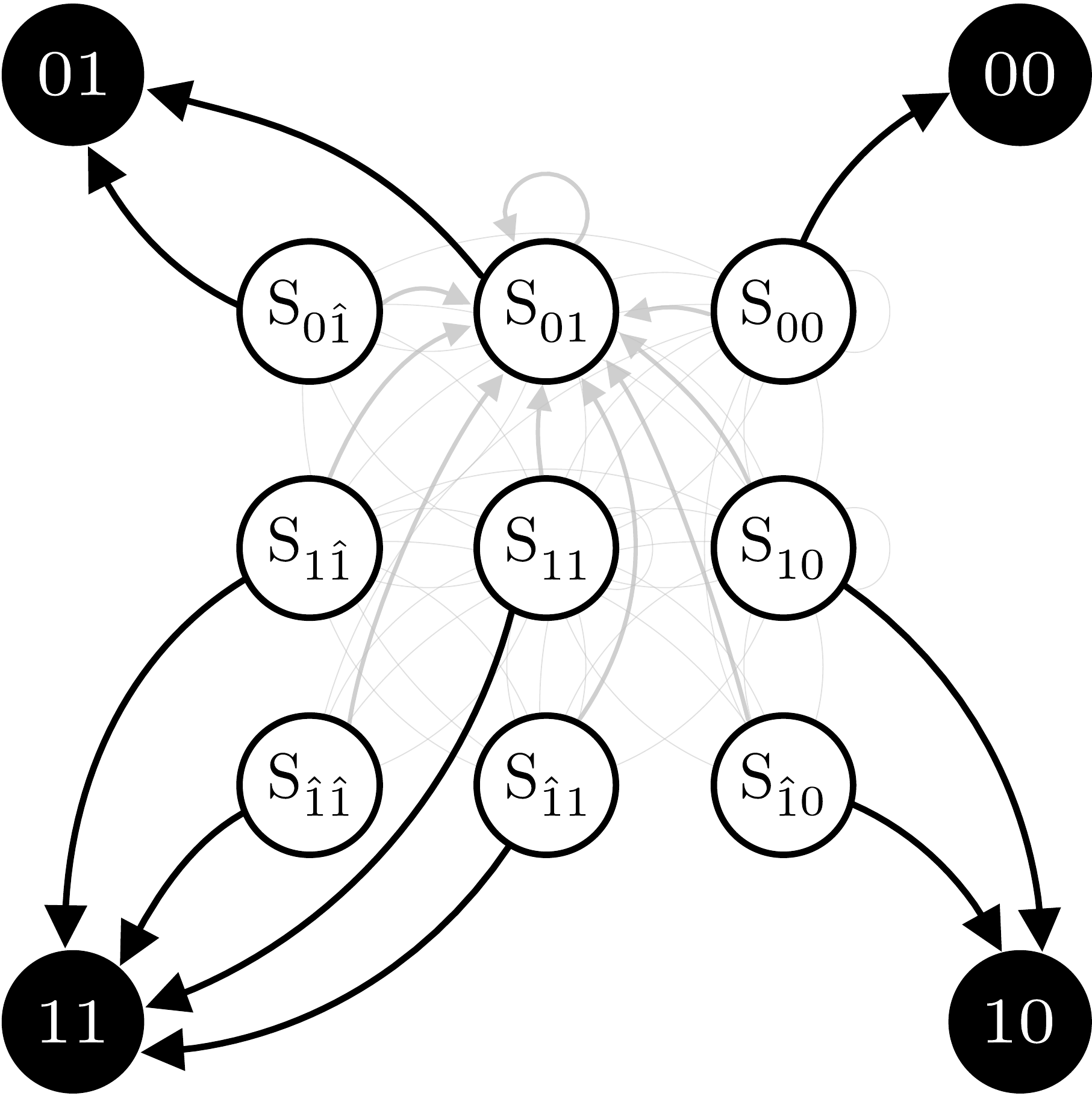}
    \caption{State transition diagram for a specific set of pulse parameters. The figure illustrates the transitions for the probabilities $p = [0.1~0.7~0.1~0.1]$ and after-pulse probabilities $q^{}_{0} = q^{}_{1} = 0.1$. Each state represents a unique combination of detector responses and after-pulse conditions. Arrows indicate the possible transitions between states, with the thickness of the lines being proportional to their associated transition probabilities.}
    \label{fig:hmm_pair}
\end{figure}

\subsubsection{Two Detectors, Multiple Probabilities}
\label{sec:inter-intra-mode}
In the previous subsection, we considered the scenario where Bob receives a repeated pulse with fixed parameters ($a$, $b$, $x$, $e$, and $\lambda$). Under this condition, the behavior of the system could be modeled using a Hidden Markov Model (HMM) with a single transition matrix that corresponds to these fixed parameters.

Now, we consider a more complex scenario where Alice alternates her bit choice $x$ between $0$ and $1$. Although in QKD protocols this bit switching is typically assumed to be independent and occur with equal probability, we will formulate a more general case. We define the transition probabilities between bit states as follows:

\begin{align}
	P(0 \mid 0) &= P(x^{}_{t} = 0 \mid x^{}_{t-1} = 0), \\
	P(0 \mid 1) &= P(x^{}_{t} = 0 \mid x^{}_{t-1} = 1), \\
	P(1 \mid 0) &= P(x^{}_{t} = 1 \mid x^{}_{t-1} = 0), \\
	P(1 \mid 1) &= P(x^{}_{t} = 1 \mid x^{}_{t-1} = 1).
\end{align}

Let $\mathbf{T}_{0}$ and $\mathbf{T}_{1}$ represent the detector dynamics for cases where Alice's bit $x$ was consistently $0$ or $1$, respectively. Given this setup, our goal is to construct an HMM that captures the combined dynamics of the system as Alice's bit choice switches over time. The key is to describe both the intra-mode transitions (transitions within the same bit choice) and inter-mode transitions (transitions between different bit choices), effectively combining the two HMMs ($\mathbf{T}_{0}$ and $\mathbf{T}_{1}$) that represent the system's behavior for $x = 0$ and $x = 1$ (see \Cref{fig:hmms_switch}).

To demonstrate this, consider a generic state $S^{}_{ij}$ representing the detectors' internal states in a particular mode, either $x = 0$ or $x = 1$. The mode here indicates which of the two separate HMMs (either $\mathbf{T}_{0}$ or $\mathbf{T}_{1}$) currently governs the transitions.

If the system is in state $S^{}_{ij}$ and mode $x = 0$, there are two possible types of transitions:

\begin{itemize}
    \item \textbf{Intra-Mode Transition}: With probability $P(0 \mid 0)$, the bit choice remains $x = 0$. The system transitions according to the probabilities defined by the current state's corresponding entry in $\mathbf{T}_{0}$. For example, the probability of transitioning from $S^{}_{ij}$ in mode $0$ to $S^{}_{kl}$ in mode $0$ is given by $P(0 \mid 0) \times \mathbf{T}_{0}(S^{}_{ij} \to S^{}_{kl})$.

    \item \textbf{Inter-Mode Transition}: With probability $P(1 \mid 0)$, the bit choice switches to $x = 1$. The system then transitions according to the probabilities defined by the corresponding entry in $\mathbf{T}_{1}$. For example, the probability of transitioning from $S^{}_{ij}$ in mode $0$ to $S^{}_{kl}$ in mode $1$ is $P(1 \mid 0) \times \mathbf{T}_{1}(S^{}_{ij} \to S^{}_{kl})$.
\end{itemize}

The analysis is equivalent if the system is in mode $x = 1$; however, the intra-mode transitions will involve $P(1 \mid 1)$ and the inter-mode transitions will involve $P(0 \mid 1)$, transitioning according to either $\mathbf{T}_{1}$ or $\mathbf{T}_{0}$, respectively.

Using this framework, the combined transition matrix $\mathbf{T}$ can be constructed to account for both the detector state transitions and the changes in Alice's bit choice. This leads to the following block matrix structure:

\begin{align}
	\mathbf{T} = 
	\begin{bmatrix}
		P(0 \mid 0) \mathbf{T}_{0} & P(1 \mid 0) \mathbf{T}_{1} \\
		P(0 \mid 1) \mathbf{T}_{0} & P(1 \mid 1) \mathbf{T}_{1}
	\end{bmatrix}.
\end{align}

The combined transition matrix $\mathbf{T}$ captures all possible transitions from any state $S^{}_{ij}$ in any mode to any state $S^{}_{kl}$ in any mode, encompassing the full dynamic process of the system as it evolves both within and between modes.

\begin{figure}[h]
    \centering
    \includegraphics[scale=0.25]{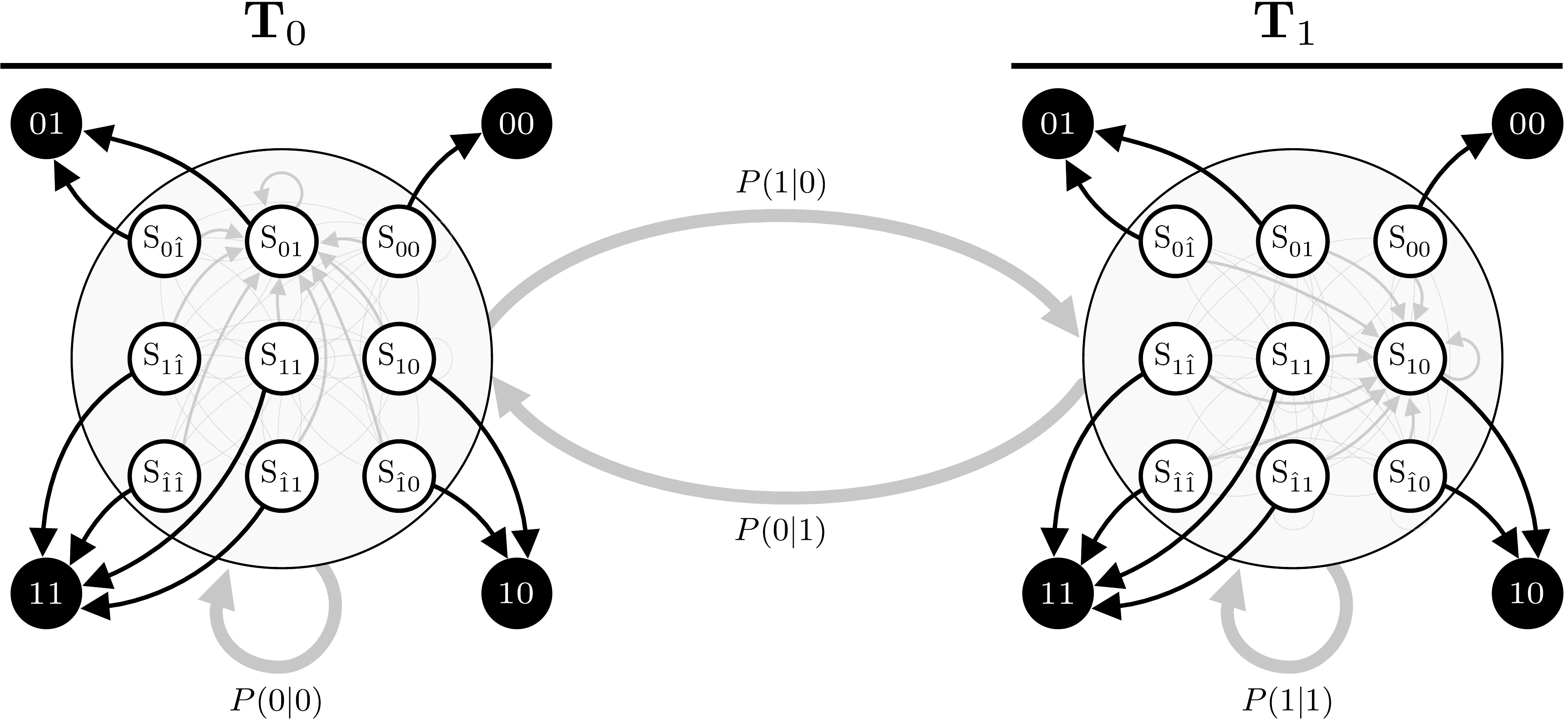}
    \caption{State transition diagram for a specific set of pulse parameters. The left and right clusters represent transitions within the HMMs corresponding to different pulse parameters, specifically for cases where Alice's bit $x$ changes. For illustrative purposes, $\mathbf{T}_{0}$ was constructed using probabilities $p = [0.1, 0.7, 0.1, 0.1]$ and $\mathbf{T}_{1}$ using probabilities $p = [0.1, 0.1, 0.7, 0.1]$, both with after-pulse probabilities $q = [0.1, 0.1]$.}
    \label{fig:hmms_switch}
\end{figure}

\subsubsection{Constructing the Full Transition Matrix}

This section constructs the full HMM by incorporating transitions for all pulse parameters, such as basis choices ($a$ and $b$), bit choices ($x$), the eavesdropping flag ($e$), and intensity settings ($\lambda$), to fully describe the detection dynamics in a quantum key distribution (QKD) system. The transition matrix is built incrementally, starting with individual components and progressively incorporating all relevant variables.

\subsubsection*{Intensity-Based Transitions}

We begin by constructing the transition matrix for a specific set of pulse parameters. This matrix captures the transition dynamics when Alice and Bob's bases ($a, b$), Alice's bit choice ($x$), the presence or absence of eavesdropping ($e$), and the photon intensity ($\lambda^{}_{i}$) are all fixed:

\begin{align}
	\mathbf{T}_{abxe\lambda}(\theta) = \mathbf{T}(\mathbf{P}_{abxe\lambda}(\theta), p^{}_{a}).
\end{align}

Here, $\mathbf{P}_{abxe\lambda}(\theta)$ represents the click probabilities given the session and pulse parameters, as defined in \Cref{eq:P_abxel}. The vector $p^{}_{a} = \left[p^{}_{a^{}_{0}} \; p^{}_{a^{}_{1}}\right]$ denotes the after-pulse probabilities for detectors $D^{}_{0}$ and $D^{}_{1}$, respectively. The parameter set $\theta = \left\{\theta^{}_{A}, \theta^{}_{B}, \theta^{}_{E}\right\}$ includes Alice's, Bob's, and Eve's parameters, where Bob's parameters are updated to include the after-pulse probabilities:

\begin{align}
	\theta^{}_{B} = \left\{p^{}_{a}, p^{}_{c}, p^{}_{d}, p^{}_{e}\right\}.
\end{align}

Next, we construct the transition matrix across all possible intensity settings $\lambda^{}_{i}$, using the same approach outlined in \Cref{sec:inter-intra-mode}, by combining the individual matrices $\mathbf{T}_{abxe\lambda^{}_{i}}(\theta)$ for each $\lambda^{}_{i}$. If Alice transitions between the different intensities uniformly and independently, the transition matrix simplifies to the following:

\begin{align}
	\mathbf{T}_{abxe}(\theta) = 
\begin{bmatrix}
\mathbf{T}_{abxe\lambda^{}_{1}}(\theta) & \cdots & \mathbf{T}_{abxe\lambda^{}_{N^{}_{\lambda}}}(\theta) \\
\vdots & \ddots & \vdots \\
\mathbf{T}_{abxe\lambda^{}_{1}}(\theta) & \cdots & \mathbf{T}_{abxe\lambda^{}_{N^{}_{\lambda}}}(\theta)
\end{bmatrix} \times \frac{1}{N^{}_{\lambda}}.
\end{align}
\subsubsection*{Eavesdropping Transitions}

To incorporate the possibility of eavesdropping, we expand the matrix $\mathbf{T}_{abxe}(\theta)$ to account for scenarios where Eve transitions between intercepting ($e=1$) and not intercepting ($e=0$), with probabilities $\Delta$ and $1-\Delta$, respectively:

\begin{align}
	\mathbf{T}_{abx}(\theta) = 
\begin{bmatrix}
(1 - \Delta) \mathbf{T}_{abx(e = 0)}(\theta) & \Delta \mathbf{T}_{abx(e = 1)}(\theta) \\
(1 - \Delta) \mathbf{T}_{abx(e = 0)}(\theta) & \Delta \mathbf{T}_{abx(e = 1)}(\theta)
\end{bmatrix}.
\end{align}

\subsubsection*{Bit Choice Transitions}

To model transitions associated with Alice's bit choices, we assume she switches between $x = 0$ and $x = 1$ with equal probability (50\%). The corresponding transition matrix that captures this behavior is given by:

\begin{align}
	\mathbf{T}_{ab}(\theta) = 
\begin{bmatrix}
\mathbf{T}_{ab(x = 0)}(\theta) & \mathbf{T}_{ab(x = 1)}(\theta) \\
\mathbf{T}_{ab(x = 0)}(\theta) & \mathbf{T}_{ab(x = 1)}(\theta)
\end{bmatrix} \times \frac{1}{2}.
\end{align}

\subsubsection*{Basis Choice Transitions}

Similarly, we can construct the transition matrix to account for the basis choices made by Alice ($a$) and Bob ($b$), assuming that both choose their bases independently and uniformly, with equal probability for each choice. First by incorporating Bob's basis transition

\begin{align}
	\mathbf{T}_{a}(\theta) = 
\begin{bmatrix}
\mathbf{T}_{a(b = 0)}(\theta) & \mathbf{T}_{a(b = 1)}(\theta) \\
\mathbf{T}_{a(b = 0)}(\theta) & \mathbf{T}_{a(b = 1)}(\theta)
\end{bmatrix} \times \frac{1}{2},
\end{align}
then incorporating Alice's basis transition
\begin{align}
	\mathbf{T}(\theta) = 
\begin{bmatrix}
\mathbf{T}_{(a = 0)}(\theta) & \mathbf{T}_{(a = 1)}(\theta) \\
\mathbf{T}_{(a = 0)}(\theta) & \mathbf{T}_{(a = 1)}(\theta)
\end{bmatrix} \times \frac{1}{2}.
\end{align}

This final matrix $\mathbf{T}(\theta)$ encapsulates the complete set of transition probabilities, integrating all pulse parameters—basis choices, bit choices, eavesdropping scenarios, and intensity settings—to model the full dynamics of the QKD system.

\subsubsection{Extracting the Probabilities of Click Events}

The probabilities of different click events in an HMM, denoted as $\hat{\mathbf{P}}(\theta)$, are derived by first computing the state probabilities from the transition matrix $\mathbf{T}(\theta)$ and then mapping these state probabilities to observations using the emission matrix $\mathbf{E}$. The stationary distribution, $\mathbf{v}(\theta)$, represents the long-term probabilities of visiting each state and satisfies the following equation:

\begin{align}
	\mathbf{T}^{\top}(\theta) \mathbf{v}(\theta) = \mathbf{v}(\theta).
\end{align}

This is equivalent to finding the eigenvector corresponding to the eigenvalue of 1. Since the transition matrix $\mathbf{T}(\theta)$ is aperiodic (states are not revisited at fixed intervals) and irreducible (every state is reachable from any other), the Perron-Frobenius theorem~\cite{perron1907theorie,frobenius1912theorie} shows that this eigenvector exists, is the maximum eigenvector, and is unique. This property guarantees that the system converges to an equilibrium representing its long-term behavior after sufficient transitions, regardless of the initial state~\cite{norris1998markov}. Furthermore, since we are only interested in the eigenvector corresponding to the maximum eigenvalue, we can leverage more efficient decomposition algorithms, such as the Power Method or Krylov subspace methods~\cite{murphy2022probabilistic,stewart2001krylov}.

By normalizing $\mathbf{v}(\theta)$, we obtain a valid probability distribution over the hidden states:

\begin{align}
	\bar{\mathbf{v}}(\theta) = \frac{\mathbf{v}(\theta)}{\sum^{}_{i} v^{}_{i}(\theta)}.\label{eq:v_norm}
\end{align}

The vector $\bar{\mathbf{v}}(\theta)$ has a length of $144 \times N^{}_{\lambda}$, corresponding to the number of states ($|S|=9$) times the possible choices of $a$, $b$, $x$, and $e$ ($=16$) times the number of intensities ($|\Lambda|=N^{}_{\lambda}$). To correctly project the vector $\bar{\mathbf{v}}(\theta)$ of state probabilities into observable probabilities using the emission matrix $\mathbf{E}$, which is of size $|S| \times |O|$, the vector is first ``folded" into a matrix $\mathbf{V}(\theta)$ of size $|S| \times (16 \times N^{}_{\lambda})$. This folding operation is defined by the mapping function:

\[
\mathcal{F}(\mathbf{v},m) = \mathbf{V}, \quad \text{such that,}
\]

\begin{align}
	\mathcal{F}\left(\mathbf{v}=[v^{}_{1}, v^{}_{2}, \dots, v^{}_{n}],m\right) \rightarrow \mathbf{V} = \begin{bmatrix}
v^{}_{1} & v^{}_{m+1} & \dots & v^{}_{(\frac{n}{m} - 1) \times m + 1} \\
v^{}_{2} & v^{}_{m+2} & \dots & v^{}_{(\frac{n}{m} - 1) \times m + 2} \\
\vdots & \vdots & \ddots & \vdots \\
v^{}_{m} & v^{}_{2m} & \dots & v^{}_{n}
\end{bmatrix},
\end{align}
where each column corresponds to a specific pulse configuration, and each row corresponds to one of the HMM states. We then compute the observable probabilities by projecting this folded matrix onto the emission matrix $\mathbf{E}$:

\begin{align}
\hat{\mathbf{P}}(\theta) = \mathbf{E}^{\top} \bar{\mathbf{V}}(\theta),\label{eq:P_hmm}
\end{align}
where $\bar{\mathbf{V}}(\theta) = \mathcal{F}\left(\bar{\mathbf{v}}(\theta),|S|\right)$ and $\hat{\mathbf{P}}(\theta)$ represents the adjusted probabilities of the observable detection events.

At this point, we have obtained the probabilities of the four possible detection outcomes given any configuration of $(a, b, x, e, \lambda)$. However, Bob only observes $a$, $b$, and $\lambda$. Therefore, we can marginalize over the variables $x$ and $e$ by summing over the corresponding columns in $\hat{\mathbf{P}}(\theta)$.

We can apply this to $x$, $e$ and utilize the symmetry in \Cref{eq-symmetric-bases} to only focus on whether the basis match or did not match. This allows us to obtain $\mathbf{P}(\theta)$, which contains $8 N^{}_{\lambda}$ outcomes (corresponding to $4$ detection events, $2$ cases of matching or non-matching bases, and $N^{}_{\lambda}$ intensity levels).

While this approach provides a complete probabilistic view for analytical purposes, it may not be the most computationally efficient due to the large size of the transition matrix $\mathbf{T}(\theta)$ and the eigenvalue decomposition required. If we are only interested in $\mathbf{P}_{m \lambda}(\theta)$ instead of $\mathbf{P}_{a b x e \lambda}(\theta)$, we can construct a smaller transition matrix directly based on the matching or non-matching cases. This results in a much smaller matrix of size $18 N^{}_{\lambda}$, significantly reducing computational complexity.

\subsection{Error and Gain Probabilities}

To accurately compute the secure-key rate in the QKD protocol, particularly following the Gottesman-Lo-Lütkenhaus-Preskill (GLLP) formula, it is essential to derive the probabilities associated with the gain (probability of a detection events) and erroneous click events (probability of a click at the incorrect detector). Our previous results, encapsulated in the probability vector $\mathbf{P}_{abxe\lambda}(\theta)$, allow us to compute the detection probabilities under various configurations of basis settings, bit choices, intensities, and interception scenarios. By leveraging these configurations, we can rigorously determine the gain and error probabilities required for the secure-key rate calculation, which will be described in the following subsections.

\subsubsection{Marginalizing Over Eve's Interception Flag}

For each combination of basis settings ($a$ and $b$), bit choice ($x$), and intensity ($\lambda$), we first marginalize over the possible states of Eve's interception flag ($e$), where $e = 1$ represents interception and $e = 0$ represents no interception. This marginalization enables us to calculate the effective probabilities $\mathbf{P}_{abx\lambda}(\theta)$ under the influence of Eve's interception rate $\Delta$:

\begin{align}
    \mathbf{P}_{abx\lambda}(\theta) & = \sum^{}_{e\in \left\{0,1\right\}} \mathbf{P}_{abex\lambda}(\theta)P(e),\nonumber\\
    & = \mathbf{P}_{abx(e=0)\lambda}(\theta)(1 - \Delta)+\mathbf{P}_{abx(e=1)\lambda}(\theta)\Delta.
\end{align}

The resulting vector $\mathbf{P}_{abx\lambda}(\theta)$ comprises four elements corresponding to each detection outcome:

\begin{equation}
    \mathbf{P}_{abx\lambda}(\theta) = \left[ \mathbf{P}_{abx\lambda}^{00}(\theta) \quad \mathbf{P}_{abx\lambda}^{01}(\theta) \quad \mathbf{P}_{abx\lambda}^{10}(\theta) \quad \mathbf{P}_{abx\lambda}^{11}(\theta) \right],
\end{equation}
where each element represents the probability of a particular click configuration (e.g., $\mathbf{P}_{abx\lambda}^{01}(\theta)$ represents a click at detector $D^{}_{0}$ only for a specific bit choice ($x$), bases configuration ($a$, $n$) and intensity $\lambda$).

\subsubsection{Gain and Error Probabilities}

The gain probability $Q^{}_{abx\lambda}(\theta)$ for a given bit choice $x$ is obtained by summing the probabilities of the events where only one detector clicks (i.e., single-click events):

\begin{equation}
    Q^{}_{abx\lambda}(\theta) = \mathbf{P}_{abx\lambda}^{01}(\theta) + \mathbf{P}_{abx\lambda}^{10}(\theta).
\end{equation}

The probability of an erroneous click event, $EQ^{}_{abx\lambda}(\theta)$, where the detector registers an error, depends on the bit choice $x$. Specifically:

\begin{align}
    EQ^{}_{ab(x=0)\lambda}(\theta) &= \mathbf{P}_{ab(x=0)\lambda}^{10}(\theta), \quad \text{for} \; x = 0,	\\
    EQ^{}_{ab(x=1)\lambda}(\theta) &= \mathbf{P}_{ab(x=1)\lambda}^{01}(\theta), \quad \text{for} \; x = 1.
\end{align}

\subsubsection{Marginalizing Over Bit Choice}

To obtain the total click and error probabilities for a given basis configuration and intensity, we marginalize over the bit choice $x$, assuming Alice's bit choices are uniformly distributed. This leads to the overall gain $Q^{}_{ab\lambda}(\theta)$ and error probabilities $EQ^{}_{ab\lambda}(\theta)$ given by:

\begin{align}
	Q^{}_{ab\lambda}(\theta) & = \frac{1}{2} \left( Q^{}_{ab(x=0)\lambda}(\theta) + Q^{}_{ab(x=1)\lambda}(\theta) \right),\\
    EQ^{}_{ab\lambda}(\theta) & = \frac{1}{2} \left( EQ^{}_{ab(x=0)\lambda}(\theta) + EQ^{}_{ab(x=1)\lambda}(\theta) \right).
\end{align}

These expressions represent the probabilities of a signal and an erroneous click under the conditions specified by $\theta$. To further simplify notation, we consider the two distinct cases of basis alignment, where $a = b$ and $a \neq b$. Referring to the symmetry from \Cref{eq-symmetric-bases}, where $\mathbf{P}_{ab\lambda}(\theta) = \mathbf{P}_{ba\lambda}(\theta)$, we redefine the gain and error probabilities in terms of alignment cases:

\begin{align}
	Q^{}_{m\lambda}(\theta) &= Q^{}_{(a=m)(b=1)\lambda}(\theta),\\
    EQ^{}_{m\lambda}(\theta) &= EQ^{}_{(a=m)(b=1)\lambda}(\theta),
\end{align}
where $m = 1$ denotes matching bases ($a = b$) and $m = 0$ denotes non-matching bases ($a \neq b$).

\subsubsection{Distribution of Gain and Error Counts}

The number of gain clicks, $G^{}_{m\lambda}$, and error clicks, $R^{}_{m\lambda}$, for a specific bases alignment $m$ and intensity $\lambda$ are distributed as follows (assuming i.i.d.):

\begin{align}
    G^{}_{m\lambda} &\sim \text{Binomial}\left( N^{}_{m\lambda}, \, Q^{}_{m\lambda}(\theta) \right), \\
    R^{}_{m\lambda} &\sim \text{Binomial}\left( N^{}_{m\lambda}, \, EQ^{}_{m\lambda}(\theta) \right),\label{eq-r_N_E}
\end{align}
where $N^{}_{m\lambda}$ is the total number of pulses with basis alignment $m$ and intensity $\lambda$. Here, $G^{}_{m\lambda}$ and $R^{}_{m\lambda}$ are not independent, as the number of error clicks is a subset of the number of gain clicks.

The number of error clicks $R^{}_{m\lambda}$ can be viewed as resulting from sequential binomial sampling. First, the gain clicks $G^{}_{m\lambda}$ are selected from $N^{}_{m\lambda}$ with probability $Q^{}_{m\lambda}(\theta)$, then $R^{}_{m\lambda}$ is selected from $G^{}_{m\lambda}$ with probability $\delta^{}_{m\lambda}(\theta)$ (the conditional probability of an error given a signal). More formally,

\begin{equation}
    R^{}_{m\lambda} \sim \text{Binomial}\left( G^{}_{m\lambda}, \, \delta^{}_{m\lambda}(\theta) \right).\label{eq-r_s_delta}
\end{equation}

To derive the conditional error probability $\delta^{}_{m\lambda}(\theta)$, we apply \Cref{lem:fiber_beamsplitter}, which can be rephrased in terms of two sequential binomial sampling processes. If we first sample $s \sim \text{Binomial}(N, p)$, and then sample $r \sim \text{Binomial}(s, q)$, then $r \sim \text{Binomial}(N, p\cdot q)$

This lemma shows that the effective success probability for sequential binomial sampling is the product of the individual probabilities in each stage. Applying this to $R^{}_{m\lambda}$, we have:

\begin{equation}
    R^{}_{m\lambda} \sim \text{Binomial}\left( N^{}_{m\lambda}, \, Q^{}_{m\lambda}(\theta) \, \delta^{}_{m\lambda}(\theta) \right).
\end{equation}

Comparing this with \Cref{eq-r_N_E}, we see that:

\begin{align}
    Q^{}_{m\lambda}(\theta) \, \delta^{}_{m\lambda}(\theta) &= EQ^{}_{m\lambda}(\theta), \\
    \delta^{}_{m\lambda}(\theta) &= \frac{EQ^{}_{m\lambda}(\theta)}{Q^{}_{m\lambda}(\theta)}.\label{eq-delta_E_Q}
\end{align}

We can apply the same principle of sequential binomial sampling to express $G^{}_{m\lambda}$ and $R^{}_{m\lambda}$ distributions in terms of $N$, the number of session pulses, directly. Since $N^{}_{m\lambda}$ results from first selecting $N^{}_{m}$ with probability $P(m) = 1/2$ (the probability of bases alignment), then selecting $N^{}_{m\lambda}$ with probability $P(\lambda) = 1/N^{}_{\lambda}$ (assuming uniform intensity selection), we can also express them as:

\begin{align}
    G^{}_{m\lambda} &\sim \text{Binomial}\left( N, \, P(m) P(\lambda) Q^{}_{m\lambda}(\theta) \right), \label{eq-G_m_lambda_E}\\
    R^{}_{m\lambda} &\sim \text{Binomial}\left( N, \, P(m) P(\lambda) EQ^{}_{m\lambda}(\theta) \right).\label{eq-R_m_lambda_E}
\end{align}

While $\delta^{}_{m\lambda}(\theta)$ and $Q^{}_{m\lambda}(\theta)$ are sufficient for secure-key rate calculations, we also derive $\mathbb{E}[\delta^{}_{m\lambda}(\theta)]$ and $\mathbb{V}[\delta^{}_{m\lambda}(\theta)]$ to facilitate a robust comparison between theoretical expectations and simulated data. The following lemma provides these values, establishing a basis for evaluating the alignment of our model with empirical results.

\begin{theorem}
\label{theorem:rho_E_V}
Let $r$ be a random variable obtained through a two-step sampling process where first $s \sim \text{Binomial}(N, p)$ and then, given $s$, $r \mid s \sim \text{Binomial}(s, q)$. Define $\rho = \frac{r}{s}$ when $s > 0$. Then, the expected value and variance of $\rho$ are given by:

\begin{align}
	\mathbb{E}[\rho] & = q,\\
\mathbb{V}[\rho] & = q(1 - q) \, \mathbb{E}\left[\frac{1}{s}\right],
\end{align}
where $\mathbb{E}\left[\frac{1}{s}\right]$ is the expected value of the reciprocal of $s$.
\end{theorem}

\begin{proof}
First, compute the expected value $\mathbb{E}[\rho]$ using the law of total expectation~\cite{murphy2022probabilistic}:
\[
\mathbb{E}[\rho] = \mathbb{E}\left[\mathbb{E}\left[\rho \mid s\right]\right].
\]
Given $s$, since $r \mid s \sim \text{Binomial}(s, q)$ and $\rho = \frac{r}{s}$, we have:

\[
\begin{aligned}
\mathbb{E}\left[\rho \mid s\right] =  \mathbb{E}\left[\frac{r}{s} \mid s\right] = \frac{1}{s} \mathbb{E}\left[r \mid s\right]  = \frac{1}{s} (s q)  = q,\\
\end{aligned}
\]
therefore,
\[
\mathbb{E}[\rho] = \mathbb{E}[q] = q.
\]

The variance $\mathbb{V}[\rho]$ can be expressed using the law of total variance~\cite{murphy2022probabilistic}:
\[
\mathbb{V}[\rho] = \mathbb{E}\left[\mathbb{V}\left[\rho \mid s\right]\right] + \mathbb{V}\left[\mathbb{E}\left[\rho \mid s\right]\right].
\]
Given $s$, the conditional variance is:
\[
\begin{aligned}
	\mathbb{V}\left[\rho \mid s\right] = \mathbb{V}\left[\frac{r}{s} \mid s\right]  = \frac{1}{s^2} \mathbb{V}\left[r \mid s\right]  = \frac{1}{s^2} (s q (1 - q))  = \frac{q (1 - q)}{s}.
\end{aligned}
\]
Thus,
\[
\mathbb{E}\left[\mathbb{V}\left[\rho \mid s\right]\right] = q (1 - q) \, \mathbb{E}\left[\frac{1}{s}\right].
\]
Since $\mathbb{E}\left[\rho \mid s\right] = q$ is constant,
\[
\mathbb{V}\left[\mathbb{E}\left[\rho \mid s\right]\right] = \mathbb{V}[q] = 0.
\]
Combining these results, we obtain:
\[
\mathbb{V}[\rho] = q (1 - q) \, \mathbb{E}\left[\frac{1}{s}\right].
\]
\end{proof}

Note that the derivation in \Cref{{theorem:rho_E_V}} assumes $s > 0$, ensuring that $\rho = \frac{r}{s}$ is well-defined. For clarity and simplicity, this condition has been omitted from the main expressions, but all expectations and variances are implicitly conditioned on $s > 0$. Additionally, $\mathbb{E}\left[\frac{1}{s}\right]$ does not have a closed-form expression for $s \sim \text{Binomial}(N, p)$; however, for large $Np$, it can be approximated as:
\[
\mathbb{E}\left[\frac{1}{s}\right] \approx \frac{1}{\mathbb{E}\left[s\right]} = \frac{1}{Np}.
\]
This leads to an approximate variance:
\[
\mathbb{V}[\rho] \approx \frac{q(1 - q)}{Np}.
\]
This approximation holds well when $N$ is large and $p$ is not too small, rendering the probability $P(s = 0)$ negligible.

Applying this result to \Cref{eq-r_s_delta} and \Cref{eq-G_m_lambda_E}, we obtain the expected value and variance of the error rate $\rho$ for a basis alignment $m$ and intensity $\lambda$ as follows:
\begin{align}
	\mathbb{E}\left[\rho\right]&=\delta^{}_{m\lambda}(\theta),\label{eq-E_rho}\\
	\mathbb{V}\left[\rho\right]&\approx \frac{\delta^{}_{m\lambda}(\theta)(1-\delta^{}_{m\lambda}(\theta))}{N Q^{}_{m\lambda}(\theta)P(m)P(\lambda)}.\label{eq-V_rho}
\end{align}

Since $\rho \in \left[0,1\right]$, it is natural to model it with a Beta distribution. Using \Cref{eq-beta_parameters}, we derive parameters for a Beta distribution that match the expected value and variance in \Cref{eq-E_rho} and \Cref{eq-V_rho}. This allows us to \textit{approximate} the distribution of the error rate as follows:
\begin{align}
	P(\rho) \approx \text{Beta}\left(\rho\mid \alpha^{}_{m\lambda}(\theta), \beta^{}_{m\lambda}(\theta)\right),\label{eq-beta_rho}
\end{align}
where $\alpha^{}_{m\lambda}(\theta)$ and $\beta^{}_{m\lambda}(\theta)$ are calculated to ensure that the Beta distribution’s mean and variance align with \Cref{eq-E_rho} and \Cref{eq-V_rho}.

\section{Bayesian Inference}
\label{sec:bayesian-inference}
In this section, we aim to infer the unknown parameters of the QKD system based on the observed detection events. Recall that the parameter set $\theta = \left\{\theta^{}_{A}, \theta^{}_{B}, \theta^{}_{E}\right\}$ encompasses all aspects of the QKD system, where $\theta^{}_{A}$ represents Alice's parameters, $\theta^{}_{B}$ represents Bob's parameters, and $\theta^{}_{E}$ captures Eve's potential strategies. Let $C$ be the vector of counts for each detection outcome:

\begin{align}
    C^{}_{m} &= \left[C^{00}_{m}\quad C^{01}_{m}\quad C^{10}_{m}\quad C^{11}_{m}\right],\nonumber\\
    C &= \left[C^{}_{0}\quad C^{}_{1}\right],
\end{align}
where $C^{01}_{0}$, for example, represents the count of events where $D^{}_{0}=1$, $D^{}_{1}=0$, and the bases did not match.

The objective is to infer the unknown parameters associated with Eve, specifically $\theta^{}_{E} = \{d^{}_{AE}, p^{}_{EB}, k, \Delta\}$, based on observed data $C$ and the known parameters $\theta^{}_{A}$ and $\theta^{}_{B}$. Utilizing Bayes' theorem, the posterior distribution of Eve's parameters, given the observed data and known parameters, is formulated as:

\begin{align}
	\text{Posterior}(\theta^{}_{E} \mid C,N,\theta^{}_{A}, \theta^{}_{B}, \theta^{}_{P}) &= \frac{\mathcal{L}(C \mid N,\theta^{}_{A}, \theta^{}_{B}, \theta^{}_{E}) \mathcal{P}(\theta^{}_{E} \mid \theta^{}_{P})}{\mathcal{M}(C \mid N,\theta^{}_{A}, \theta^{}_{B}, \theta^{}_{P})},
\end{align}
where
\begin{align}
	\mathcal{M}(C \mid N,\theta^{}_{A}, \theta^{}_{B},\theta^{}_{P}) =\int \mathcal{L}(C \mid N,\theta^{}_{A}, \theta^{}_{B}, \theta^{}_{E})  \mathcal{P}(\theta^{}_{E} \mid \theta^{}_{P}) \, d\theta^{}_{E},\label{eq:marginal}
\end{align}
where $\mathcal{L}(C \mid N,\theta^{}_{A}, \theta^{}_{B}, \theta^{}_{E})$ is the likelihood of the data given the parameters, and $\mathcal{P}(\theta^{}_{E} \mid \theta^{}_{P})$ is the prior distribution of Eve's parameters, with $\theta^{}_{P}$ representing the hyper-parameters of these priors, and $\mathcal{M}(C \mid N,\theta^{}_{A}, \theta^{}_{B},\theta^{}_{P})$ is the marginal likelihood (or evidence) that normalizes the posterior distribution.

\subsection{Constructing the Likelihood}

To apply Bayesian inference, we first need to construct the likelihood function that describes the probability of observing the data given the parameters of the system. This formally presented in the following theorem:

\begin{theorem}
\label{theorem:likelihood_iid}
Let Alice emits Poisson-distributed number of photons with average intensities selected uniformly at random from a set of $N^{}_{\lambda}$ intensities $\{\lambda^{}_{1}, \lambda^{}_{2}, \ldots, \lambda^{}_{N^{}_{\lambda}}\}$. For each pulse, Alice randomly selects a basis $a \in \{0, 1\}$ with equal probability, and Bob independently measures in a randomly chosen basis $b \in \{0, 1\}$, also with equal probability. 

Let Eve be positioned at a distance $d^{}_{AE}$ from Alice, potentially intercepting a fraction $\Delta$ of the total transmission and capturing $k$ photons from each intercepted pulse. Bob is positioned at a distance $d^{}_{AB}$ from Alice, with fiber attenuation $\alpha$. Bob uses two detectors characterized by detection efficiencies $p^{}_{c^{}_{0}}$ and $p^{}_{c^{}_{1}}$, dark count probabilities $p^{}_{d^{}_{0}}$ and $p^{}_{d^{}_{1}}$, and a beam splitter with a misalignment error $p^{}_{e}$. The detection events are assumed to be independent.

Under these conditions, the vector of counts for the $8N^{}_{\lambda}$ distinct detection outcomes—corresponding to each combination of intensity, basis choices, and detection results—is distributed according to a multinomial distribution:

\begin{align}
	P(C\mid N, \theta) = \operatorname{Multinomial}\left(N, \, \mathbf{P}(\theta)\right),
\end{align}
where $\mathbf{P}(\theta)$ is the vector of probabilities for each detection outcome, derived from the set of system parameters $\theta = \{\theta^{}_{A}, \theta^{}_{B}, \theta^{}_{E}\}$ as in \Cref{eq:P_iid}.
\end{theorem}

\subsubsection{HMM Likelihood}

The likelihood of observing a specific sequence of detection events $\mathbf{O}=\{O^{}_{t}\}_{t=1}^N$ given state transition probabilities $\mathbf{T}$, emission probabilities $\mathbf{E}$, and initial state probabilities $\bm{\pi}$ can be computed using the forward algorithm~\cite{bishop2006pattern}:

\begin{align}
L(\mathbf{O} \mid \mathbf{T}, \mathbf{E},\bm{\pi}) &= \sum^{N^{}_{s}}_{i=1}\alpha^{}_{N}(i),
\end{align}
where $N^{}_{s}$ is the number of states and the forward probability at time $t$ for state $i$, $\alpha^{}_{t}(i)$, is recursively expressed as:
\begin{align}
	\alpha^{}_{t}(i) = \mathbf{E}_{i,O^{}_{t}} \cdot \sum^{N^{}_{s}}_{j=1} \mathbf{T}_{j,i} \cdot \alpha^{}_{t-1}(j),
\end{align}
and the initial forward probability
\begin{align}
\alpha^{}_{1}(i) = \bm{\pi}_{i} \mathbf{E}_{i,O^{}_{1}}.
\end{align}

Computing the forward algorithm is computationally intensive and impractical for large $N$. Additionally, our primary interest lies not in the likelihood of a specific sequence of detection events but in the aggregated counts of detection outcomes over the session.

To address this, we can use the long term stationary distribution $\hat{\mathbf{P}}(\theta)$ as the parameter for the multinomial distribution. Although individual detection events are dependent, the aggregated counts over a large number of trials can be modeled using a multinomial distribution with these adjusted probabilities. 

\begin{theorem}
\label{theorem:likelihood_hmm}
Given a hidden Markov model (HMM) with an irreducible and an  aperiodic state transition matrix $\mathbf{T}(\theta)$, emission probabilities $\mathbf{E}$, and initial state probabilities $\bm{\pi}$, where the system accounts for dependencies induced by after-pulsing, the likelihood function for the aggregated counts $C$ over $N$ trials is formulated as:
\begin{align}
P(C \mid N, \theta) = \text{Multinomial}(C \mid N, \hat{\mathbf{P}}(\theta)),
\end{align}
where $\hat{\mathbf{P}}(\theta) = \mathbf{E}^{\top} \bar{\mathbf{v}}(\theta)$ represents the adjusted probabilities derived from the stationary distribution $\bar{\mathbf{v}}(\theta)$ of the Markov chain as in \Cref{eq:v_norm}.
\end{theorem}

By replacing the i.i.d. probabilities with the adjusted probabilities from the stationary distribution, we effectively account for the dependencies introduced by after-pulsing. This allows us to use the multinomial likelihood function for the aggregated counts, enabling accurate and computationally efficient inference of the system parameters.

\subsection{Constructing the Prior}

To construct a robust Bayesian inference framework, it is essential to carefully choose priors for the unknown parameters $\theta^{}_{E}$. The parameters $\theta^{}_{E} = \{d^{}_{AE}, p^{}_{EB}, k, \Delta\}$ represent Eve's distance from Alice, the channel efficiency chosen by Eve, the number of photons intercepted per pulse, and the fraction of intercepted pulses, respectively. 

\subsubsection{Assumptions and Independence}

To address the uncertainty surrounding Eve’s strategy, we encode a bias towards worst-case scenarios while incorporating practical constraints. We assume that Eve’s choices are independent of Alice and Bob’s configurations, allowing her to select parameters freely from the allowable domains. Although a strategic Eve might optimize her settings to maximize information gain based on Alice’s and Bob’s choices, modeling such behavior would introduce another exploitable vulnerability. Instead, we take a precautious approach, treating Eve’s parameters as independent, ensuring that she cannot exploit any alignment with Alice and Bob’s configurations.

\subsubsection{Priors for Bounded Parameters}

The parameters $d^{}_{AE}$, $p^{}_{EB}$, and $\Delta$ are all bounded within specific ranges:
\begin{align}
    0 \leq  d^{}_{AE} &\leq d^{}_{AB},  \\
    0 \leq  p^{}_{EB} &\leq 1,  \\
    0 \leq ~\,\Delta~\, &\leq 1. 
\end{align}

Let $\theta^{}_{b}$ a bounded parameter in $\theta^{}_{E}$. For such parameter, we utilize the Beta distribution due to its flexibility in modeling probabilities over a finite range after rescaling the parameter to the $[0,1]$ range. Specifically, the prior for $\theta^{}_{b}$:

\begin{align}
	P(\theta^{}_{b} \mid \vartheta^{}_{b}) = \text{Beta} \left( \frac{\theta^{}_{b}-l^{}_{b}}{u^{}_{b}-l^{}_{b}} \; \middle| \; \alpha^{}_{b}, \beta^{}_{b} \right) \frac{1}{u^{}_{b}-l^{}_{b}},
\end{align} 
where $\vartheta^{}_{b} = \{\alpha^{}_{b}, \beta^{}_{b},u^{}_{b},l^{}_{b}\}$ are the hyper-parameters of the prior distribution, namely the shape parameters of the Beta distribution ($\alpha^{}_{b}$ and $\beta^{}_{b}$), as well as the upper and lower bound of $\theta^{}_{b}$ ($u^{}_{b}$ and $l^{}_{b}$). The scaling factor $\frac{1}{u^{}_{b}-l^{}_{b}}$ ensures proper normalization over the interval $[l^{}_{b}, u^{}_{b}]$.

\subsubsection{Prior for Semi-Bounded Parameters}

The parameter $k$, representing the number of photons intercepted per pulse, is semi-bounded with $k \geq 1$. In \Cref{theorem:PNS} we utilized the incomplete gamma function to extend the domain of $k$ to the positive real numbers. This extension allows us to treat $k$ as a continuous variable, which is particularly useful for gradient-based optimization or numerical integration techniques and more nuanced probabilistic modeling.

Given this extended domain, we employ a Gamma distribution for $k$, which is suitable for modeling continuous variables that are bounded below. Although in this specific setting, there is only one semi-bounded parameter, i.e. $k$, in preparation for a fully Bayesian approach where we can treat other fixed variables as random variables, we will denote a semi-bounded parameter modeled with a Gamma prior as $\theta^{}_{g}$:

\begin{align}
	P(\theta^{}_{g} \mid \vartheta^{}_{g}) = \text{Gamma}(\theta^{}_{g} + l^{}_{g} \mid \alpha^{}_{g}, \beta^{}_{g}),
\end{align}
where $\vartheta^{}_{g} = \{\alpha^{}_{g}, \beta^{}_{g}, u^{}_{g}, l^{}_{g}\}$ are the hyper-parameters of the Gamma distribution. In the case of $k$,The distribution is shifted by $l^{}_{g}=1$ to ensure alignment with the lower bound $k = 1$, reflecting that if Eve intercepts, she must capture at least one photon.

\subsubsection{Selection of Prior Parameters}

For the bounded parameters, setting $\alpha = \beta = 1$ in the Beta distribution yields a uniform distribution, reflecting maximal uncertainty within the given bounds. We apply this uniform distribution to parameters like $p^{}_{EB}$, encoding no prior preference over Eve’s channel efficiency. However, to introduce a cautious bias against Eve's activity, we set $\alpha^{}_{\Delta} = 2$ and $\beta^{}_{\Delta} = 1$, making $\Delta = 1$ more likely a priori. This choice implies a precautionary stance, assuming Eve's presence unless data suggests otherwise.

For the distance parameter $d^{}_{AB}$, we set the prior parameters in the Beta distribution oppositely, with $\alpha^{}_{d^{}_{AB}} = 1$ and $\beta^{}_{d^{}_{AB}} = 2$. This encodes the assumption that Eve is more likely to position herself closer to Alice, where the signal intensity is higher, minimizing the efficiency requirements for her alternative channel when intercepting photons.

For the semi-bounded parameter $k$, setting $\alpha^{}_{k} = 1$ in the Gamma distribution places the mode of the distribution at $k = 1$. This choice encodes the assumption that, in the absence of data, Eve is most likely to capture the minimum number of photons from an intercepted pulse, which is 1. To provide practical constraints within this framework, we define a useful \textit{pseudo} upper bound on $k$:

\begin{itemize}
    \item $k^{}_{\max}$: The maximum number of photons Eve can intercept from a pulse immediately emitted by Alice such that she would require a channel efficiency of $p^{}_{EB} = 1$ to avoid detection.
\end{itemize}

This \textit{pseudo} upper bound, $k^{}_{\max}$, serves as a guideline for setting the rate parameter $\beta^{}_{k}$ of the Gamma distribution. Specifically, $\beta^{}_{k}$ is chosen such that the expected value lies halfway between 1 and $k^{}_{\max}$, ensuring that the prior realistically constrains Eve's capabilities while remaining theoretically unbounded.

\subsubsection{Final Formulation of the Prior}

Combining the priors for each parameter, the joint prior distribution for Eve’s parameter $\theta^{}_{E}$ is given by:

\begin{align}
	\mathcal{P}(\theta^{}_{E} \mid \theta^{}_{P}) = &~\text{Beta} \left( \frac{d^{}_{AE}}{d^{}_{AB}} \; \middle| \; \alpha^{}_{d^{}_{AE}}, \beta^{}_{d^{}_{AE}} \right) \frac{1}{d^{}_{AB}} \times \nonumber \\
	&~\text{Beta} \left( p^{}_{EB} \mid \alpha^{}_{p^{}_{EB}}, \beta^{}_{p^{}_{EB}} \right) \times \nonumber \\
	&~\text{Beta} \left( \Delta \mid \alpha^{}_{\Delta}, \beta^{}_{\Delta} \right) \times \nonumber \\
	&~\text{Gamma} \left( k + 1 \mid \alpha^{}_{k}, \beta^{}_{k} \right),\label{eq:prior}
\end{align}
and

\begin{align}
	\theta^{}_{P}=&~\{\bm{\alpha}, \bm{\beta}, \mathbf{u}, \mathbf{l}\},\\
	\bm{\alpha} =&~\{\alpha^{}_{d^{}_{AE}}, \alpha^{}_{p^{}_{EB}}, \alpha^{}_{k}, \alpha^{}_{\Delta}\},\\
	\bm{\beta}=&~\{\beta^{}_{d^{}_{AE}}, \beta^{}_{p^{}_{EB}}, \beta^{}_{k}, \beta^{}_{\Delta}\},\\
	\mathbf{u}=&~\{d^{}_{AB}, 1, \infty, 1\},\\
	\mathbf{l}=&~\{0, 0, 1, 0\},
\end{align}
where $\mathbf{u}$ and $\mathbf{l}$ are the upper and lower bounds of the parameters. We store these in the set of priors' hyper-parameters as they will be used later to un-bound the parameters when constructing the posterior.

\subsection{Constructing the Posterior}

With the likelihood $\mathcal{L}(C \mid N, \theta^{}_{A}, \theta^{}_{B}, \theta^{}_{E})$ and prior $\mathcal{P}(\theta^{}_{E} \mid \theta^{}_{P})$ established, the final step is to compute the marginal likelihood $\mathcal{M}(C \mid N, \theta^{}_{A}, \theta^{}_{B}, \theta^{}_{P})$, as defined in \Cref{eq:marginal}, to derive the posterior distribution of Eve's parameters $\theta^{}_{E}$ conditioned on the system parameters and observed data.

However, due to the complexity of the likelihood function in our model, stemming from the multiple parameters and the interactions modeled within the QKD system, analytically computing this marginal is infeasible. Additionally, an analytical approach might restrict our choice of priors to conjugate priors, which could fail to accurately represent the underlying physics and real-world constraints of the problem.

\subsubsection{Methods for Approximating the Posterior}

Given the challenges in deriving the posterior analytically, we consider several computational methods. One common approach is \emph{Markov Chain Monte Carlo} (MCMC)\cite{metropolis1953equation,hastings1970monte}, which can handle complex, high-dimensional distributions. MCMC methods, such as Metropolis-Hastings\cite{metropolis1953equation,hastings1970monte} or Gibbs sampling~\cite{geman1984stochastic}, do not require the posterior to have a specific form, allowing flexible modeling with realistic priors that reflect the physical constraints of the QKD system. Although MCMC converges to the true posterior with sufficient samples, it can be computationally intensive and requires careful tuning and convergence checks~\cite{cowles1996markov, robert2004monte}.

\emph{Variational Bayes} (VB)~\cite{murphy2022probabilistic} approximates the posterior by optimizing a simpler, parameterized distribution to minimize KL divergence. This can be more efficient than MCMC, especially with large datasets, but yields only an approximate posterior that depends on the chosen family of distributions. In security-critical QKD applications, where assumptions about the posterior can introduce vulnerabilities, we prefer methods that let the data define the distribution.

\emph{Bayesian Quadrature} could be used for low-dimensional parameter spaces by modeling the integrand as a Gaussian process~\cite{osborne2012active,gunter2014sampling}. Although it may require fewer evaluations than traditional quadrature, it assumes integrand smoothness, which might miss important features in security-critical scenarios. Consequently, like VB, Bayesian Quadrature is less suitable when avoiding assumptions is essential for security.

Given the complexity of our parameter space, we also considered gradient-based numerical integration methods like \emph{Hamiltonian Monte Carlo} (HMC)\cite{duane1987hybrid, neal2011mcmc} and its adaptive variant, the \emph{No-U-Turn Sampler} (NUTS)\cite{hoffman2014no}. We computed gradients for all variables via the incomplete gamma function, enabling a continuous likelihood function and methods that exploit gradient information to navigate the parameter space more efficiently~\cite{betancourt2017conceptual}.

Ultimately, we chose \emph{Covariance-Adaptive Slice Sampling}\cite{thompson2010covariance}, which incorporates gradient-based adaptation while retaining the robustness of slice sampling. With fewer hyper-parameters, it minimizes manual tuning and enhances reliability while efficiently exploring high-dimensional parameter spaces by adapting to local covariance structure. Our fully Bayesian framework—where any parameter can be switched from fixed to random by assigning a prior—maintains security guarantees in QKD systems by avoiding restrictive modeling assumptions\cite{gelman2013bayesian}.

\subsection{From Bounded to Unbounded Variables}
To perform effective Bayesian inference, especially using gradient-based numerical integration methods, it is crucial to transform the parameter space so that all variables are unbounded. However, this process is not as straightforward as simply applying a transformation step during posterior computation. Transformations alter the probability distribution non-uniformly, effectively mapping the original distribution into a different space. To ensure that sampling remains consistent with the true posterior, the probability density function (PDF) must be adjusted to account for these changes. This adjustment is achieved through the transform of variables method, which ensures an accurate and mathematically consistent mapping between the original and transformed spaces.

\subsubsection{Defining Transformations and Their Inverses}

To transform the parameters $\theta^{}_{E} = \{d^{}_{AE}, p^{}_{EB}, k, \Delta\}$ into an unbounded space and vice versa, we use the sigmoid, $\sigma(x)$, and logit, $\sigma^{-1}(x)$ transformations for the bounded parameters, and the exponential, $e^{x}$, and logarithmic, $\log(x)$, for the semi-bounded.
\begin{align}
	   \sigma(x) = \frac{1}{1 + e^{-x}}, \quad \sigma^{-1}(x) = \log\left(\frac{x}{1 - x}\right).
\end{align}
   
Let $b$ be the set of bounded parameters modeled with a Beta prior, and $g$ is the set of semi-bounded parameters modeled with a Gamma prior, we define the width $w^{}_{b} = u^{}_{b}-l^{}_{b}$, where $u^{}_{b}$ and $l^{}_{b}$ are the parameter's upper and lower bounds respectively. The transformation to the unbounded space and its reverse are given by:

\begin{align}
	\phi^{}_{b} = \sigma^{-1}\left(\frac{\theta^{}_{b} - l^{}_{b}}{w^{}_{b}}\right),\quad \theta^{}_{b} = \sigma(\phi^{}_{b})   w^{}_{b} + l^{}_{b},
\end{align}
and for the semi-bounded variables

\begin{align}
	\phi^{}_{g} = \log(\theta^{}_{g} - l^{}_{g}),\quad \theta^{}_{g} = e^{\phi^{}_{g}} + l^{}_{g}.
\end{align}

More generally, let:
\begin{align}
	\Phi(\theta^{}_{E}, \theta^{}_{P}) &= \phi^{}_{E},\\
	\Theta(\phi^{}_{E}, \theta^{}_{P}) &= \theta^{}_{E},
\end{align}
be the two mapping functions that transform the variables from and to different spaces. By applying these transformations, we define the new set of transformed variables $\phi^{}_{E} = \{\phi^{}_{d^{}_{AE}}, \phi^{}_{p^{}_{EB}}, \phi^{}_{k}, \phi^{}_{\Delta}\}$, where each $\phi$ now lies in the real, unbounded space $\mathbb{R}$.

\subsubsection{Jacobian of the Transformation}

When transforming variables in a probability distribution, it is essential not only to transform the variables but also to update the probability density function (PDF) accordingly. This is because the original PDF is defined over the original variable space, and a transformation changes the volume element in this space. The transformation can stretch, compress, or otherwise distort the space, altering the relative spacing between points. 

To ensure that the transformed PDF correctly reflects the distribution of probabilities over the new variable space, we must incorporate the Jacobian determinant of the transformation. The Jacobian determinant represents the factor by which the volume element changes under the transformation. Mathematically, if we have a transformation $y = g(x)$ and the original PDF is $f^{}_{X}(x)$, the transformed PDF $f^{}_{Y}(y)$ is given by~\cite{bishop2006pattern}:

\begin{align}
	f^{}_{Y}(y) = f^{}_{X}(g^{-1}(y)) \left| \frac{d}{dy} g^{-1}(y) \right|.
\end{align}

This ensures that the total probability remains normalized to one after the transformation, preserving the properties of a probability distribution. For multidimensional distributions, it generalizes to finding the determinant of the Jacobian matrix $\mathbf{J}(\phi^{}_{E}, \theta^{}_{P})$ of the transformation from $\phi^{}_{E}$ back to $\theta^{}_{E}$:

\begin{align}
	\mathbf{J}(\phi^{}_{E}, \theta^{}_{P}) = 
\begin{bmatrix}
    \frac{\partial \theta^{}_{d^{}_{AE}}}{\partial \phi^{}_{d^{}_{AE}}} & 0 & 0 & 0 \\
    0 & \frac{\partial \theta^{}_{p^{}_{EB}}}{\partial \phi^{}_{p^{}_{EB}}} & 0 & 0 \\
    0 & 0 & \frac{\partial \theta^{}_{k}}{\partial \phi^{}_{k}} & 0 \\
    0 & 0 & 0 & \frac{\partial \theta^{}_{\Delta}}{\partial \phi^{}_{\Delta}}
\end{bmatrix}.
\end{align}

The individual partial derivatives, using the inverse functions of the transformations, are:

\begin{align}
\frac{\partial \theta^{}_{b}}{\partial \phi^{}_{b}} &= \frac{\partial}{\partial \phi^{}_{b}} \left( \sigma(\phi^{}_{b})   w^{}_{b} + l^{}_{b} \right) =    \sigma(\phi^{}_{b}) (1 - \sigma(\phi^{}_{b}))w^{}_{b}, \\
\frac{\partial \theta^{}_{g}}{\partial \phi^{}_{g}} &= \frac{\partial}{\partial \phi^{}_{g}} \left( e^{\phi^{}_{g}} + l^{}_{g} \right) = e^{\phi^{}_{g}}.
\end{align}

Considering that $w^{}_{p^{}_{EB}}=w^{}_{\Delta}=1$, and $w^{}_{d^{}_{AE}}=d^{}_{AB}$, the determinant of the Jacobian matrix, $\left |\mathbf{J}(\phi^{}_{E}, \theta^{}_{P})\right|$, can be expressed as:
\begin{align}
	\mathcal{J}(\phi^{}_{E}\mid \theta^{}_{P}) & = \left |\mathbf{J}(\phi^{}_{E}, \theta^{}_{P})\right|\nonumber\\
	& = \sigma(\phi^{}_{d^{}_{AE}}) (1 - \sigma(\phi^{}_{d^{}_{AE}}))d^{}_{AE}\times\sigma(\phi^{}_{p^{}_{EB}}) (1 - \sigma(\phi^{}_{p^{}_{EB}}))\times\sigma(\phi^{}_{\Delta}) (1 - \sigma(\phi^{}_{\Delta}))\times e^{\phi^{}_{k}}.\label{eq:jacobian}
\end{align}

\subsubsection{Updating the Posterior with the Jacobian}

The posterior distribution with respect to the transformed variables $\phi^{}_{E}$ is given by incorporating the Jacobian determinant to adjust for the change of variables:

\begin{align}
\text{Posterior}(\phi^{}_{E} \mid C, N, \theta^{}_{A}, \theta^{}_{B}, \theta^{}_{P}) \propto \mathcal{L}(C\mid N,\theta^{}_{A}, \theta^{}_{B}, \Theta(\phi^{}_{E}))\times \mathcal{P}(\Theta(\phi^{}_{E})\mid \theta^{}_{P})\times \mathcal{J}(\phi^{}_{E}\mid \theta^{}_{P}),\label{eq:posterior}
\end{align}
where $\mathcal{L}(C\mid N,\theta^{}_{A}, \theta^{}_{B}, \Theta(\phi^{}_{E}))=\text{Multinomial}\left(C\mid N,\mathbf{P}\left(\left\{\theta^{}_{A}, \theta^{}_{B}, \Theta(\phi^{}_{E})\right\}\right)\right)$ is the likelihood as derived in \Cref{theorem:likelihood_iid} and $\mathcal{P}(\Theta(\phi^{}_{E})\mid \theta^{}_{P})$ is the prior as in \Cref{eq:prior} but with $\theta^{}_{E}=\Theta(\phi^{}_{E})$.

\subsection{From Fixed to Random Variables}

Until now, our modeling approach has operated under the assumption that the parameters governing our system are fixed and known, leading to the treatment of detection events as independent and identically distributed (i.i.d.). Under the i.i.d. framework, each detection event is considered independent, and the probability distribution governing these events remains constant over time. However, this assumption of identically distributed data does not align with the inherent stochasticity present in real-world quantum key distribution (QKD) systems, where variables such as the intensity of a laser source or detector efficiencies may vary due to external conditions or device imperfections.

By transitioning from a fixed-parameter model to a Bayesian framework, we allow these parameters to be treated as random variables. This shift introduces prior distributions over the parameters, which are updated as more data becomes available, reflecting their true variability. As an example, consider the intensity $\lambda$ of a laser source. Instead of assuming $\lambda$ is fixed and known, we can model it as a random variable with a Gamma prior distribution, Gamma$(\alpha^{}_{\lambda}, \beta^{}_{\lambda})$, where $\alpha^{}_{\lambda}$ and $\beta^{}_{\lambda}$ are the shape and rate parameters of the Gamma distribution, respectively.

To illustrate, suppose the laser's intensity $\lambda$ has an expected value $\mathbb{E}[\lambda]$ and a variance $\mathbb{V}[\lambda]$ based on the laser’s specifications. We can set the parameters of the Gamma prior to match these moments~\cite{murphy2022probabilistic}:
    
\begin{align}
	\mathbb{E}[\lambda] = \frac{\alpha^{}_{\lambda}}{\beta^{}_{\lambda}}, \quad \mathbb{V}[\lambda] = \frac{\alpha^{}_{\lambda}}{\beta^{2}_{\lambda}}.	
\end{align}

Solving these equations for $\alpha^{}_{\lambda}$ and $\beta^{}_{\lambda}$ gives:
    
\begin{align}
	\alpha^{}_{\lambda} = \frac{\mathbb{E}[\lambda]^2}{\mathbb{V}[\lambda]}, \quad \beta^{}_{\lambda} = \frac{\mathbb{E}[\lambda]}{\mathbb{V}[\lambda]}.\label{eq-gamma_parameters}
\end{align}
    
This formulation provides a prior distribution for $\lambda$ that reflects both its expected value and the uncertainty due to variability, allowing our model to capture the true stochastic nature of the laser intensity. Similarly, for the detector efficiency $p^{}_{c}$, a Beta distribution $\text{Beta}(\alpha^{}_{p^{}_{c}}, \beta^{}_{p^{}_{c}})$ can be used as a prior.
    
Applying the same method to the Beta distribution, the parameters $\alpha^{}_{p^{}_{c}}$ and $\beta^{}_{p^{}_{c}}$ can be derived from the expected value $\mathbb{E}[p^{}_{c}]$ and variance $\mathbb{V}[p^{}_{c}]$ as~\cite{murphy2022probabilistic}:
    
\begin{align}
\alpha^{}_{p^{}_{c}} = \frac{\mathbb{E}[p^{}_{c}]^2 (1 - \mathbb{E}[p^{}_{c}])}{\mathbb{V}[p^{}_{c}]} - \mathbb{E}[p^{}_{c}], \quad \beta^{}_{p^{}_{c}} = \frac{\mathbb{E}[p^{}_{c}] (1 - \mathbb{E}[p^{}_{c}])^2}{\mathbb{V}[p^{}_{c}]} + \mathbb{E}[p^{}_{c}] - 1.\label{eq-beta_parameters}
\end{align}
    
By setting these parameters, we construct a Beta prior distribution for $p^{}_{c}$ that accurately reflects both the expected efficiency and the uncertainty inherent in the detector's performance.

By adopting this approach, we move along a spectrum from deterministic models with fixed parameters to fully Bayesian models that account for uncertainties in all parameters. At one extreme, a fixed parameter can be viewed as having a Dirac delta function as its prior~\cite{murphy2022probabilistic}, reflecting complete certainty about its value. At the other extreme, we have models like those used for Eve's parameters, where the priors are non-informative or represent maximal uncertainty. In between, partially informative priors, such as those for $\lambda$ and $p^{}_{c}$, provide a balanced representation that incorporates both prior knowledge and observed data.

An alternative approach is to remodel the likelihood by integrating over the prior distribution of the parameter, thus accounting for its uncertainty directly in the likelihood function. For example, the distribution of photons $n$ from a laser with intensity $\lambda$ could be updated as follows:
    
\begin{align}
	P(n \mid \alpha^{}_{\lambda}, \beta^{}_{\lambda}) = \int^{\infty}_{0} \text{Poisson}(n \mid \lambda) \, \text{Gamma}(\lambda \mid \alpha^{}_{\lambda}, \beta^{}_{\lambda}) \, d\lambda.
\end{align}    
This is known to follow a negative binomial distribution with parameters $r=\alpha^{}_{\lambda}$ (number of successes) and $p=\frac{\beta^{}_{\lambda}}{\beta^{}_{\lambda}+1}$ (success probability)~\cite{greenwood1920inquiry}, or more precisely:
    
\begin{align}
	P(n \mid \alpha^{}_{\lambda}, \beta^{}_{\lambda}) &= \text{NegativeBinomial}\left(n\mid r=\alpha^{}_{\lambda},p=\frac{\beta^{}_{\lambda}}{\beta^{}_{\lambda}+1}\right),\nonumber\\
	&= \frac{\Gamma(n + \alpha^{}_{\lambda})}{n! \, \Gamma(\alpha^{}_{\lambda})} \left(\frac{\beta^{}_{\lambda}}{1 + \beta^{}_{\lambda}}\right)^{\alpha^{}_{\lambda}} \left(\frac{1}{1 + \beta^{}_{\lambda}}\right)^{n}.
	\end{align}
    
A similar approach can be applied to the $p$ parameter in a binomial distribution. If $p$ is as a random variable following a Beta distribution, the probability of observing $k$ successes in $n$ trials with a variable success probability $p$ can be expressed as:
\begin{align}
	P(k \mid n, \alpha^{}_{p}, \beta^{}_{p}) = \int^{1}_{0} \text{Binomial}(k \mid n, p) \, \text{Beta}(p \mid \alpha^{}_{p}, \beta^{}_{p}) \, dp.
\end{align}
This scenario is modeled by a Beta-Binomial distribution~\cite{murphy2022probabilistic}, specifically:
\begin{align}
P(k \mid n, \alpha^{}_{p}, \beta^{}_{p}) &= \text{Beta-Binomial}\left(k \mid n, \alpha^{}_{p}, \beta^{}_{p}\right), \nonumber\\
&= \binom{n}{k} \frac{\text{B}\left(k + \alpha^{}_{p}, n - k + \beta^{}_{p}\right)}{\text{B}\left(\alpha^{}_{p}, \beta^{}_{p}\right)},
\end{align}
where $\text{B}(x, y)$ denotes the beta function:
\begin{align}
	\text{B}(x, y) = \frac{\Gamma(x) \Gamma(y)}{\Gamma(x + y)}
\end{align}
    
Modeling the photon distribution of the laser and the detector's efficiency in this manner accounts for the uncertainty in the device parameters. However, while this approach is mathematically elegant and can simplify calculations under certain conditions, it constrains the choice of priors and necessitates a comprehensive reformulation of the likelihood function. By contrast, the first approach—placing priors directly on the parameters and updating them through the posterior—offers greater flexibility and allows for more nuanced modeling of the actual dynamics of the physical system. This flexibility is particularly crucial in security-critical applications like QKD, where accurate modeling of all system components is essential to ensure robustness against potential attacks~\cite{scarani2009security}.
    
With this framework in place, we have moved beyond the assumption of identically distributed detection events by treating key parameters as random variables with their own distributions.

\section{Experimental Results}
\label{sec:experimental-results}

This section presents the validation and performance assessment of the proposed probabilistic framework through a series of simulations and comparisons with the decoy-state protocol. First, we test the accuracy of the probabilistic model (both i.i.d. and HMM) in predicting detection events and errors. Next, we evaluate the error modeling and its alignment with simulated results. The third subsection focuses on the secure key generation rate, highlighting improvements over the decoy-state approach. Finally, we explore the advantages of the Bayesian inference framework, demonstrating its robustness under after-pulsing scenarios (via HMM) and noisy conditions, emphasizing the superiority of the fully Bayesian approach.

\subsection{Validation through Simulation}

This section is dedicated to assessing how well the theoretical models developed in this work fit the data generated by simulations. Specifically, we aim to evaluate the extent to which the theoretical model captures the observed behavior under both i.i.d. and Hidden Markov Model (HMM) conditions.

The first subsection focuses on validating the i.i.d. model as formalized in \Cref{theorem:likelihood_iid}, using data generated by the simulation algorithm described in \Cref{alg:simulate_iid}. The second subsection examines the HMM model, as formulated in \Cref{theorem:likelihood_hmm}, with data produced according to \Cref{alg:simulate_hmm}. By examining these models in turn, we determine the degree to which each theoretical framework explains the simulated detection events.

\begin{algorithm}
\caption{Simulation of QKD detection events under the i.i.d assumption}
\label{alg:simulate_iid}
\begin{algorithmic}[1]
\Function{simulate\textsubscript{iid}}{$N, \theta^{}_{A}, \theta^{}_{B}, \theta^{}_{E}$}
    \State $[\Lambda, \alpha, d^{}_{AB}] \gets \theta^{}_{A}$
    \State $[p^{}_{a^{}_{0}}, p^{}_{a^{}_{1}}, p^{}_{c^{}_{0}}, p^{}_{c^{}_{1}}, p^{}_{d^{}_{0}}, p^{}_{d^{}_{1}}, p^{}_{e}] \gets \theta^{}_{B}$
    \State $[d^{}_{AE}, p^{}_{EB}, k, \Delta] \gets \theta^{}_{E}$
    \State $p^{}_{AB} \gets 10^{-\alpha d^{}_{AB}/10}$
    \State $p^{}_{AE} \gets 10^{-\alpha d^{}_{AE}/10}$
    \State $N^{}_{\lambda} \gets |\Lambda|$
    \State $\hat p^{}_{e} \gets \arcsin\left(\sqrt{p^{}_{e}}\right)$
	\State Initialize $D^{}_{0}$, $D^{}_{1}$, $a$, $b$, $x$, $e$ and $e$

    \For{$i = 1$ to $N$}
        \State \makebox[20.5em][l]{$l[i] \sim \text{Multinomial}(1, [1~1~\hdots 1]/N^{}_{\lambda})$}\% Choose random intensities
        \State \makebox[20.5em][l]{$x[i], a[i], b[i] \sim \text{Bernoulli}(\frac{1}{2})$}\% Alice and Bob choose random bits \& bases
        \State \makebox[20.5em][l]{$e[i] \sim \text{Bernoulli}(\Delta)$}\% Eve chooses random pulses to intercept

        \State \makebox[20.5em][l]{$n^{}_{a} \sim \text{Poisson}(\Lambda[l[i]])$}\% Sample n photons given intensity $l[i]$
        \If{$e[i] == \text{false}$}
            \State \makebox[19em][l]{$n^{}_{b} \sim \text{Binomial}(n^{}_{a}, p^{}_{AB})$}\% If Eve does not intercepts, apply fiber $p^{}_{AB}$
        \Else
            \State \makebox[19em][l]{$n^{}_{e} \sim \text{Binomial}(n^{}_{a}, p^{}_{AE})$}\% If Eve intercepts, apply fiber $p^{}_{AE}$,
            \State \makebox[19em][l]{$n^{}_{e} \gets \max(n^{}_{e} - k, 0)$}\% then grab $k$ photons,
            \State \makebox[19em][l]{$n^{}_{b} \sim \text{Binomial}(n^{}_{e}, p^{}_{EB})$}\% then apply fiber $p^{}_{EB}.$
        \EndIf
        \State \makebox[20.5em][l]{$p^{}_{0} \gets \cos\left( \frac{\pi}{2}x[i] - \frac{\pi}{4}(a[i] - b[i]) + \hat p^{}_{e} \right)^2$}\% Compute probability of $D^{}_{0}$
        \State \makebox[20.5em][l]{$n^{}_{0} \sim \text{Binomial}(n^{}_{b}, p^{}_{0})$}\% Photons directed to $D^{}_{0}$
        \State \makebox[20.5em][l]{$n^{}_{1} \gets n^{}_{b} - n^{}_{0}$}\% The remaining are directed to $D^{}_{1}$
        \\
        \For{$j = 0$ to $1$}\hspace{12.85em}\% For each detector,
        	\State \makebox[19em][l]{$d^{}_{j} \sim \text{Bernoulli}(p^{}_{d^{}_{j}})$}\% Produce a random dark click
        	\State \makebox[19em][l]{$s^{}_{j} \sim \text{Binomial}(n^{}_{j},p^{}_{c^{}_{j}})$}\% Sample $s$ photons to be detected
        	\State \makebox[19em][l]{$D^{}_{j}[i] \gets d^{}_{j}\vee \left(s^{}_{j} \ge 1\right)$}\% Click if at least one photon or dark count
        \EndFor	
    \EndFor
    \State \Return $D^{}_{0}, D^{}_{1}, l, a, b, x, e$
\EndFunction
\end{algorithmic}
\end{algorithm}

\begin{algorithm}
\caption{Simulation of QKD detection events under the HMM assumption}
\label{alg:simulate_hmm}
\begin{algorithmic}[1]
\Function{simulate\textsubscript{HMM}}{$N, \theta^{}_{A}, \theta^{}_{B}, \theta^{}_{E}$}
    \State \makebox[23em][l]{$D^{}_{0}, D^{}_{1}, l, a, b, x, e \gets  \text{SIMULATE}_{\text{IID}}\left(N, \theta^{}_{A}, \theta^{}_{B}, \theta^{}_{E}\right)$}\% First run an i.i.d. simulation
    \State \makebox[23em][l]{Initialize $a^{}_{0},a^{}_{1}$}\% After-pulses for detectors $D^{}_{0}$ and $D^{}_{1}$
    \For{$i = 2$ to $N$}
    	\For{$j = 0$ to $1$}
        	\If{$D^{}_{j}[i-1] \wedge \neg a^{}_{j}[i-1]$}\hspace{7.08em}\% If previous is a click \& not an after-pulse
            	\State $a^{}_{j}[i] \sim \text{Bernoulli}(p^{}_{a^{}_{j}})$\hspace{9.2em}\% Set to after-pulse with probability $p^{}_{a^{}_{j}}$
	        \EndIf
	        \State $D^{}_{j}[i] \gets D^{}_{j}[i] \vee a^{}_{j}[i]$\hspace{11.2em}\% Click if a genuine click or  an after-pulse 
		\EndFor
    \EndFor
    \State \Return $D^{}_{0}, D^{}_{1}, l, a, b, x, e$
\EndFunction
\end{algorithmic}
\end{algorithm}

To model realistic QKD conditions, we adopted the Gobby-Yuan-Shields (GYS) configuration~\cite{gobby2004quantum}, providing parameterization suitable for real-world applications and establishing a foundation for later comparison with the decoy-state protocol~\cite{ma2005practical}. The main simulation parameters are presented in \Cref{tab:main_parameters}. Additionally, we introduced variability in detector-specific parameters to test asymmetric scenarios, as shown in \Cref{tab:detector_specific_adjusted}. For the i.i.d. test, the after-pulse probability $p^{}_{a}$ was set to 0. The selected intensities for this experiment correspond to the minimum and maximum values, $\lambda^{}_{\min}$ and $\lambda^{}_{\max}$, as described in \Cref{sec-multiple-intensities}. Each simulation was run 10,000 times with 10,000 pulses per run, ensuring sufficient data for a comprehensive validation of the theoretical predictions.

\begin{table}
    \centering
    \renewcommand{\arraystretch}{1.5}
    \begin{tabular}{l l l}
        \multicolumn{3}{c}{\textbf{Parameter Configuration for the Simulation Experiment}} \\
        
        \toprule
        \multicolumn{3}{c}{$\theta^{}_{A}$ (Alice)} \\
        \midrule
        Parameter & Description & Value \\
        \midrule
        $\Lambda$& Intensity levels & $\left\{\lambda^{}_{\min}, \lambda^{}_{\max}\right\}$ \\
        $\alpha$ & Channel attenuation (dB/km) & 0.21 \\
        $d^{}_{AB}$ & Distance from Alice to Bob (km) & 50 \\
        
        \toprule
        \multicolumn{3}{c}{$\theta^{}_{B}$ (Bob)} \\
        \midrule
        Parameter & Description & Value \\
        \midrule
        $p^{}_{a}$ & After-pulse probability & 0.1 \\
        $p^{}_{c}$ & Detection efficiency & 0.045 \\
        $p^{}_{d}$ & Dark count probability & $1.7 \times 10^{-6}$ \\
        $p^{}_{e}$ & Misalignment probability & 0.033 \\
        
        \toprule
        \multicolumn{3}{c}{$\theta^{}_{E}$ (Eve)} \\
        \midrule
        Parameter & Description & Value \\
        \midrule
        $d^{}_{AE}$ & Distance from Alice to Eve (km) & 10 \\
        $\Delta$ & Fraction of intercepted pulses & 0.2 \\
        $k$ & Photons intercepted per pulse & 3 \\
        $p^{}_{EB}$ & Channel efficiency from Eve to Bob (if Eve intercepts) & optimized \\
        
        \toprule
        \multicolumn{3}{c}{$\theta^{}_{S}$ (Session)} \\
        \midrule
        Parameter & Description & Value \\
        \midrule
        $N$ & Number of pulses per session & 10,000 \\
        $N^{}_{R}$ & Number of simulation runs & 10,000 \\
        
        \bottomrule
    \end{tabular}
    \caption{Parameter configuration for the simulation experiment, grouped by Alice's parameters ($\theta^{}_{A}$), Bob's parameters ($\theta^{}_{B}$), Eve's parameters ($\theta^{}_{E}$), and session parameters ($\theta^{}_{S}$).}
    \label{tab:main_parameters}
\end{table}

\begin{table}
    \centering
    \renewcommand{\arraystretch}{1.5}
    \begin{tabular}{l c c}
        \toprule
        \multicolumn{3}{c}{\textbf{Detector-Specific Parameter Adjustments}} \\
        \midrule
        Parameter & $D^{}_{0}$ & $D^{}_{1}$ \\
        \midrule
        After-pulse probability & $p^{}_{a^{}_{0}} = p^{}_{a} \times 0.9$ & $p^{}_{a^{}_{1}} = p^{}_{a} \times 1.1$ \\
        Detection efficiency & $p^{}_{c^{}_{0}} = p^{}_{c} \times 0.9$ & $p^{}_{c^{}_{1}} = p^{}_{c} \times 1.1$ \\
        Dark count probability & $p^{}_{d^{}_{0}} = p^{}_{d} \times 0.9$ & $p^{}_{d^{}_{1}} = p^{}_{d} \times 1.1$ \\
        \bottomrule
    \end{tabular}
    \caption{Detector-specific values for Bob's parameters: after-pulse probability $p^{}_{a}$, detection efficiency $p^{}_{c}$, and dark count probability $p^{}_{d}$ for detectors $D^{}_{0}$ and $D^{}_{1}$.}
    \label{tab:detector_specific_adjusted}
\end{table}

\subsubsection{Validation of the i.i.d. Model}
\label{sec:validate_iid}
According to \Cref{theorem:likelihood_iid}, the distribution of clicks in the i.i.d. scenario should follow a multinomial distribution with parameters $N$ and $\mathbf{P}(\theta)$. Here, $\mathbf{P}(\theta)$ is a probability vector of length $2 \times 4 \times 2 = 16$, corresponding to two intensity levels ($\lambda^{}_{\min}$ and $\lambda^{}_{\max}$), four possible detection outcomes (00, 01, 10, and 11), and two basis configurations (matching and non-matching).

We can infer the distribution for any specific detection scenario—defined by a particular basis alignment, click pattern, and intensity level—by marginalizing over all remaining variables in $\mathbf{P}(\theta)$. When focusing on a single outcome within a multinomial distribution, the marginal distribution of each individual outcome follows a binomial distribution. More formally,
\[
\begin{aligned}
    \text{If} \quad c^{}_{\phantom{i}} &\sim \text{Multinomial}(n, \mathbf{p}), \\
    \text{then} \quad c^{}_{i} &\sim \text{Binomial}(n, p^{}_{i}),
\end{aligned}
\]
where $i$ denotes the index of the variable of interest, $c^{}_{i}$ represents its count, and $p^{}_{i}$ is the corresponding probability component in the vector $\mathbf{p}$.

Having derived the binomial distribution for each specific case within $\mathbf{P}(\theta)$, we compare the resulting histograms from simulated data against the theoretical binomial distributions inferred from the marginalized multinomial model. Using the parameters of each binomial distribution, we computed the 99\% credible interval to capture the expected range of variation. 

The results are reported in \Cref{fig:validate_iid_Ps}, where each panel represents a unique combination of intensity level, detection outcome, and basis alignment configuration. In each plot, the normalized histogram of simulated data is shown alongside the theoretical probability mass function (PMF) for the binomial distribution, with the 99\% credible interval shaded. This comparison assesses the degree of alignment between theory and data, providing empirical validation of the i.i.d. model under the conditions specified.

\begin{figure}[H]
    \centering
    \includegraphics[width=\textwidth]{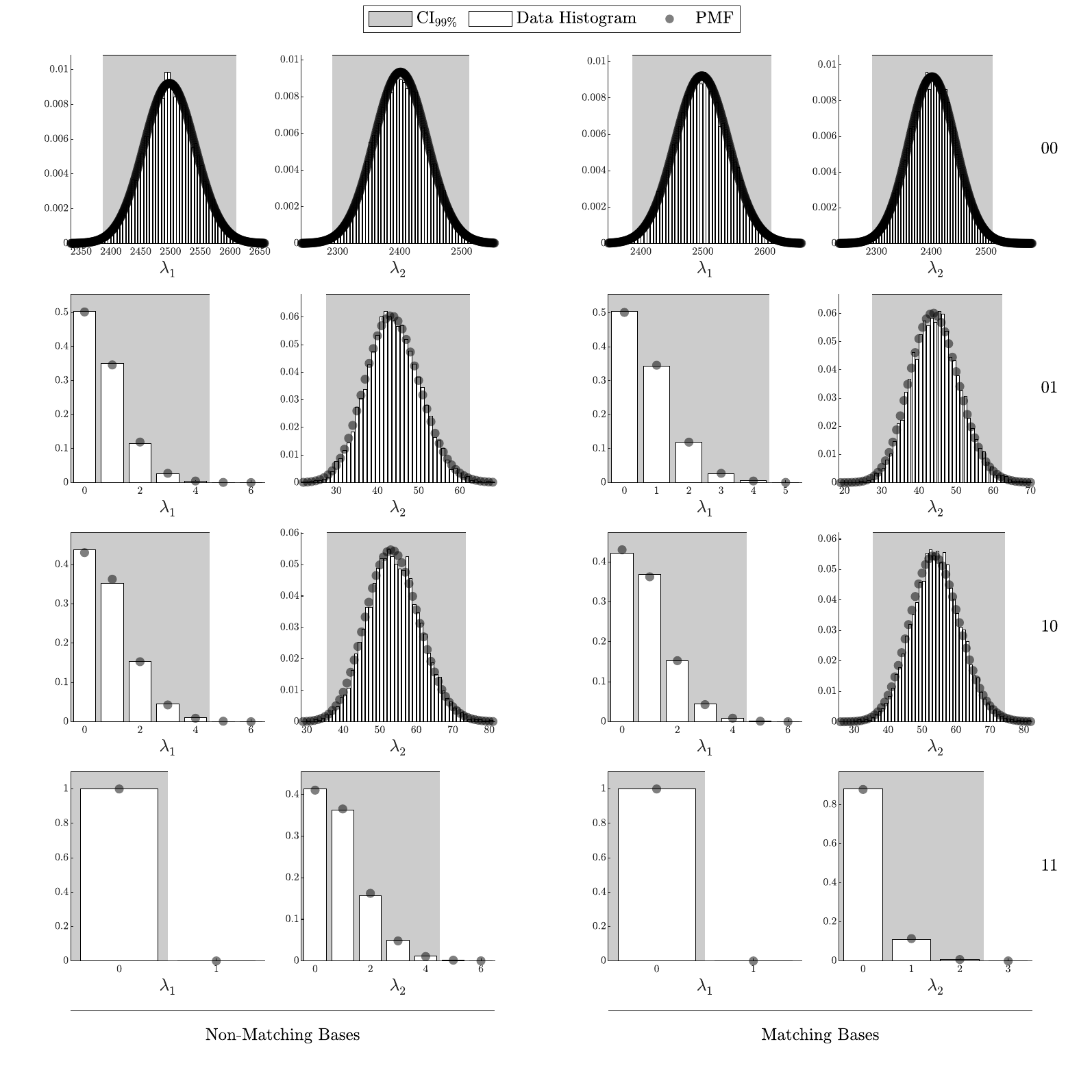}
    \caption{Comparison of simulated detection events (histograms) against theoretical binomial PMFs (solid lines) for various detection outcomes and intensities, with 99\% confidence intervals (shaded area). Each subplot represents a unique combination of intensity and detection outcome (00, 01, 10, 11) across matching and non-matching bases.}
    \label{fig:validate_iid_Ps}
\end{figure}

In addition to validating the general detection probabilities, we conducted an experiment to assess the accuracy of our theoretical model for the distribution of signal and erroneous clicks. Specifically, we tracked the number of gain clicks $G^{}_{m\lambda}$ and erroneous clicks $R^{}_{m\lambda}$ for each basis alignment configuration $m$ (matching or non-matching) and intensity level $\lambda$.

For each combination of basis alignment and intensity, the histogram of simulated counts for gain clicks $G^{}_{m\lambda}$ and erroneous clicks $R^{}_{m\lambda}$ was compared with the theoretical binomial distributions as defined in \Cref{eq-G_m_lambda_E} and \Cref{eq-R_m_lambda_E}.

The results are reported in \Cref{fig:validate_iid_EQs}, where it clearly shows the alignment between the theoretical probability mass functions (PMFs) and the empirical histograms from the simulated data. The 99\% confidence intervals derived from the binomial distributions are shaded in each plot, providing a benchmark for the expected range of counts. Observing that the simulated data closely adheres to the theoretical PMF within these confidence intervals further confirms the robustness of our model in capturing both signal and error probabilities for different basis alignments and intensities.

\begin{figure}
    \centering
    \includegraphics[width=\textwidth]{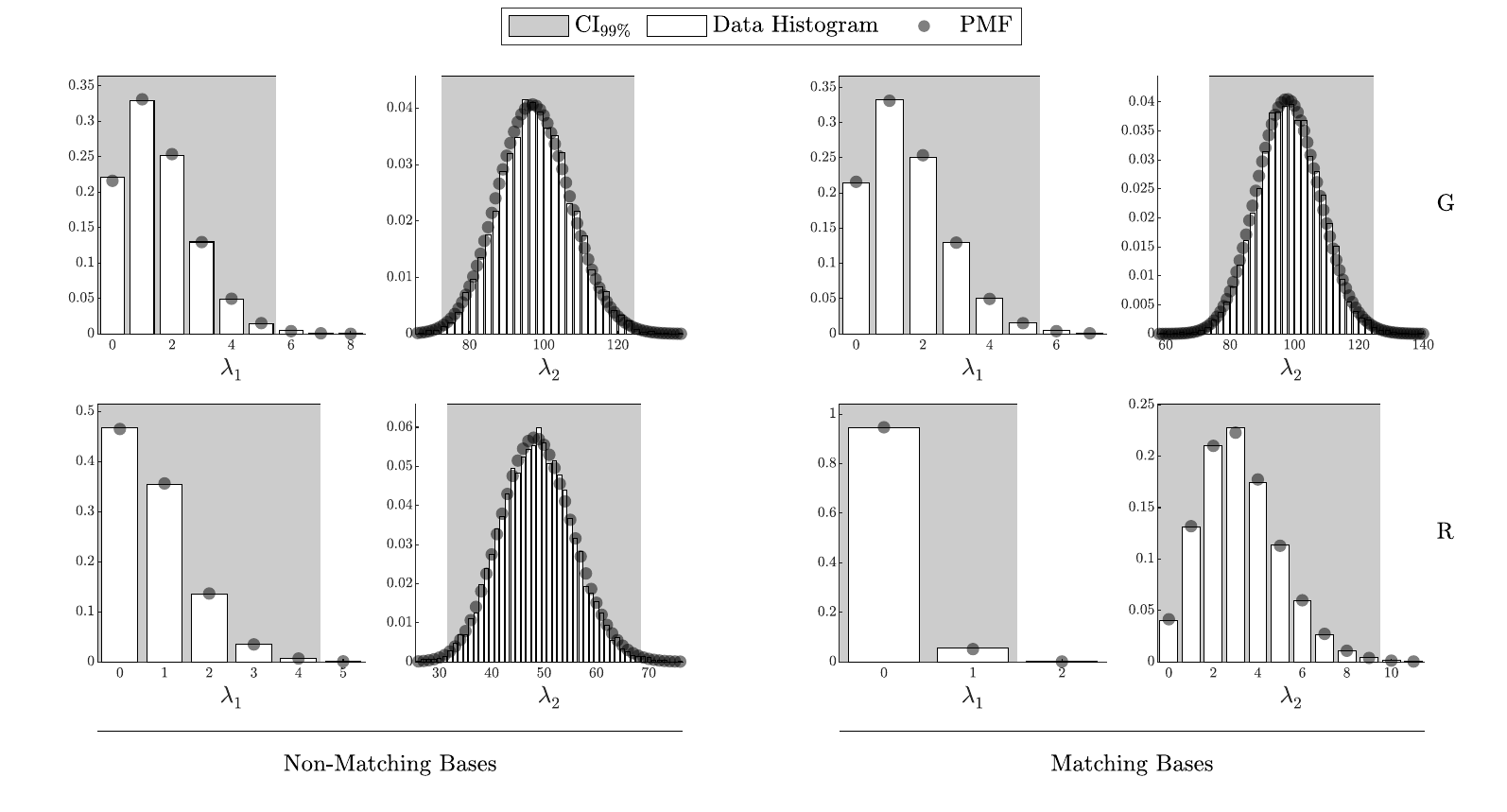}
    \caption{Comparison of simulated gain ($G$) and error ($R$) counts (histograms) against theoretical binomial PMFs (solid lines) for various basis alignments and intensity levels. The shaded areas represent the 99\% confidence intervals for each case.}
    \label{fig:validate_iid_EQs}
\end{figure}

\subsubsection*{Comparison with the Decoy-state Protocol}
As a key difference between the proposed method and the decoy-state protocol, the proposed framework provides the ability to compute the probabilities of individual detection events (00, 01, 10 and 11) for both matching and non-matching bases. In contrast, the decoy-state protocol does not distinguish between these specific detection outcomes. Instead, the decoy-state protocol simplifies the analysis by focusing on aggregated quantities—the gain count ($G$) and error count ($R$)—while treating double-click events as both signal and error, and restricting its analysis to the matching bases ($a = b$).

Due to these distinctions, a figure equivalent to \Cref{fig:validate_iid_Ps} cannot be generated for the decoy-state protocol, as its authors did not provide theoretical predictions for individual click patterns. However, the decoy-state protocol does give testable theoretical predictions for number of signal and error clicks under the matching basis case, which can be evaluated. For this experiment, we adjusted the counting to include double-click events in both ($G$) and ($R$) and restricted the analysis to matching bases  ($a = b$). The results of this comparison are presented in \Cref{fig:validate_iid_decoy}.

\begin{figure}
    \centering
    \includegraphics[width=0.5\textwidth]{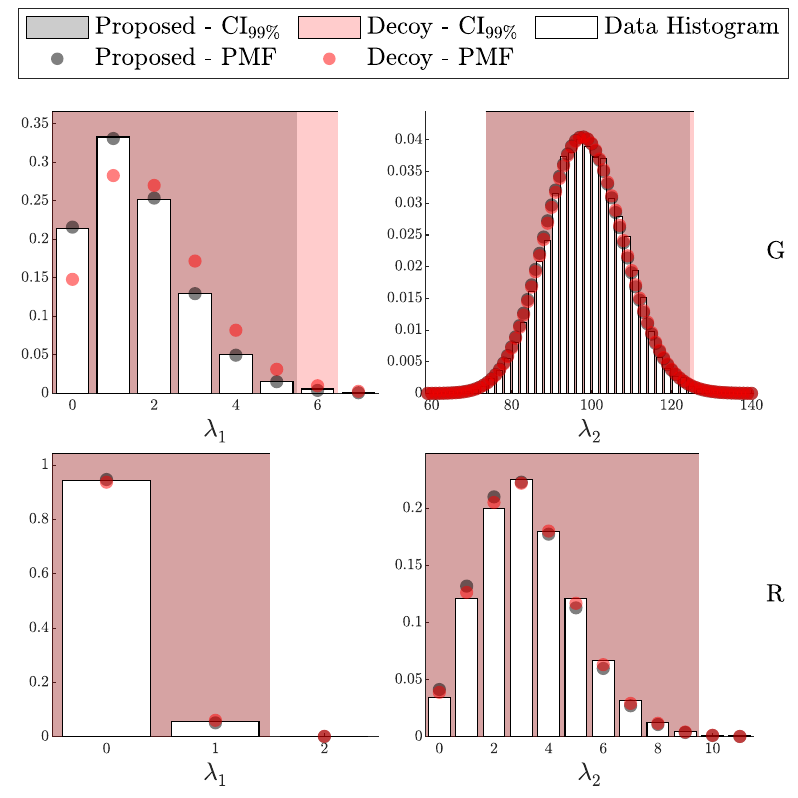}
    \caption{Comparison of simulated gain ($G$) and error ($R$) counts (histograms) against theoretical binomial PMFs (solid lines) for the decoy-state protocol. Double-click events are included in both gain and error counts, and only matching bases are considered, consistent with the assumptions of the decoy-state protocol. The shaded areas represent the 99\% confidence intervals derived from the theoretical binomial distributions.}
    \label{fig:validate_iid_decoy}
\end{figure}

The results in \Cref{fig:validate_iid_decoy} reveal discrepancies between the simulated data and the theoretical predictions of the decoy-state protocol, particularly for the gain counts ($G$) under the low-intensity level $\lambda^{}_{1}$. These deviations highlight the limitations of the decoy-state protocol in accurately modeling the detection process, as it does not account for the finer distinctions between individual click patterns or non-matching bases. While this analysis focuses on the specific configuration of the matching case, later experiments will demonstrate that such deviations become more significant intensity decreases.

By comparison, the proposed framework, as demonstrated in \Cref{fig:validate_iid_Ps} and \Cref{fig:validate_iid_EQs}, consistently aligns closely with theoretical predictions across a broader range of scenarios. This further underscores the robustness of the proposed method in accurately capturing both the individual detection events and the aggregated signal and error counts, providing a more comprehensive and reliable framework for QKD security analysis.

\subsubsection{Validation of the HMM Model}

In this subsection, we validate the Hidden Markov Model (HMM) framework, as formalized in \Cref{theorem:likelihood_hmm}. This validation method we employ here mirrors the i.i.d. approach in \Cref{sec:validate_iid} but incorporates after-pulsing in the simulation by setting $p^{}_{a} = 0.1$ and uses the HMM-based probability vector $\hat{\mathbf{P}}(\theta)$ in the likelihood (10\% variation between the detectors was employed, see \Cref{tab:detector_specific_adjusted}).

For comparison, we also applied the i.i.d. based probability vector $\mathbf{P}(\theta)$ to the same HMM-generated data to assess the extent of deviation when dependencies are present.

\Cref{fig:validate_hmm_Ps,fig:validate_hmm_EQs} illustrate strong alignment between the simulated HMM data and the theoretical distributions derived from $\hat{\mathbf{P}}(\theta)$, supporting the validity of the HMM model under these experimental conditions. In contrast, a clear deviation is shown when the i.i.d. model is applied to the data affected by after-pulsing.

\begin{figure}
    \centering
    \includegraphics[width=\textwidth]{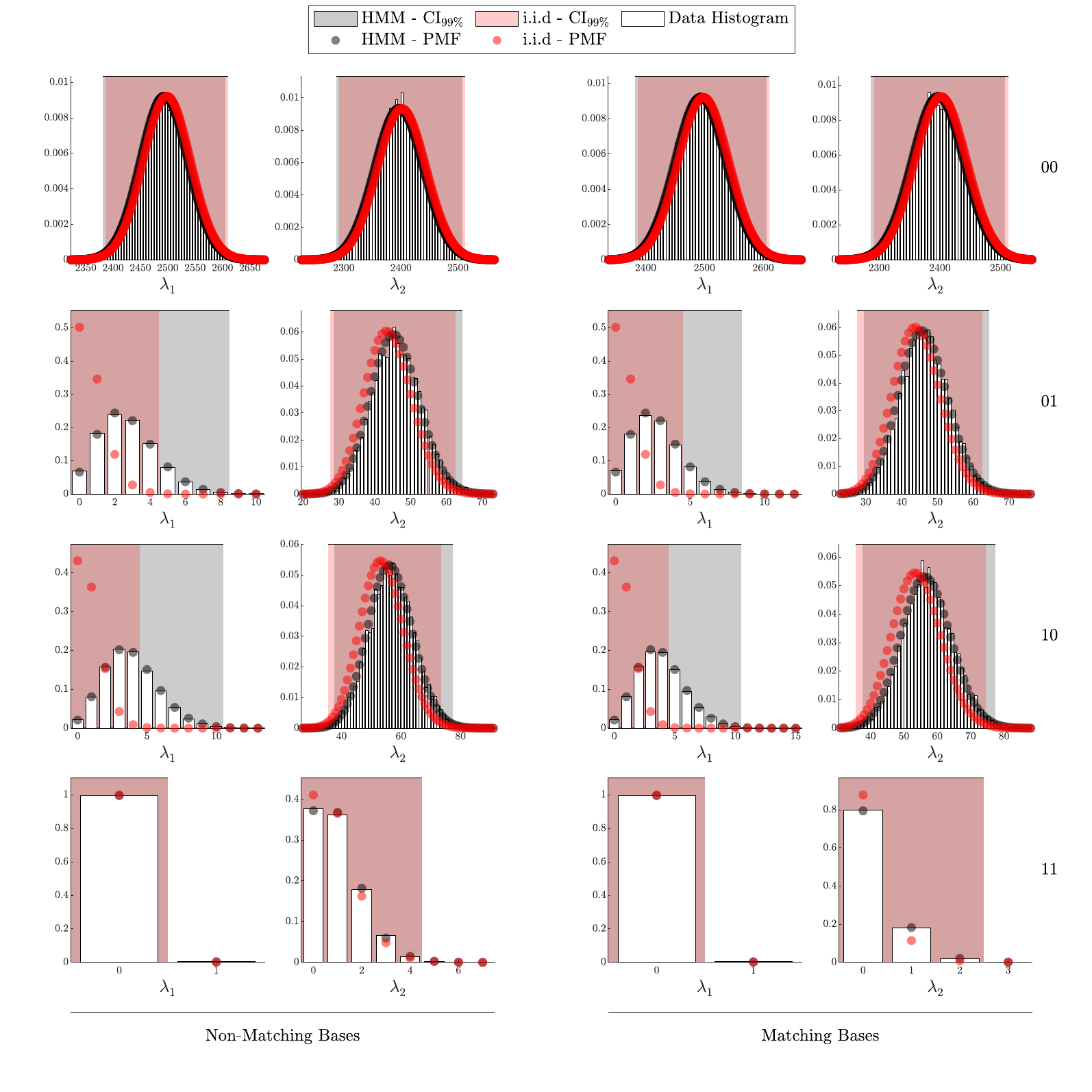}
    \caption{Comparison of simulated detection events (histograms) against theoretical PMFs under the HMM model (solid lines) for various detection outcomes and intensities. The shaded areas represent the 99\% confidence intervals for each case.}
    \label{fig:validate_hmm_Ps}
\end{figure}

\begin{figure}
    \centering
    \includegraphics[width=\textwidth]{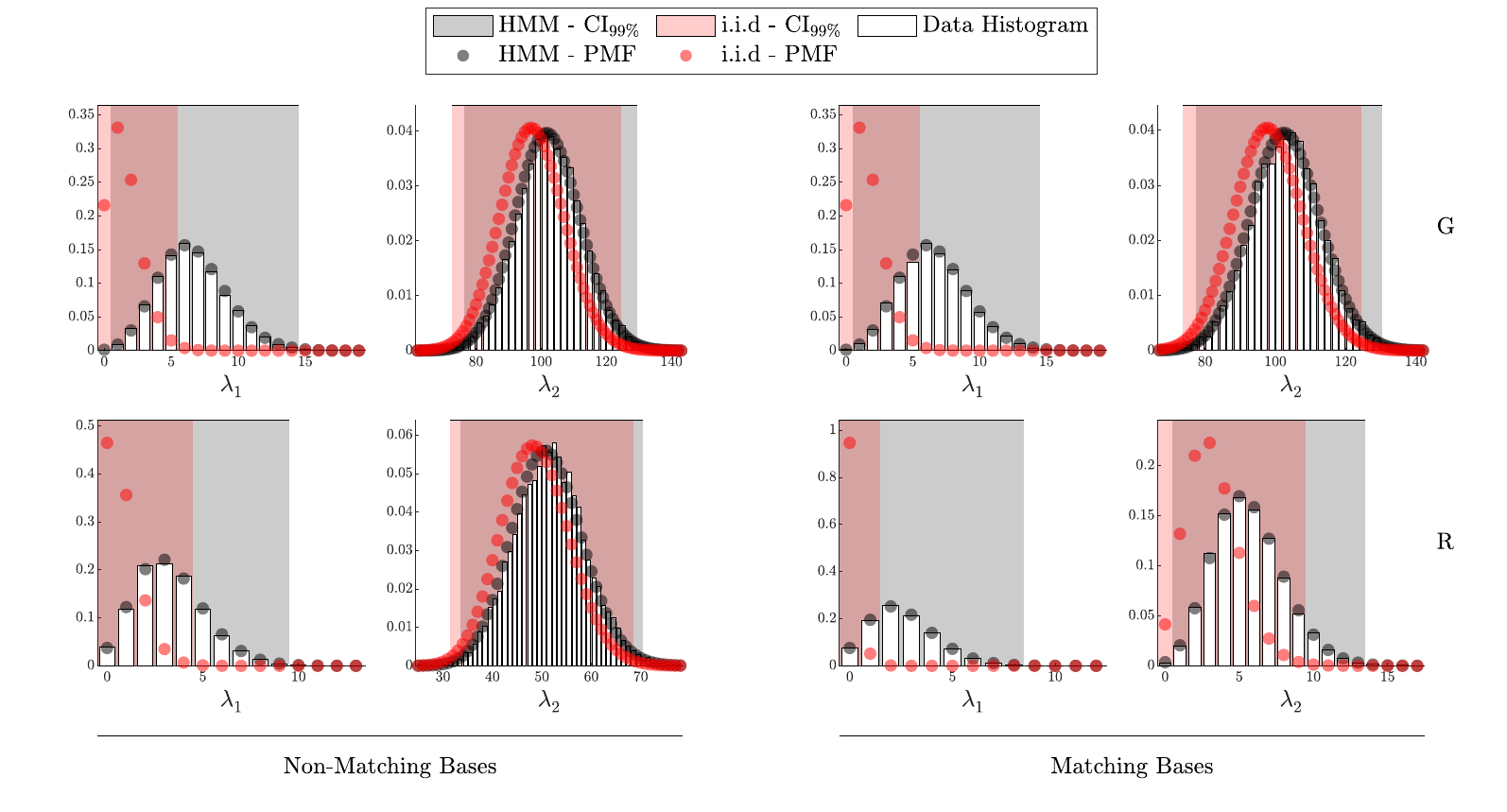}
    \caption{Comparison of simulated gain ($G$) and error ($R$) counts (histograms) under HMM conditions against theoretical PMFs from the HMM model (solid lines) for various basis alignments and intensity levels. Shaded areas indicate the 99\% confidence intervals.}
    \label{fig:validate_hmm_EQs}
\end{figure}

\subsection{Error Rate Estimation}

To validate the error rate calculations derived from our proposed method and compare them to the decoy-state protocol, we analyze the error rate $\delta^{}_{m\lambda}(\theta)$ as defined in \Cref{eq-delta_E_Q}. To ensure comparability with the decoy-state protocol, we consider a scenario with symmetric detectors and define a gain as any pulse where either of the detectors registers a click. By applying \Cref{lem:pair_detectos} and \Cref{lem:laser_detector}, we simplify the calculations for $Q^{}_{m\lambda}(\theta)$ and $EQ^{}_{m\lambda}(\theta)$ under the conditions of matching bases, no eavesdropper intervention, and a signal intensity $\mu$. This yields the following expressions for our model:

\begin{align}
    Q^{}_{\mu}(\theta) &= 1 - (1 - p^{}_{d})^2 e^{-\mu\cdot p^{}_{AB}\cdot p^{}_{c}},\label{eq:Q_corrected} \\
    EQ^{}_{\mu}(\theta) &= 1 - (1 - p^{}_{d}) e^{-\mu\cdot p^{}_{AB}\cdot p^{}_{c}\cdot p^{}_{e}},\label{eq:E_corrected}
\end{align}
whereas the decoy-state protocol, the corresponding calculations are given by:

\begin{align}
    Q^{}_{\mu}(\theta) &= p^{}_{d} + 1 - e^{-\mu\cdot p^{}_{AB}\cdot p^{}_{c}},\label{eq:Q_decoy} \\
    EQ^{}_{\mu}(\theta) &= e^{}_{0}\cdot p^{}_{d} + p^{}_{e} (1 - e^{-\mu\cdot p^{}_{AB}\cdot p^{}_{c}}),\label{eq:E_decoy}
\end{align}
where $e^{}_{0} = \frac{1}{2}$, and we omit the subscript $m = 1$ for brevity. The decoy-state protocol justified their approximation  under the assumption that $p^{}_{d}$ is small, making it a reasonable simplification for most practical scenarios. However, in a subsequent experiment, we will demonstrate that this approximation can lead to significant errors in certain cases.

For the purpose of this experiment, both the proposed method and the decoy-state protocol treat double-click events as both a gain and an error to enable a one-to-one comparison. While necessary for consistency, double-click events inherently provide no useful information and have an error rate of 50\%. The decoy-state protocol, however, assumes an error rate of 100\% for double-click events, even when employing more refined calculations. This assumption contradicts fundamental information-theoretic principles, as a 100\% error rate corresponds to minimum Shannon entropy, implying more information than a 50\% error rate, where entropy is maximized. Consequently, this misrepresentation underestimates the true entropy of the data and leads to an overestimation of the secure-key rate, introducing significant discrepancies in the security analysis.

\subsubsection{Simulation Setup and Results}

To evaluate the performance of the decoy-state protocol and our proposed method in estimating $\delta$, we conducted simulations using the parameter set from \Cref{tab:main_parameters}. To enable a direct comparison with the decoy-state protocol, we configured the detectors symmetrically and set the intensity to $\mu = 0.48$ (the optimal intensity for the GYS configuration as per~\cite{ma2005practical}). We also incorporated after-pulse probabilities of 5\% and 10\% to evaluate their impact on performance. Each session consisted of $10^9$ pulses, and the experiment was repeated over various distances between Alice and Bob, ranging from 0 to 150 km in 5 km increments.

For each configuration, we recorded the number of gain and erroneous clicks, treating double-click events as both a signal and an error to maintain consistency with the decoy-state approach. We then calculated the empirical error rates and compared them to the theoretical 99\% confidence intervals of the approximate Beta distribution described in \Cref{eq-beta_rho}, whose parameters were derived based on the mean and approximate variance from \Cref{lem:rho_E_V}. While the confidence intervals are approximate, the expected value is derived exactly.

The results are presented in \Cref{fig:error_rates}, which shows the error rates as a function of distance for various after-pulse probabilities. The shaded regions represent the 99\% confidence intervals, while the expected values are depicted as solid lines. The figure highlights the differences between the error rate estimates provided by our method and the decoy-state protocol across different scenarios.

\begin{figure}[h]
    \centering
    \includegraphics[width=\textwidth]{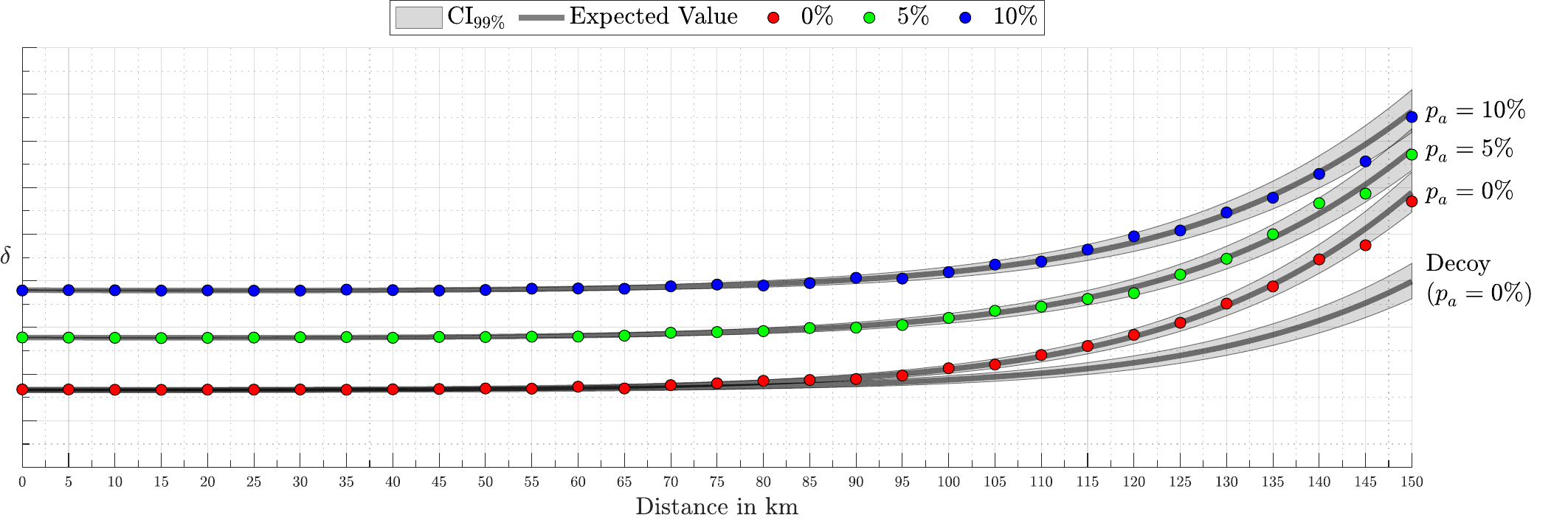}
    \caption{Error rate $\delta$ as a function of distance between Alice and Bob, plotted for different after-pulse probabilities ($p^{}_{a}$ = 0\%, 5\%, and 10\%). The shaded regions represent the 99\% confidence intervals, and the solid lines indicate the expected values. Comparisons are made with the decoy-state protocol and our proposed method.}
    \label{fig:error_rates}
\end{figure}

\subsubsection{Analysis of Results}

The results in \Cref{fig:error_rates} clearly illustrate that the error rate increases with distance for all scenarios, with higher after-pulse probabilities exacerbating the increase. The proposed method consistently provides more accurate error rate estimates, as indicated by the alignment of the empirical error rates with the theoretical confidence intervals. In contrast, the decoy-state protocol's approximation underestimates the error rate, particularly at greater distances and higher after-pulse probabilities.

This discrepancy highlights a significant limitation of the decoy-state protocol's approximation, which assumes a negligible dark count probability $p^{}_{d}$. As shown, this assumption becomes increasingly problematic as distance and after-pulse probability increase, leading to overconfident secure-key rate calculations. Our approach, which accounts for these effects more rigorously, provides a more reliable estimate, emphasizing the importance of using accurate error rate models in QKD implementations.

\subsection{Secure-Key Rate Comparison}

Building on our validation of the theoretical gain and error rates, we now examine the secure-key rates of the proposed protocol in comparison with the decoy-state protocol.

\subsubsection{Experimental Setup}
For this experiment, we implemented the weak+vacuum decoy-state protocol with intensity parameters $\mu = 0.48$, $\nu^{}_{1} = 0.05$, and $\nu^{}_{2} = 0$ to match the setup detailed in~\cite{ma2005practical}. The experiment was conducted using the GYS configuration with parameters specified in \Cref{tab:main_parameters}. Additionally, the detectors were symmetrized to align with the assumptions made in the decoy-state protocol paper. To highlight the advantages of the proposed protocol, we extended the analysis to include three additional intensities: $\mu = 1$, $\mu = 5$, and $\mu = 10$. These intensities were chosen to demonstrate the broader operational range achievable with the proposed method.

\subsubsection{Secure-Key Rate Calculation}
The secure-key rate for the decoy-state protocol is calculated as follows (after a slight change of notation from the original paper and omitting the dependency on $\theta$ for brevity):
\begin{align}
    K^{}_{\text{Decoy}} \ge q \cdot \left\{-Q^{}_{\mu} f\left(\delta^{}_{\mu}\right) H^{}_{2}\left(\delta^{}_{\mu}\right) + Q^{}_{1} \left[1 - H^{}_{2}\left(e^{}_{1}\right)\right]\right\},
\end{align}
where $q$ is the protocol efficiency, $\frac{1}{2}$ for the BB84 protocol, and $f(\delta^{}_{\mu})$ is the error correction efficiency, set to 1.22 (an upper bound for this experiment~\cite{ma2005practical}). The terms are defined as:
\begin{align}
    e^{}_{1} & = \frac{EQ^{}_{\nu^{}_{1}} e^{\nu^{}_{1}} - EQ^{}_{\nu^{}_{2}} e^{\nu^{}_{2}}}{Y^{}_{1} (\nu^{}_{1} - \nu^{}_{2})}, \\
    Q^{}_{1} & = Y^{}_{1} \mu e^{-\mu}, \\
    Y^{}_{1} & = \frac{\mu}{\mu \nu^{}_{1} - \mu \nu^{}_{2} - \nu^{2}_{1} + \nu^{2}_{2}} \left( Q^{}_{\nu^{}_{1}} e^{\nu^{}_{1}} - Q^{}_{\nu^{}_{2}} e^{\nu^{}_{2}} - (\nu^{2}_{1} - \nu^{2}_{2}) \frac{Q^{}_{\mu} e^{\mu} - Y^{}_{0}}{\mu^{2}_{}} \right), \\
    Y^{}_{0} & = \max\left( \frac{\nu^{}_{1} Q^{}_{\nu^{}_{2}} e^{\nu^{}_{2}} - \nu^{}_{2} Q^{}_{\nu^{}_{1}} e^{\nu^{}_{1}}}{\nu^{}_{1} - \nu^{}_{2}}, 0 \right).
\end{align}

Furthermore, we recalculated the secure-key rate for the decoy-state protocol using corrected gain and error probabilities. Specifically, we replaced $EQ^{}_{\mu}$ and $Q^{}_{\mu}$ in \Cref{eq:E_decoy,eq:Q_decoy} with the corrected versions in \Cref{eq:E_corrected,eq:Q_corrected} and proceeded with the secure-key rate calculation as outlined.

For the proposed protocol, the secure-key rate calculation is based on the GLLP formulation, incorporating the signal yield $Q^{}_{\mu}$, the error-correction factor $f(\delta^{}_{\mu})$, and the protocol efficiency $q$:
\begin{align}
    K \ge q \cdot Q^{}_{\mu} \cdot \left\{ -f\left(\delta^{}_{\mu}\right) H^{}_{2}\left(\delta^{}_{\mu}\right) + \left(1 - \Delta\right) \left[ 1 - H^{}_{2}\left(\frac{\delta^{}_{\mu}}{1 - \Delta}\right) \right] \right\}.
\end{align}

The calculations for the proposed method assume $\Delta = 0$, indicating no eavesdropping intervention. This assumption is informed by our previously developed Bayesian framework, which allows for the inference of $\Delta$ from the observed data. Unlike the decoy-state protocol, which relies on an intensity-based approach, the proposed method leverages an observation-based approach, providing greater flexibility in assessing the presence of eavesdropping. For the purpose of this analysis, we present the secure-key rate assuming $\Delta = 0$, with the understanding that as the session length increases, the inference of $\Delta = 0$ becomes increasingly robust. Thus, the reported secure-key rate reflects the asymptotic behavior under this assumption.

The resulting secure-key rate performance is shown in \Cref{fig:key_rates}. The figure includes the secure-key rates for the proposed protocol across various intensities, alongside both the original and corrected versions of the decoy-state protocol for direct comparison.

\begin{figure}[h]
    \centering
    \includegraphics[width=\textwidth]{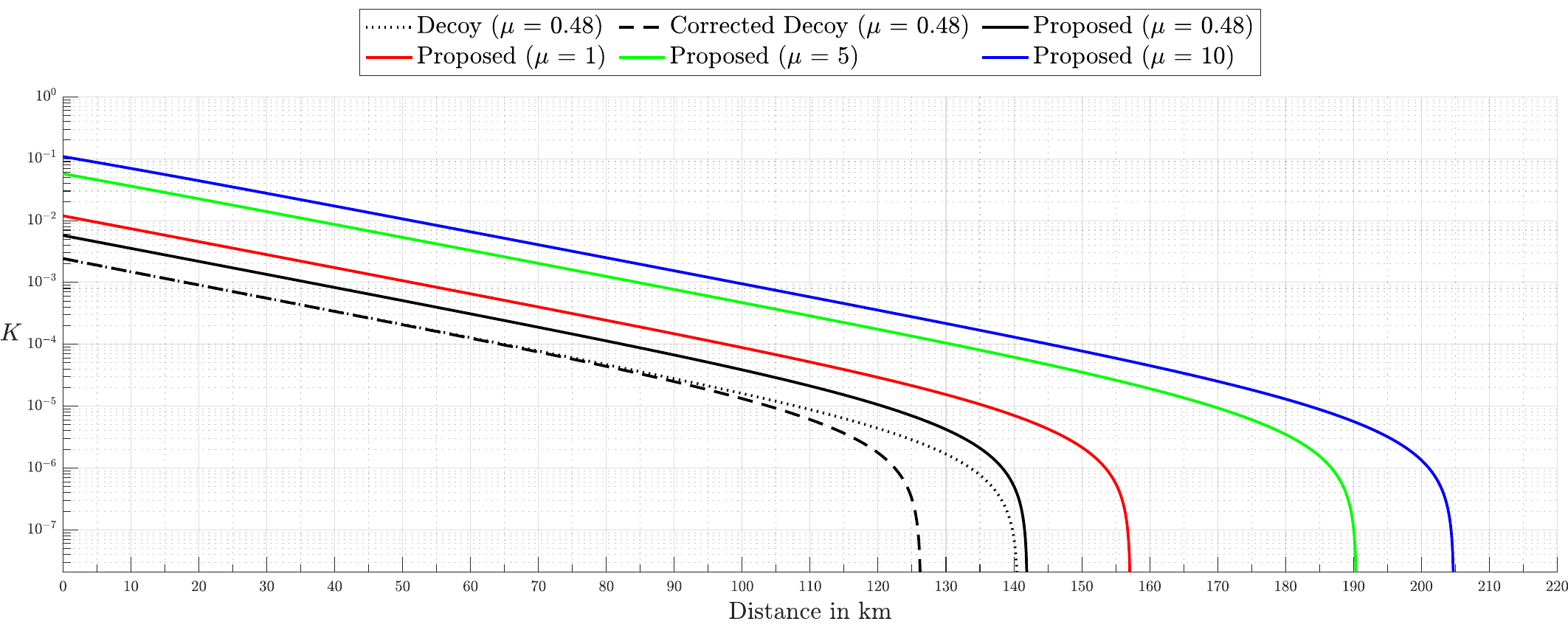}
    \caption{Comparison of secure-key rates as a function of distance for the proposed protocol and the decoy-state protocol. The decoy-state protocol is evaluated using the original error and gain rate calculations (dotted lines) and the corrected calculations (dashed lines). The proposed protocol is shown for intensities $\mu = 0.48$, $\mu = 1$, $\mu = 5$, and $\mu = 10$, illustrating higher key-rates and an extended operational range.}
    \label{fig:key_rates}
\end{figure}

\subsubsection{Analysis of Secure-Key Rates}

\Cref{fig:key_rates} compares the secure-key rates of our proposed protocol against the decoy-state protocol in this asymptotic regime. Under the same intensity setting and across its operational range, our approach achieves a median improvement of about 2.5 times the key generation rate and a 12.44\% increase in distance compared to the decoy-state protocol using the same intensity. Moreover, by increasing the pulse intensity to 10 photons per pulse (a cap chosen to ensure numerical precision in the incomplete gamma function computations), we obtain up to 50 times the key generation rate and a 62.2\% increase in operational distance (approximately 200 km).

In addition to increasing key rate, this approach grants access to operational regimes that were previously inaccessible. Moreover, operating at higher intensities enables the use of less sensitive and more cost-effective detectors.

\subsection{Testing the Bayesian Framework}

In this section, we explore the robustness and accuracy of the Bayesian framework in inferring key system parameters under different modeling assumptions. We conduct two main experiments: one assuming the i.i.d. model and another incorporating the Hidden Markov Model (HMM) to account for after-pulse effects. The aim is to validate the posterior distributions constructed using the Bayesian framework and assess the impact of these modeling choices on the inferred secure-key rates.

\subsubsection{Inference of Eve's Parameters Under the i.i.d. Model}

The first experiment focuses on inferring Eve’s parameters, $\theta^{}_{E} = \{ d^{}_{AE}, p^{}_{EB}, k, \Delta \}$, while assuming an i.i.d. model for photon detection events. All other parameters were kept constant as per \Cref{tab:main_parameters}. The number of intensity levels was set to eight (twice the number of parameters in $\theta^{}_{E}$), evenly spaced between $\lambda^{}_{\min}$ and a capped $\lambda^{}_{\max} = 10$, to ensure numerical precision due to limitations in the software's incomplete gamma function implementation.\footnote{All computations were performed using Matlab R2024b.} The capping avoids potential inaccuracies caused by the function returning identical values for different parameters when the intensity difference exceeds 10, ensuring reliable and precise calculations within the function’s resolution.

We utilized the Shrinking-Rank Slice Sampler (SRSS)~\cite{radford2010cass}, a variant of Covariance-Adaptive Slice Sampling~\cite{thompson2010covariance}. Although the derivation of the gradients is excluded from this paper for brevity, they were computed and are available in our code implementation~\cite{almosallam2024generalized}. The session length was $10^9$ pulses, with $10^5$ samples drawn and a burn-in of $10^3$ iterations. Instead of sampling multiple chains and performing convergence diagnostics, we relied on manual inspection and ensured that our chosen sample size exhibited convergence for this configuration. To minimize the risk of poor starting points, we began with the expected value of the prior, optimized to the maximum a posteriori (MAP) estimate using gradient-based optimization techniques, and used this MAP estimate as our initial sample. This approach allowed us to employ a relatively small burn-in size, as the initialization ensured sampling started from a high-probability region of the posterior.

The posterior distributions for Eve's parameters are shown in \Cref{fig:eve_iid_params}. The histograms, along with kernel density estimates (KDEs) normalized as probability density functions (PDFs), illustrate the inferred distributions. The empirical cumulative distribution functions (CDFs) from the KDEs were used to highlight the 99\% confidence intervals, with the true parameter values indicated for comparison.

Given that the secure-key rate $K$ is a one-to-one function of these parameters, we transformed the posterior distributions into distributions of secure-key rates (grouped by intensity levels). Unlike the decoy-state protocol, which primarily uses the signal intensity ($\mu$) for key generation, our approach ensures that no intensity is wasted. As long as $K > 0$ for a given intensity level, it contributes to the secure-key generation, enabling efficient utilization of all pulses and maximizing the extracted key rate. The transformed secure-key rate distributions are depicted in \Cref{fig:eve_iid_Ks}.

\begin{figure}[h!]
    \centering
    \includegraphics[width=\textwidth]{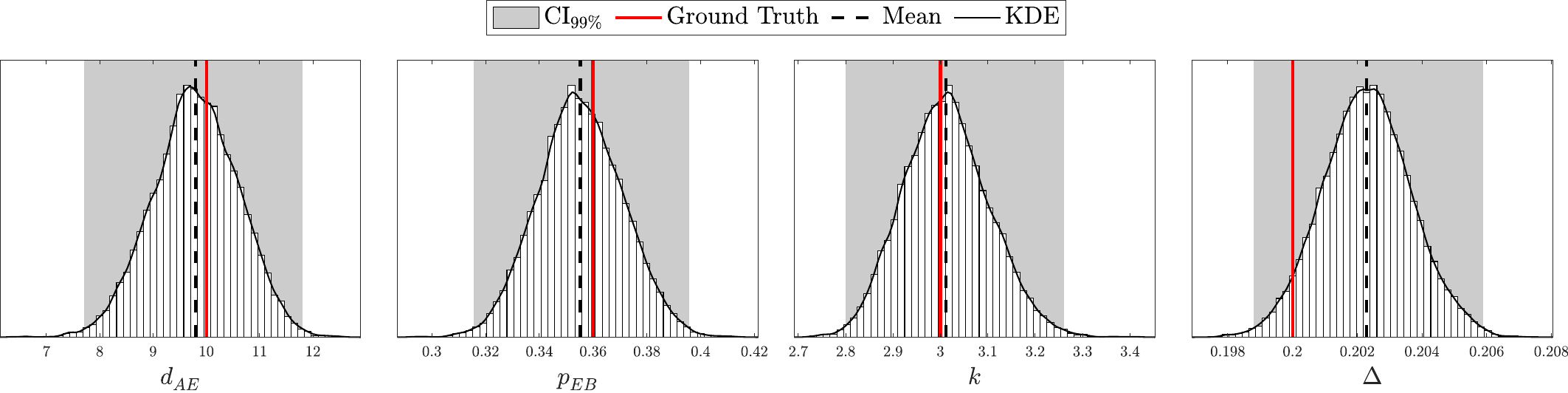}
    \caption{Posterior distributions of Eve’s parameters $\theta^{}_{E}$. Data was simulated using \Cref{alg:simulate_hmm}. The plots include sample histograms, ground truth values (red lines), empirical means (dashed lines), kernel density estimates (solid black lines), and 99\% confidence intervals (shaded regions).}
    \label{fig:eve_iid_params}
\end{figure}

\begin{figure}[h!]
    \centering
    \includegraphics[width=\textwidth]{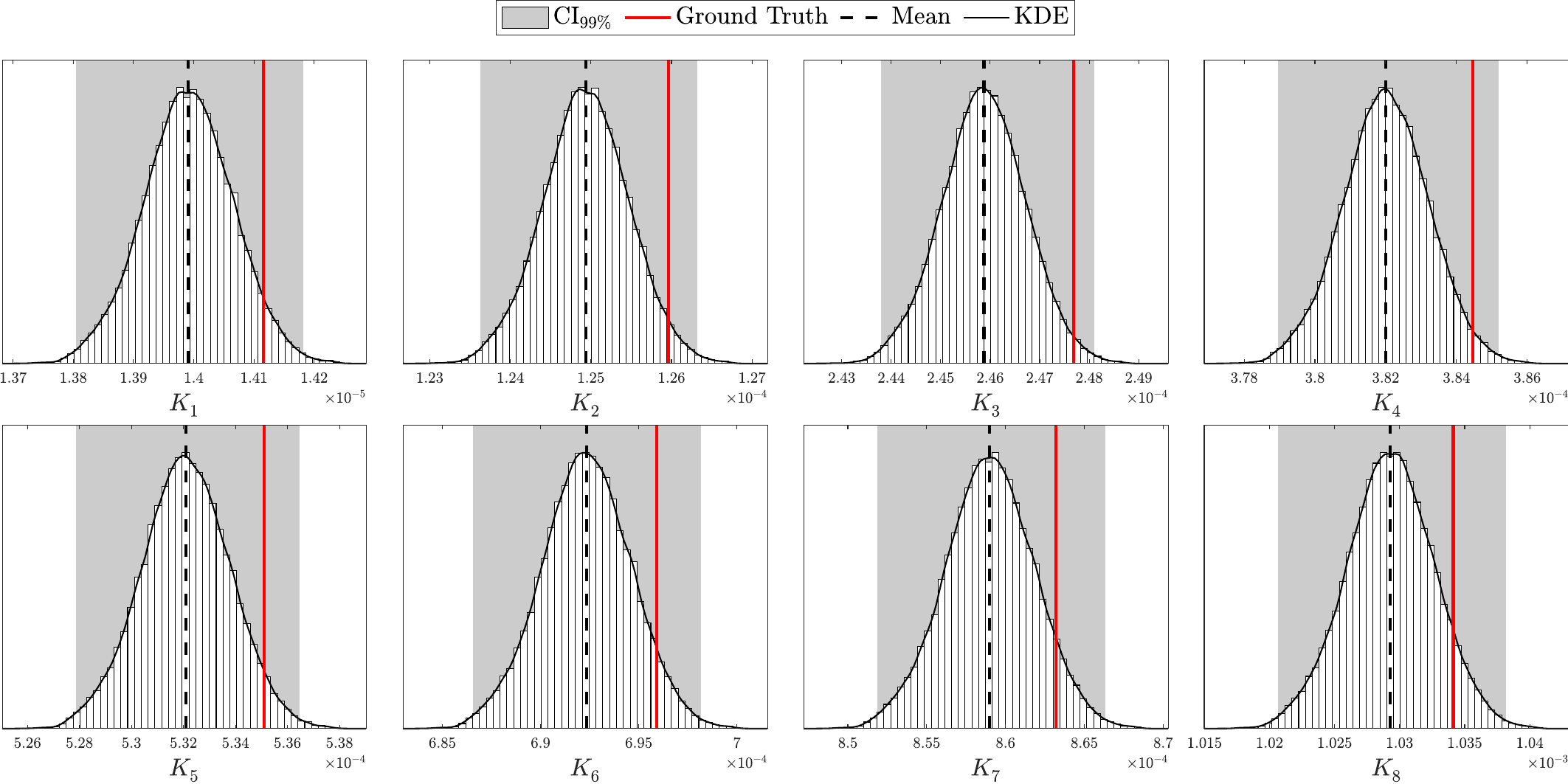}
    \caption{Posterior distributions of the secure-key rate for each intensity level under the i.i.d. model (data generated using \Cref{alg:simulate_iid}). The plots show sample histograms, ground truths (red lines), empirical means (dashed lines), kernel density estimates (solid black lines), and 99\% confidence intervals (shaded regions).}
    \label{fig:eve_iid_Ks}
\end{figure}

\subsubsection*{Analysis of the Inference Results}

The results in \Cref{fig:eve_iid_params} indicate a high degree of accuracy in inferring Eve’s parameters. The ground truth values lie well within the 99\% confidence intervals and close to the empirical mean, demonstrating that the Bayesian framework can reliably capture the underlying parameter distributions. Additionally, the kernel density estimates provide smooth approximations of the posteriors, validating the sampling method’s efficiency. The corresponding secure-key rate distributions in \Cref{fig:eve_iid_Ks} show that for all intensity levels, the key rates remain consistently above zero, confirming that the proposed method can effectively leverage the full range of intensities for key generation.

\subsubsection{Inference of Eve's Parameters Under the HMM Model}

The second experiment incorporates after-pulsing effects into the simulation and uses the Hidden Markov Model (HMM) to capture the resulting temporal correlations. In this setup, the after-pulse probability $p^{}_{a}$ was set to 0.1, with a 10\% variation applied between the two detectors, consistent with the parameters used in \Cref{tab:detector_specific_adjusted}. The experimental configuration and sampling approach remained the same as in the i.i.d. case.

\paragraph{Results for the HMM Model}

The posterior distributions for Eve’s parameters under the HMM model are shown in \Cref{fig:eve_hmm_params}. As in the i.i.d. case, the histograms represent sampled values, the solid black lines correspond to kernel density estimates (KDEs), and the shaded regions denote the 99\% confidence intervals. The true parameter values are indicated by red vertical lines for reference.

When incorporating after-pulsing effects, the results demonstrate a high level of accuracy in parameter inference, with the ground truth values well-contained within the confidence intervals. This highlights the robustness of the Bayesian framework under the more complex after-pulsing scenario. The inferred distributions show slightly broader confidence intervals compared to the i.i.d. case, reflecting the increased uncertainty introduced by temporal correlations. Despite this, the model effectively captures the true parameter values.

\begin{figure}[h!]
    \centering
    \includegraphics[width=\textwidth]{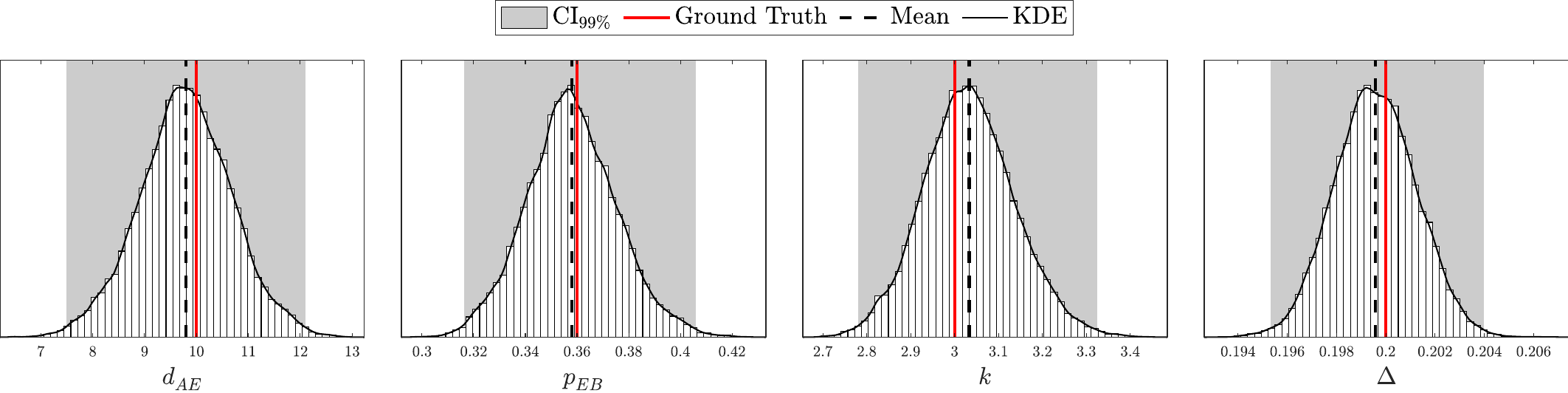}
    \caption{Posterior distributions of Eve’s parameters under the HMM. model, with after-pulsing present (data generated using \Cref{alg:simulate_hmm}). The plots show sample histograms, ground truths (red lines), empirical means (dashed lines), kernel density estimates (solid black lines), and 99\% confidence intervals (shaded regions).}
    \label{fig:eve_hmm_params}
\end{figure}

\paragraph{Comparison with the i.i.d. Model}

To assess the impact of neglecting after-pulsing, we performed the same experiment under the incorrect i.i.d. assumption while using data generated with \Cref{alg:simulate_hmm}, which incorporates after-pulsing effects. The posterior distributions inferred using the i.i.d. model are shown in \Cref{fig:eve_hmm_params_with_iid}. The results highlight severe discrepancies between the inferred parameters and their true values, accompanied by unwarranted confidence in these erroneous estimates. Most importantly, $\Delta$—a critical parameter for secure-key rate calculations—is significantly underestimated, with a mean value of 0.0149 compared to the ground truth of 0.2. This underestimation of $\Delta$ is particularly costly, as it leads to overestimated secure-key generation rates, thereby compromising the protocol's security. The high confidence intervals around these incorrect predictions further underscore the failure of the i.i.d. model to capture the temporal correlations introduced by after-pulsing.

\begin{figure}[h!]
    \centering
    \includegraphics[width=\textwidth]{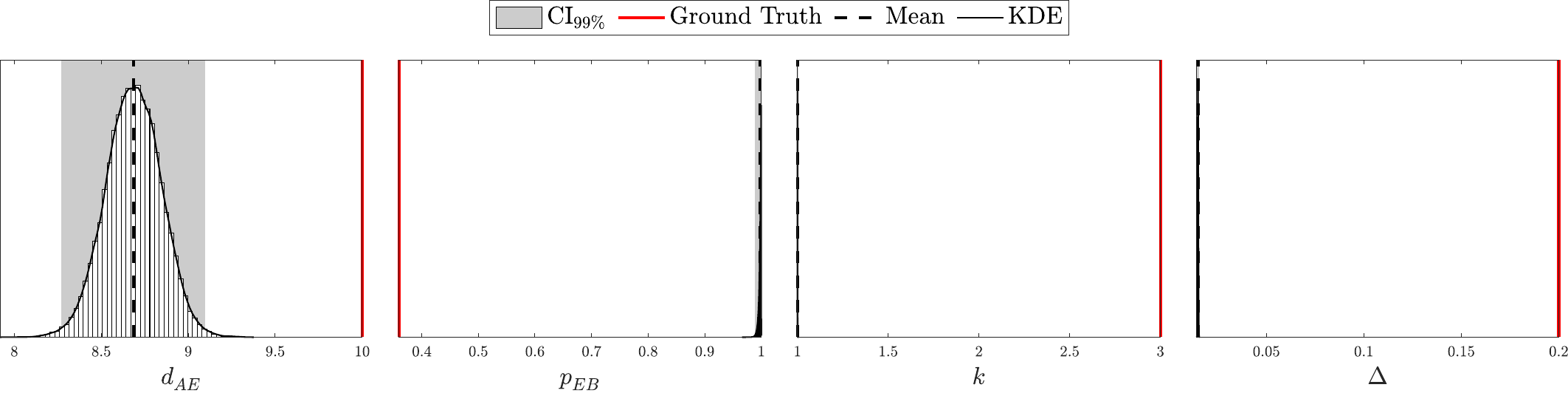}
    \caption{Posterior distributions of Eve’s parameters under the i.i.d. model, with after-pulsing present (data generated using \Cref{alg:simulate_hmm}). The plots show sample histograms, ground truths (red lines), empirical means (dashed lines), kernel density estimates (solid black lines), and 99\% confidence intervals (shaded regions).}
    \label{fig:eve_hmm_params_with_iid}
\end{figure}

\paragraph{Secure-Key Rate Analysis}

We also transformed the inferred posterior distributions into secure-key rate distributions for both the HMM model and the mismatched i.i.d. model. The results are depicted in \Cref{fig:eve_hmm_Ks} and \Cref{fig:eve_hmm_Ks_with_iid}, respectively. Under the HMM model, the secure-key rates for the first four intensity levels are effectively zero. This behavior differs from the case without after-pulsing, where these intensities still contributed to the key rate. After-pulsing increases the error rate, which, in turn, suppresses secure-key generation for these lower intensities. However, for higher intensity levels, the HMM model effectively leverages as much information as possible, resulting in significantly larger secure-key rates with tight confidence intervals closely centered around the ground truth values.

In contrast, the i.i.d. model applied to data generated with after-pulsing results in highly inaccurate secure-key rate estimates. The key rates were consistently and significantly overestimated. Moreover, the model exhibited unwarranted confidence in these highly inaccurate results.

\begin{figure}[h!]
    \centering
    \includegraphics[width=\textwidth]{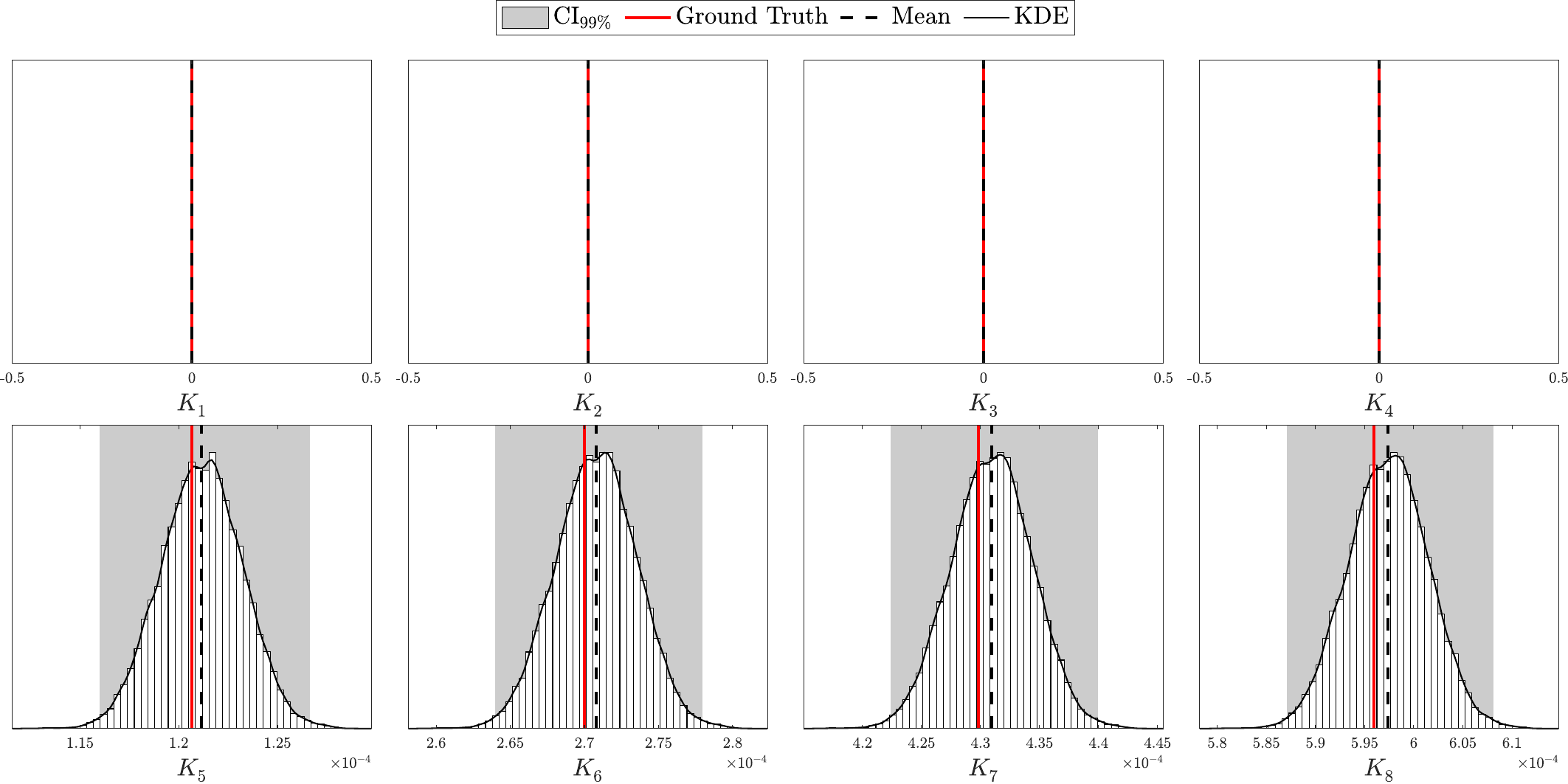}
    \caption{Posterior distributions of the secure-key rate for each intensity level under the HMM model, with after-pulsing present (data generated using \Cref{alg:simulate_hmm}). The plots show sample histograms, ground truths (red lines), empirical means (dashed lines), kernel density estimates (solid black lines), and 99\% confidence intervals (shaded regions).}
    \label{fig:eve_hmm_Ks}
\end{figure}

\begin{figure}[h!]
    \centering
    \includegraphics[width=\textwidth]{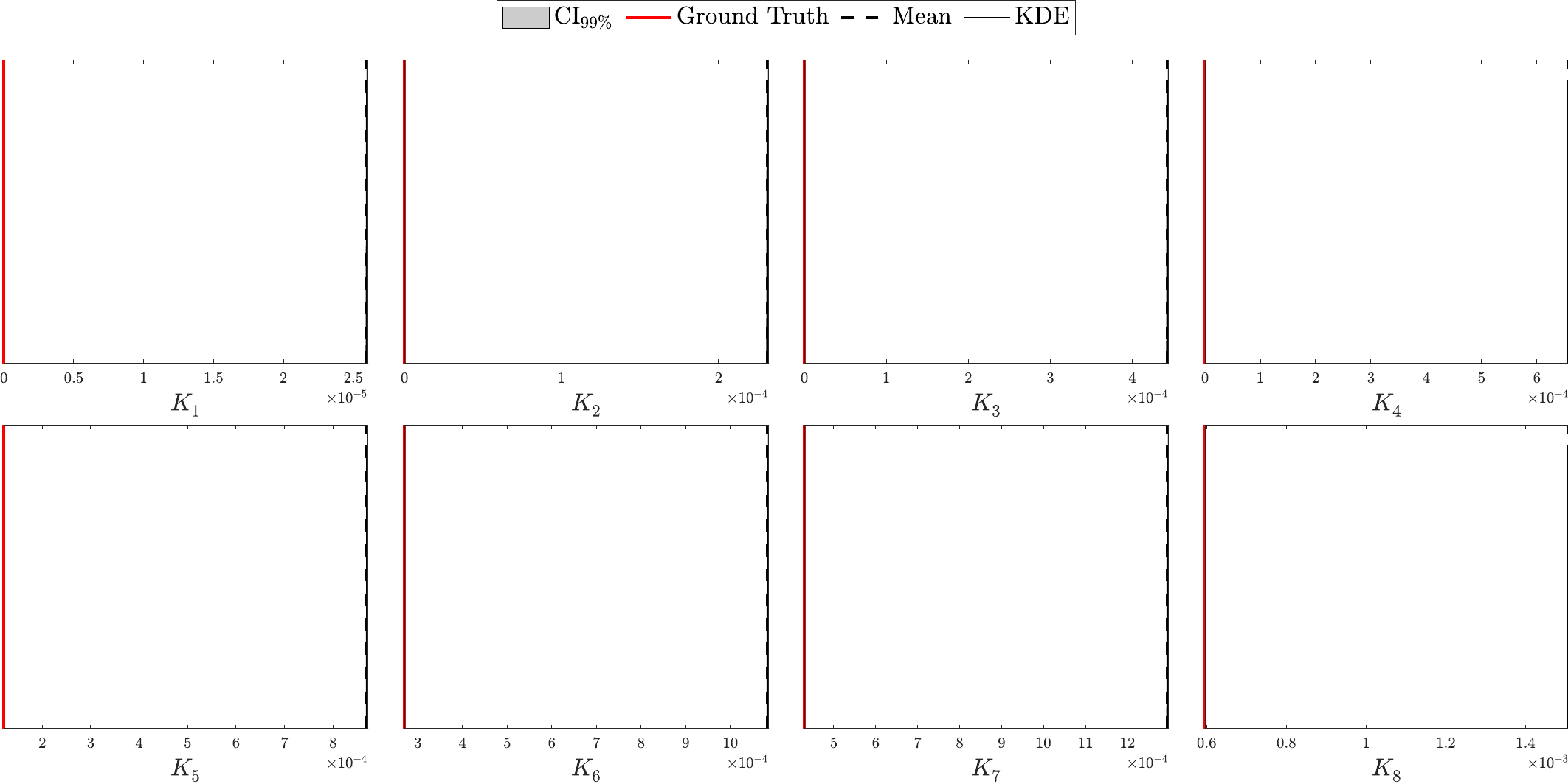}
    \caption{Posterior distributions of the secure-key rate for each intensity level under the i.i.d. model, with after-pulsing present (data generated using \Cref{alg:simulate_hmm}). The plots show sample histograms, ground truths (red lines), empirical means (dashed lines), kernel density estimates (solid black lines), and 99\% confidence intervals (shaded regions).}
    \label{fig:eve_hmm_Ks_with_iid}
\end{figure}

\subsubsection{Testing the Fully Bayesian Model}

Building on the previous evaluations, we next tested the proposed approach under a fully Bayesian framework, treating all system parameters as random variables with appropriate priors. The setup follows the same structure as the HMM experiment, but incorporates variability in all parameter values.

\paragraph{Experimental Setup}

Instead of using fixed parameter values as in \Cref{tab:main_parameters}, we introduced randomness to the system parameters by sampling them from distributions with 5\% standard deviation around their reported means. For example, if the reported value of $\alpha$ is 0.21, we sampled from a gamma distribution with parameters chosen to match the reported mean of $0.21$ and standard deviation of $0.042$. For parameters that are bounded, such as probabilities, we sampled from beta distributions with similarly matching means and standard deviations. The sampled values were assumed fixed for the duration of the session. Priors were chosen to reflect realistic uncertainties: beta distributions for bounded parameters and gamma distributions for semi-bounded ones. The priors were designed to have the same means as the reported parameters but with standard deviations doubled to 10\%, ensuring a reasonable level of variability. The simulation was run for $10^9$ pulses, with the sampling and inference conducted using the same Shrinking-Rank Slice Sampler (SRSS) methodology as in previous experiments. The inferred posterior distributions for the fully Bayesian model are shown in \Cref{fig:full_params}, while the secure-key rates derived from these posteriors are presented in \Cref{fig:full_Ks}.

\begin{figure}[h!]
    \centering
    \includegraphics[width=\textwidth]{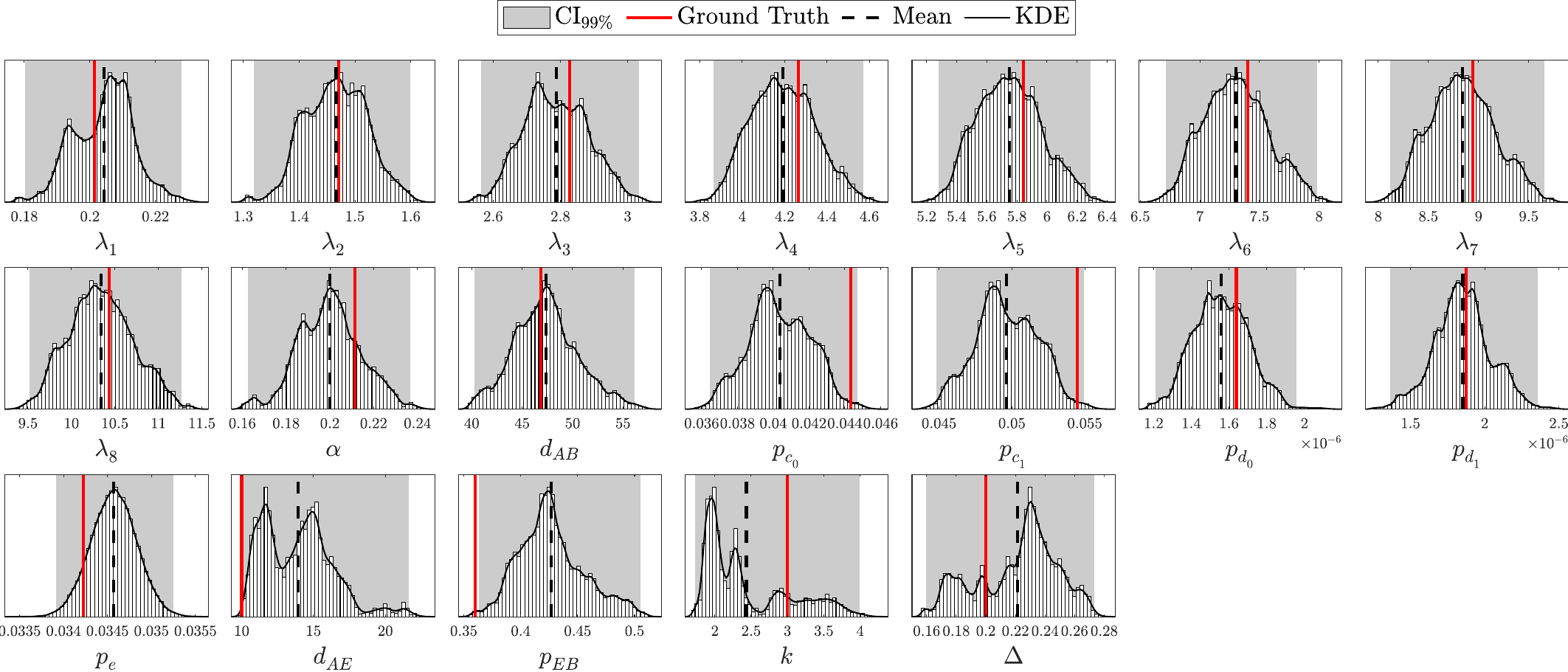}
    \caption{Posterior distributions of all parameters, treating all parameters as random variables, under the i.i.d. model (data generated using \Cref{alg:simulate_iid}). The plots show sample histograms, ground truths (red lines), empirical means (dashed lines), kernel density estimates (solid black lines), and 99\% confidence intervals (shaded regions).}
    \label{fig:full_params}
\end{figure}

\begin{figure}[h!]
    \centering
    \includegraphics[width=\textwidth]{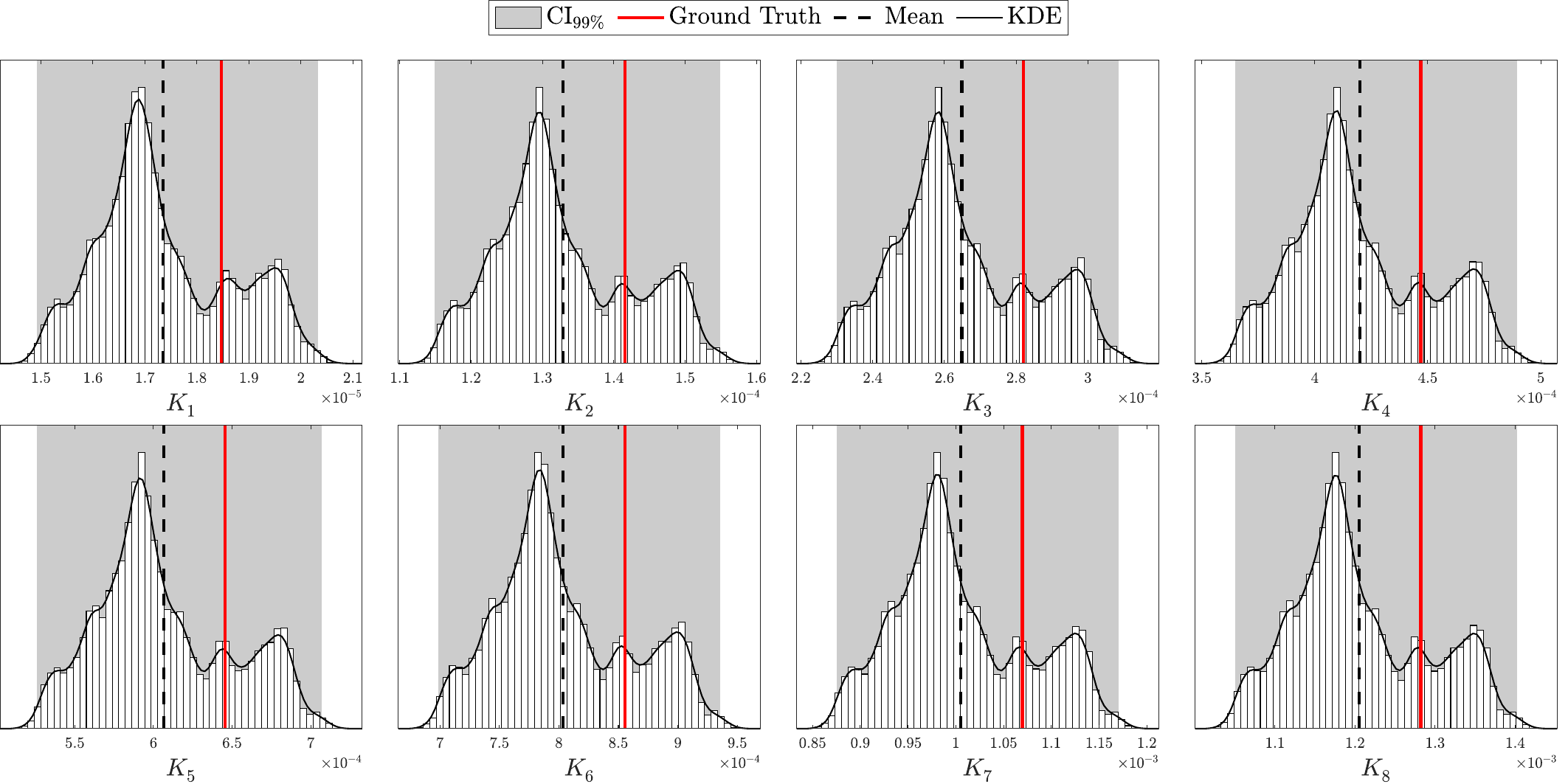}
    \caption{Posterior distributions of secure the secure-key rates for each intensity level, treating all parameters as random variables, under the i.i.d. model (data generated using \Cref{alg:simulate_iid}). The plots show sample histograms, ground truths (red lines), empirical means (dashed lines), kernel density estimates (solid black lines), and 99\% confidence intervals (shaded regions).}
    \label{fig:full_Ks}
\end{figure}

\paragraph{Testing the Fixed-Parameter Assumption}

To investigate the impact of incorrectly assuming fixed parameters, we repeated the analysis using data generated from the same randomly sampled parameters but assumed that these parameters were fixed at their expected values during inference. This assumption mirrors the earlier experiments with mismatched models, such as assuming independence when after-pulsing is present. The results for parameter inference under this fixed assumption are depicted in \Cref{fig:full_params_with_fixed}, and the corresponding secure-key rate distributions are shown in \Cref{fig:full_Ks_with_fixed}.

\begin{figure}[h!]
    \centering
    \includegraphics[width=\textwidth]{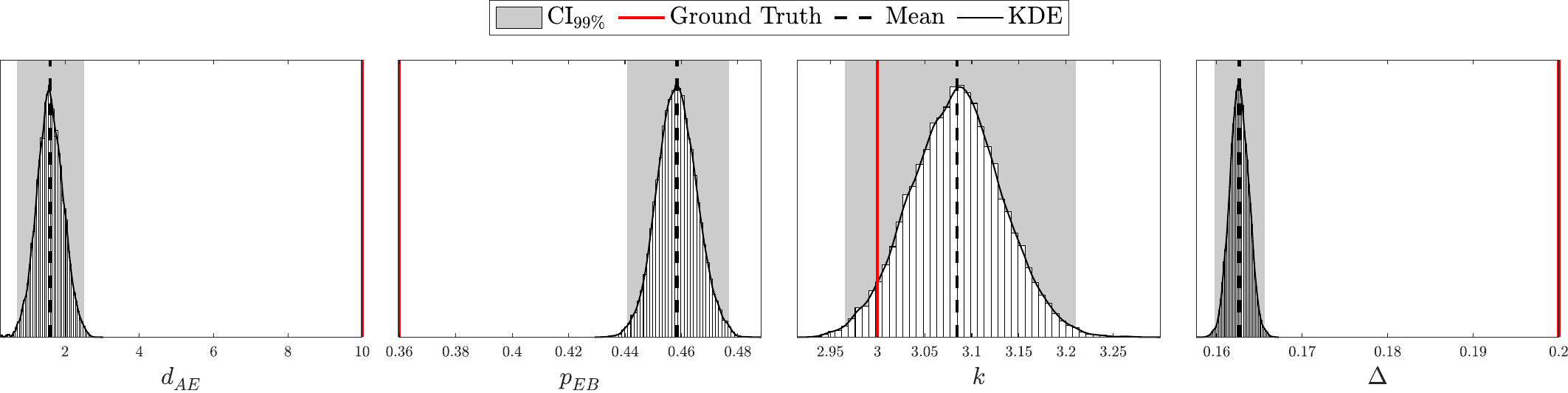}
    \caption{Posterior distributions of Eve's parameters under the i.i.d. model, assuming only the parameters in $\theta^{}_{E}$ are treated as random variables, while the others remain fixed (data generated using \Cref{alg:simulate_iid}). The plots include sample histograms, ground truth values (red lines), empirical means (dashed lines), kernel density estimates (solid black lines), and 99\% confidence intervals (shaded regions).}
    \label{fig:full_params_with_fixed}
\end{figure}

\begin{figure}[h!]
    \centering
    \includegraphics[width=\textwidth]{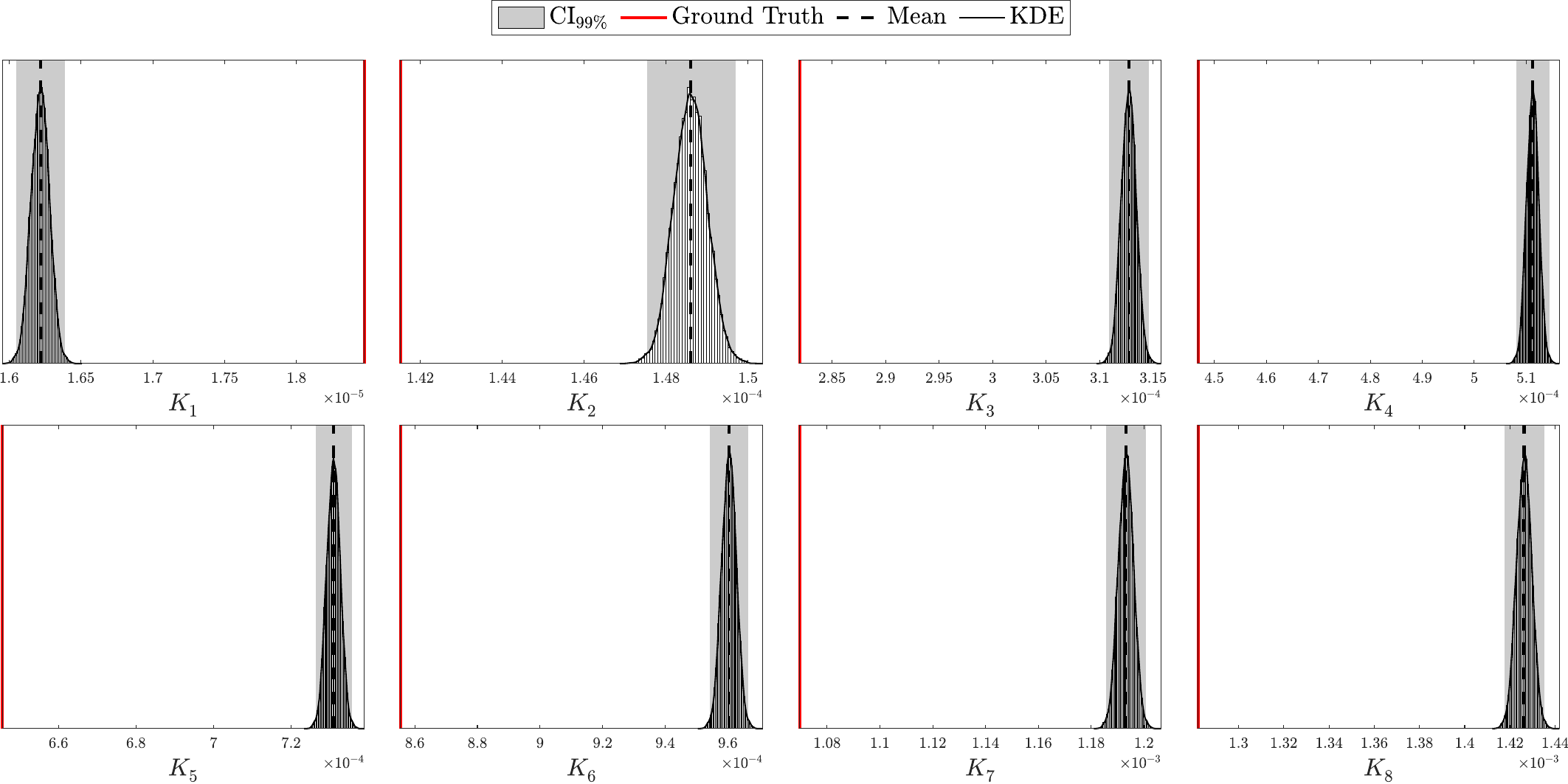}
    \caption{Posterior distributions of the secure-key rate for each intensity level under the i.i.d. model, assuming only the parameters in $\theta^{}_{E}$ are treated as random variables, while the others remain fixed (data generated using \Cref{alg:simulate_iid}). The plots include sample histograms, ground truth values (red lines), empirical means (dashed lines), kernel density estimates (solid black lines), and 99\% confidence intervals (shaded regions).}
    \label{fig:full_Ks_with_fixed}
\end{figure}

\paragraph{Analysis of Results}

The results in \Cref{fig:full_params} and \Cref{fig:full_Ks} demonstrate the effectiveness of the fully Bayesian model in accurately capturing the true distributions of both system parameters and secure-key rates. The posterior distributions align well with the true values, and the confidence intervals are tight, indicating high reliability in the inferred values. Importantly, the secure-key rate distributions remain above zero across all intensity levels, reflecting the model's capacity to utilize the full range of available data for key generation.

In contrast, \Cref{fig:full_params_with_fixed} and \Cref{fig:full_Ks_with_fixed} highlight the substantial inaccuracies that arise when assuming fixed parameters during inference. The posterior distributions show significant deviations from the true parameter values, and the secure-key rates are mostly severely overestimated. This underscores the risks of neglecting parameter variability in QKD systems, as such oversimplifications can lead to erroneous security assessments.

These findings highlight the robustness of the fully Bayesian framework in handling parameter uncertainties, providing a reliable foundation for analyzing real-world QKD systems where parameter variability is inevitable.

\section{Time Complexity Analysis}
\label{sec:time_complexity}

In this section, we analyze the computational complexity of evaluating the posterior for both the i.i.d. and Hidden Markov Model (HMM) frameworks. The vector of counts $C$ is assumed to be precomputed and available. Evaluating the posterior involves computing the likelihood function, prior distribution, and Jacobian determinant, as well as their derivatives with respect to the parameters.

\subsection{Computational Complexity for the i.i.d. Model}

The likelihood function for the i.i.d. model is given by:

\begin{align}
	\mathcal{L}(C \mid N, \theta) = \operatorname{Multinomial}\left(C\mid N, \mathbf{P}(\theta) \right),
\end{align}
where $\mathbf{P}(\theta)$ is the vector of detection probabilities computed based on the model parameters $\theta$ as in \Cref{eq:P_iid}.

\subsubsection{Function Evaluation}
The computational cost for evaluating $\mathbf{P}(\theta)$ involves calculating detection probabilities for all possible detection outcomes (00, 01, 10, 11) and matching/non-matching cases. The length of $\mathbf{P}(\theta)$ is therefore $8 N^{}_{\lambda}$. Thus, the time complexity for computing the likelihood in the i.i.d. case is:

\begin{align}
	T^{}_{\mathcal{L}} = \mathcal{O}(N^{}_{\lambda}).
\end{align}

Evaluating the prior distribution $\mathcal{P}(\theta \mid \theta^{}_{P})$ and the Jacobian determinant involves computations for each of the $N^{}_{\theta}$ parameters:

\begin{align}
	T^{}_{\mathcal{P}} = \mathcal{O}(N^{}_{\theta}), \quad T^{}_{\mathcal{J}} = \mathcal{O}(N^{}_{\theta}).
\end{align}

\subsubsection{Derivative Evaluations}
Computing the derivatives of the likelihood function with respect to each parameter involves calculating the derivatives of $\mathbf{P}(\theta)$ for each detection probability. This results in $N^{}_{\theta} \times 8 N^{}_{\lambda}$ derivative computations. The time complexity for computing the derivatives of the likelihood is therefore:

\begin{align}
	T^{}_{\partial\mathcal{L}} = \mathcal{O}(N^{}_{\theta} N^{}_{\lambda}).
\end{align}

The cost for evaluating the derivative of the prior and the Jacobian is linear in $N^{}_{\theta}$:
\begin{align}
	T^{}_{\partial \mathcal{P}} = \mathcal{O}(N^{}_{\theta}), \quad T^{}_{\partial \mathcal{J}} = \mathcal{O}(N^{}_{\theta}).
\end{align}
\subsubsection{Total Computational Cost}

Let $T^{}_{f}$ and $T^{}_{\partial f}$ represent the time complexity for evaluating the posterior and its derivatives, respectively. For the i.i.d. model, these complexities are:

\begin{align}
	T^{}_{f} & = T^{}_{\mathcal{L}} + T^{}_{\mathcal{P}} + T^{}_{\mathcal{J}} = \mathcal{O}(N^{}_{\lambda} + 2N^{}_{\theta}),\\
T^{}_{\partial f} & = T^{}_{\partial\mathcal{L}} + T^{}_{\partial\mathcal{P}} + T^{}_{\partial\mathcal{J}} = \mathcal{O}(N^{}_{\theta} N^{}_{\lambda}).
\end{align}

\subsection{Computational Complexity for the HMM Model}

In the HMM model, detection events are dependent due to after-pulsing effects. The likelihood function involves computing the stationary distribution of a Markov chain represented by the transition matrix $\mathbf{T}(\theta)$, with number of states $N^{}_{T} = 2|S|N^{}_{\lambda}$, then projecting it using the emission matrix $\mathbf{E}$. Since the matrix $\mathbf{T}(\theta)$ is constructed from the base transition matrix $\mathbf{T}(p, q)$ in \Cref{eq:T_p_q}, the density of the full matrix $\mathbf{T}(\theta)$ will be the same. Thus, from \Cref{eq:T_p_q} and \Cref{eq:emission}, we have the following density factors:

\begin{align}
	\Omega^{}_{\mathbf{T}} = \frac{49}{|S|^2} \approx 60.5\%, \quad \Omega^{}_{\mathbf{E}} = \frac{9}{|S||O|} = 25\%.
\end{align}

\subsubsection{Function Evaluation}
The likelihood function involves computing $\hat{\mathbf{P}}(\theta)$ in \Cref{eq:P_hmm}, which has three main steps, (1) constructing the transition matrix $\mathbf{T}(\theta)$, (2) computing the stationary distribution $\bar{\mathbf{v}}(\theta)$, and (3) projecting the stationary distribution onto the observable probabilities using the emission matrix $\mathbf{E}$.

\paragraph{Transition Matrix}
Constructing $\mathbf{T}(\theta)$ has a complexity proportional to the number of non-zero elements:

\begin{align}
	T^{}_{\mathbf{T}(\theta)} = \mathcal{O}(\Omega^{}_{\mathbf{T}} N^{2}_{T}) = \mathcal{O}(N^{2}_{\lambda}).
\end{align}

\paragraph{Stationary Distribution}
The stationary distribution is computed using an iterative eigenvalue solver, requiring $N^{}_{I}$ iterations. The complexity is therefore:

\begin{align}
	T^{}_{\bar{\mathbf{v}}(\theta)} = \mathcal{O}(N^{}_{I} \Omega^{}_{\mathbf{T}} N^{2}_{T}) = \mathcal{O}(N^{}_{I} N^{2}_{\lambda}).
\end{align}

\paragraph{Emission Projection}
Projecting $\bar{\mathbf{v}}(\theta)$ has a complexity:

\begin{align}
	T^{}_{\mathbf{E}^{\top} \bar{\mathbf{V}}(\theta)} = \mathcal{O}(\Omega^{}_{\mathbf{E}}|O||S|N^{}_{\lambda}) = \mathcal{O}(N^{}_{\lambda}).
\end{align}

The total time complexity for evaluating the likelihood is:

\begin{align}
	T^{}_{\hat{\mathcal{L}}} = T^{}_{\mathbf{T}(\theta)} + T^{}_{\bar{\mathbf{v}}(\theta)} + T^{}_{\mathbf{E}^{\top} \bar{\mathbf{V}}(\theta)} = O(N^{}_{I} N^{2}_{\lambda}).
\end{align}

\subsubsection{Derivative Evaluations}
Computing the derivatives of the likelihood involves taking the derivatives of the above steps with respect to each parameter, which have the same time complexities

\begin{align}
T^{}_{\partial \mathbf{T}(\theta) / \partial \theta^{}_{i}} & = T^{}_{\mathbf{T}(\theta)},\\
T^{}_{\partial \bar{\mathbf{v}}(\theta) / \partial \theta^{}_{i}} & = T^{}_{\bar{\mathbf{v}}(\theta)},\\
T^{}_{\mathbf{E}^{\top} \partial \bar{\mathbf{V}}(\theta) / \partial \theta^{}_{i}} & = T^{}_{\mathbf{E}^{\top} \bar{\mathbf{V}}(\theta)}.	
\end{align}

The total derivative complexity is:

\begin{align}
	T^{}_{\partial\hat{\mathcal{L}}} = N^{}_{\theta} \left( T^{}_{\partial \mathbf{T}(\theta) / \partial \theta^{}_{i}} + T^{}_{\partial \bar{\mathbf{v}}(\theta) / \partial \theta^{}_{i}} + T^{}_{\mathbf{E}^{\top} \partial \bar{\mathbf{V}}(\theta) / \partial \theta^{}_{i}} \right) = \mathcal{O}(N^{}_{\theta} N^{}_{I} N^{2}_{\lambda}).
\end{align}

\subsubsection{Total Computational Cost}
Let $T^{}_{\hat{f}}$ and $T^{}_{\partial \hat{f}}$ represent the time complexity for evaluating the HMM posterior and its derivatives, respectively. For the i.i.d. model, these complexities are:

\begin{align}
	T^{}_{\hat{f}} & = T^{}_{\hat{\mathcal{L}}} + T^{}_{\mathcal{P}} + T^{}_{\mathcal{J}} = \mathcal{O}(N^{}_{I} N^{2}_{\lambda} + 2N^{}_{\theta}),
\end{align}

\begin{align}
	T^{}_{\partial \hat{f}} & = T^{}_{\partial\hat{\mathcal{L}}} + T^{}_{\partial\mathcal{P}} + T^{}_{\partial\mathcal{J}}  = \mathcal{O}(N^{}_{\theta} N^{}_{I} N^{2}_{\lambda}).
\end{align}

\subsection{Analysis of Time Complexities}
\label{sec:complexity_discussion}

We summarize the computational complexities for both the i.i.d. and HMM models in Table~\ref{tab:complexity} (assuming $N^{}_{\lambda}> 2 N^{}_{\theta}$ as we do in practice). The i.i.d. model exhibits linear scaling with the number of intensity levels $N^{}_{\lambda}$ and the number of parameters $N^{}_{\theta}$ for function evaluation, and scales linearly with their product for derivative computation. In contrast, the HMM model shows a quadratic dependence on $N^{}_{T}$, the number of states in the Markov chain, which itself is proportional to $N^{}_{\lambda}$.

\begin{table}[h!]
    \centering
    \renewcommand{\arraystretch}{1.5}
    \begin{tabular}{l c c}
        \hline
        \textbf{Model} & \textbf{Function Evaluation} & \textbf{Derivative Evaluation} \\ \hline
        i.i.d. & $\mathcal{O}(N^{}_{\lambda})$ & $\mathcal{O}(N^{}_{\theta} N^{}_{\lambda})$ \\ \hline
        HMM & $\mathcal{O}(N^{}_{I}  N^{2}_{\lambda})$ & $\mathcal{O}(N^{}_{\theta} N^{}_{I}  N^{2}_{\lambda})$ \\ \hline
    \end{tabular}
    \caption{Summary of computational complexities for evaluating the posterior and its derivatives under the i.i.d. and HMM models (assuming $N{\lambda}>2N^{}_{\theta}$).}
    \label{tab:complexity}
\end{table}

Analyzing these results, we observe that the HMM model incurs significantly higher computational costs compared to the i.i.d. model. The quadratic dependence on $N^{}_{\lambda}$ in the HMM model arises from the need to compute the stationary distribution its derivatives. Despite the higher computational complexity, the HMM model captures the dependencies introduced by after-pulsing effects and other temporal correlations in the detection events, which are neglected in the i.i.d. framework.

While the HMM model is computationally more expensive than the i.i.d. model, it enables accurate analysis of QKD systems with detector dependencies and temporal correlations, such as after-pulsing. This approach allows us to assess whether after-pulsing significantly impacts error rates and evaluate trade-offs, such as increasing the detector dead time to lower the key rate or accepting additional post-processing time to maintain a higher key rate. Such flexibility ensures tailored solutions based on specific application requirements.

In addition, the added computational effort unlocks the ability to analyze scenarios that traditional methods, like the decoy-state protocol, cannot handle. This framework provides a practical means to push the limits of secure key distribution, such as operating at higher intensities for longer distances. By overcoming fundamental limitations of existing protocols, this approach prioritizes expanding the boundaries of what is possible in secure key distribution. Once these new limits are achieved, further optimizations can focus on improving efficiency, ensuring that both capability and performance evolve together.

\section{Conclusion}
\label{sec:conclusion}

In this work, we have presented a comprehensive probabilistic modeling framework for quantum key distribution (QKD) protocols, enabling accurate assessment of QKD performance under various conditions. By incorporating detailed physical processes such as asymmetric detectors, after-pulsing effects, and variable parameters, our approach allows for a deeper understanding of the factors influencing QKD efficacy.

The proposed framework not only provides precise estimations but also offers the flexibility to simplify certain assumptions when appropriate. For instance, by considering symmetric detectors or neglecting after-pulsing effects, one can accurately assess whether specific parameters can be considered negligible in certain scenarios. This adaptability is crucial for tailoring the model to different practical situations and optimizing QKD implementations accordingly.

Importantly, our analysis has revealed that some of the common assumptions made in decoy-state protocols may not hold in all cases. Specifically, we demonstrated that under certain conditions—such as long-distance communication or equipment imperfections—these assumptions can lead to inaccuracies in the estimated secure key rate and overall system performance. By highlighting these discrepancies, our work underscores the necessity of using more detailed probabilistic models to ensure the reliability and security of QKD systems.

While our work is grounded in theoretical modeling, it is not purely theoretical. In particular scenarios—especially those involving moderate distances and well-calibrated equipment—the proposed methods are both applicable and practical. We acknowledge that in some cases, such as very long distances or with poor and noisy equipment, the approach may require large sample sizes or extended processing times, potentially limiting immediate applicability. Nonetheless, by making these advanced modeling techniques available, we have paved the way for exploring the feasibility of long-distance QKD under more challenging conditions.

Our accurate simulations demonstrate the practical potential of the proposed methods, showcasing that long-distance QKD is possible in some scenarios. The theoretical tools developed here enable a more probabilistic treatment of QKD, as we have successfully constructed the underlying distribution governing the system's behavior. This advancement opens new avenues for applying probabilistic techniques to QKD analysis and design. Furthermore, while more comprehensive attack models have been proposed (e.g., \cite{lim2013concise}), our work specifically targets the generalized PNS attack model. This focus allows for a rigorous evaluation of system performance under a well-defined threat model and lays the groundwork for extending inference-based approaches to even broader attack scenarios in future work.

For future work, there is significant scope in optimizing protocol parameters. One promising direction is to determine the optimal number of intensity levels and their specific values, as well as establishing the minimum session length required to guarantee a desired variance in key parameters like the proportion of intercepted pulses by Eve, $\Delta$. Additionally, implementing the protocol in a physical system would further validate our theoretical models and simulations, bridging the gap between theory and real-world applications.

In conclusion, this work represents a significant contribution to the modeling and analysis of QKD systems. By providing a detailed probabilistic framework, we have enabled more accurate assessments of QKD performance and highlighted the limitations of certain simplifying assumptions commonly used in decoy-state protocols. This foundation lays the groundwork for future developments that will further enhance the practicality and reliability of QKD technologies, making long-distance quantum communication more attainable.

\FloatBarrier
\bibliographystyle{naturemag}
\bibliography{full}

\end{document}